\documentclass{article}

\usepackage[final]{neurips_2025}

\usepackage{graphicx}
\usepackage{booktabs}

\usepackage{lipsum}

\usepackage{wrapfig} 
\usepackage{booktabs}
\usepackage{amsthm}
\usepackage{booktabs}
\usepackage{subcaption}

\usepackage[table]{xcolor}

\usepackage[utf8]{inputenc} % allow utf-8 input
\usepackage[T1]{fontenc}    % use 8-bit T1 fonts
\usepackage{hyperref}       % hyperlinks
\usepackage{url}            % simple URL typesetting
\usepackage{booktabs}       % professional-quality tables
\usepackage{amsfonts}       % blackboard math symbols
\usepackage{nicefrac}       % compact symbols for 1/2, etc.
\usepackage{microtype}      % microtypography
\usepackage{xcolor}         % colors
\usepackage{amsmath}
\usepackage{graphicx}
\newtheorem{definition}{Definition}
\newtheorem{theorem}{Theorem}
\newtheorem{lemma}{Lemma}
\newtheorem{proposition}{Proposition}

\newtheorem*{theorem*}{Theorem}

\usepackage[ruled,vlined]{algorithm2e}
\usepackage{caption}
\usepackage{subcaption}

\title{Information Theoretic Learning for Diffusion Models with Warm Start}

\author{
\begin{tabular}{ccc}
\textbf{Yirong Shen} & \textbf{Lu Gan} & \textbf{Cong Ling} \\
\normalfont Imperial College London & \normalfont Brunel University of London & \normalfont Imperial College London \\
\texttt{ys6922@ic.ac.uk} & \texttt{lu.gan@brunel.ac.uk} & \texttt{c.ling@imperial.ac.uk}
\end{tabular}
}

\begin{document}

\maketitle

\begin{abstract}
  Generative models that maximize model likelihood have gained traction in many practical settings. Among them, perturbation-based approaches underpin many state-of-the-art likelihood estimation models, yet they often face slow convergence and limited theoretical understanding. In this paper, we derive a tighter likelihood bound for noise-driven models to improve both the accuracy and efficiency of maximum likelihood learning. Our key insight extends the classical Kullback–Leibler (KL) divergence–Fisher information relationship to arbitrary noise perturbations, going beyond the Gaussian assumption and enabling structured noise distributions. This formulation allows flexible use of randomized noise distributions that naturally account for sensor artifacts, quantization effects, and data distribution smoothing, while remaining compatible with standard diffusion training. Treating the diffusion process as a Gaussian channel, we further express the mismatched entropy between data and model, showing that the proposed objective upper-bounds the negative log-likelihood (NLL). In experiments, our models achieve competitive NLL on CIFAR-10 and state-of-the-art results on ImageNet across multiple resolutions, all without data augmentation, and the framework extends naturally to discrete data.

\end{abstract}

\section{Introduction}

Likelihood serves as a fundamental metric for evaluating density estimation and generative models. A tight negative log-likelihood (NLL) bound not only indicates a model's capacity to capture the fine-grained structure of the data distribution but also facilitates a range of downstream applications, including data compression \cite{helminger2021lossy,10.5555/3454287.3454635,ho2021anfic,hoogeboom2019integer,theis2022lossy,townsend2019practical,yang2024lossy}, anomaly detection \cite{chen2018autoencoder}, out-of-distribution detection \cite{10.5555/3495724.3497461}, semi-supervised learning \cite{10.5555/3295222.3295397}, classifier \cite{chen2024diffusion,yadin2024classification}, image generation \cite{zheng2025direct}, transfer learning \cite{ouyang2024transfer}, density ratio estimation \cite{chen2025dequantified,choi2022density}, language models \cite{gulrajani2023likelihood} and adversarial purification \cite{song2018pixeldefend}.

Rapid advancements in deep generative modelling~\cite{Theis2015ANO} have led to various families of models achieving strong likelihood estimation performance, including energy-based models \cite{gao2020flow}, normalizing flows \cite{song2021maximum,zheng2023improved}, variational autoencoders \cite{kingma2021variational,Kingma2013AutoEncodingVB,nielsen2024diffenc}, diffusion models \cite{ho2020denoising,nichol2021improved,sahoo2024diffusion}, cascaded models \cite{li2024likelihood}, and autoregressive models \cite{child2019generating,hoogeboom2022autoregressive}.
A common underlying structure among many state-of-the-art models~\cite{kingma2021variational,li2024likelihood,sahoo2024diffusion,song2021maximum,zheng2023improved} is the transformation of data into noise via distinct functional mappings. These mappings, despite their differing mathematical forms, can be viewed as variants of diffusion models operating under the same noisy process, a Gaussian channel \cite{cover1999elements,shen2005fast,8007014,wibisono2024optimal}.

Gaussian diffusion models can be broadly categorized into variance-preserving (VP) \cite{ho2020denoising,10.5555/3045118.3045358} and variance-exploding (VE) \cite{kong2023informationtheoretic,song2020score} processes by the variance of injected noise. These two paradigms differ in both the construction of the \textit{forward} process and the formulation of likelihood. VP models are typically treated as Bayesian latent-variable models and trained via variational inference using the evidence lower bound (ELBO) \cite{ho2020denoising,kingma2021variational,nielsen2024diffenc,sahoo2024diffusion,zheng2023improved}, while VE models are interpreted as information-theoretic (IT) channels \cite{1412024,verdu2010mismatched,8007014} that allow for direct likelihood estimation via estimation-theoretic tools \cite{guoscore,kong2023informationtheoretic,song2021maximum}. However, existing likelihood estimation methods for VE models have not matched the performance of VP-based approaches. This raises the natural hypothesis that likelihood performance may be highly sensitive to noise variance. Furthermore, while previous IT bounds \cite{kong2023informationtheoretic,lastras-montaño2018information,song2021maximum,wan2025pruning} may be slightly looser than ELBO, they often enjoy faster convergence and greater interpretability \cite{deconvergence,de2022riemannian,lee2022convergence,song2021maximum,wibisono2024optimal}, leaving open the question of whether IT-based bounds can also be improved as competitive likelihood estimators with enjoying the faster speed and robustness.

Addressing this question requires extending existing \textit{theoretical} frameworks. The extant Shannon–Fisher connections for diffusion models have largely assumed idealized isotropic Gaussian corruption~\cite{10.5555/1795114.1795156,song2021maximum}, yet real-world data rarely align with such simplified assumptions. Imaging data commonly feature Poisson-Gaussian sensor noise \cite{4623175}; dequantization \cite{pmlr-v97-ho19a,hoogeboom2021learning,Theis2015ANO,zheng2023improved} and data smoothing \cite{meng2021improved} often modeled by uniform or symmetric noise (\textit{e.g.} a Laplacian kernel) addition to improve tail coverage, sharp transitions and robustness. Additionally, recent generative models deliberately introduce Poisson~\cite{xu2022poisson}, heavy-tailed $\alpha$-stable perturbations~\cite{yoon2023scorebased} or structured noise with the data distribution \cite{sahoo2024diffusion,singh2025squeezed}, emphasizing the practical needs for more generalized frameworks.

In specific, when a continuous‑density model is fitted directly to discrete data, the likelihood evaluation becomes singular and severely degrades performance. The conventional remedies, uniform or variational dequantization \cite{pmlr-v97-ho19a,song2021maximum,Theis2015ANO}, inject auxiliary noise but suffer two drawbacks: they require an additional training phase, which is hard to train to the optimal, and, in general, introduce a pronounced training-evaluation gap that inflates the NLL performance \cite{zheng2023improved} via the mismatched noise.

In this work, to eliminate both sources of discrepancy and instabilities in maximum likelihood learning with diffusion models, we propose \emph{variance-aware likelihood bounds} via \emph{arbitrary isotropic warm-up noise} perturbation. Our main contributions and findings are summarized as followed:

\begin{itemize}
    \item We first prove Theorems~\ref{theorem 1} and Proposition~\ref{theorem 2}, showing that for any isotropic noise, the Fisher information measures are asymptotically equivalent to their Shannon counterparts, thereby establishing a maximum likelihood learning framework well beyond the Gaussian setting.

    \item Building on this insight, Theorem~\ref{theorem ind} and Proposition~\ref{proposition} provide tightened analytic bounds on the likelihood of diffusion models with only minor architectural changes, specifically, a logarithm signal-to-noise ratio based parameterization and an additional low‑variance noise regime, which together eliminate the train–test gap and stabilize optimization near \(t = 0\).
    \item Empirically, ablation studies confirm the effectiveness of our bounds. Using an efficient importance-sampling scheme, our method achieves \textbf{2.50 bits/dim} on CIFAR-10 and new state-of-the-art results of \textbf{3.01}, \textbf{2.91}, and \textbf{2.59 bits/dim} on ImageNet-32, -64, and -128, respectively, while requiring only \emph{0.3M training iterations} without any data augmentation. 
\end{itemize}

\section{Background}
\subsection{Incremental Gaussian Channel and Maximum Likelihood Estimation}

Let the \textit{source} dataset with $N$ datapoints be denoted by $\mathbf{X}=\{\mathbf{x}_n\}_{n = 1}^{N}$\footnote{In this paper, random objects are denoted by uppercase letters, and their realizations by lowercase letters. The expectation \(\mathbb{E}(\cdot)\) is taken over the joint distribution of the random variables inside the brackets.}. Assume every datapoint $\mathbf{x}$ is an i.i.d. sample drawn from an unknown distribution $p(\mathbf{x})$ supported on the \textit{source space} $\mathcal{X} \subset \mathbb{R}^D$. Diffusion models \cite{ho2020denoising,kingma2021variational} naturally induce what is known as an \textit{incremental signal-to-noise ratio (SNR) channel} \cite{1412024,zou2025convexity}, which generates a sequence of progressively noisier \( \mathbf{X} \) towards a pure noise:
\begin{equation}\label{channel eq}
    \mathbf{Y}_t = \alpha_t \mathbf{X} + \sigma_t \mathbf{N}, 
\end{equation}
where \(\alpha_t, \sigma^2_t \in \mathbb{R}^+\) are smooth, non-negative, scalar-valued functions with finite derivatives with respect to time \(t\) over the fixed time horizon \(t \in [0,1]\). The noise term \(\mathbf{N}\) consists of independent standard Gaussian entries. Intuitively, the ratio \(\alpha^2_t / \sigma^2_t\) can be interpreted as the signal-to-noise ratio (SNR) at time \(t\), denoted as \(\text{SNR}(t):=\alpha^2_t / \sigma^2_t\). By enforcing \(\text{SNR}(t)\) to be strictly monotonically decreasing over the time interval \(t \in [0,1]\), the output \(\mathbf{Y}_t\) asymptotically approaches a well-defined, analytically tractable stationary distribution \(\pi(\mathbf{x})\) as \(t \to 1\).

To recover and estimate the underlying data distribution $p(\mathbf{x})$  from noisy observations $\mathbf{Y}_t$, one must solve a density estimation problem under this forward noisy channel. From an information-theoretic perspective, the mismatched Gaussian channel \cite{verdu2010mismatched} models the scenario where the true input distribution $p(\mathbf{x})$ is unknown, and one instead uses a hypothesis distribution $q(\mathbf{x})$ for estimation. While classical methods for Gaussian channels remain applicable, their analytic solutions are  typically intractable because it requires sampling from the posterior distribution of the noisy channel. A more tractable alternative involves leveraging the connection between relative entropy and Fisher divergence. Specifically, \cite{10.5555/1795114.1795156,song2021maximum} demonstrated that the weighted score matching (or equivalently, Fisher divergence) objectives could approximate maximum likelihood training of diffusion models\footnote{In likelihood-based generative modeling, this formulation is equivalent to maximizing the expected log-likelihood. In contrast, methods that prioritize sample quality typically optimize the 2-Wasserstein distance.}.

\begin{figure}[t]
    \centering
    \hspace{25pt}
    \begin{minipage}[t]{0.5\textwidth}
        \includegraphics[width=\linewidth]{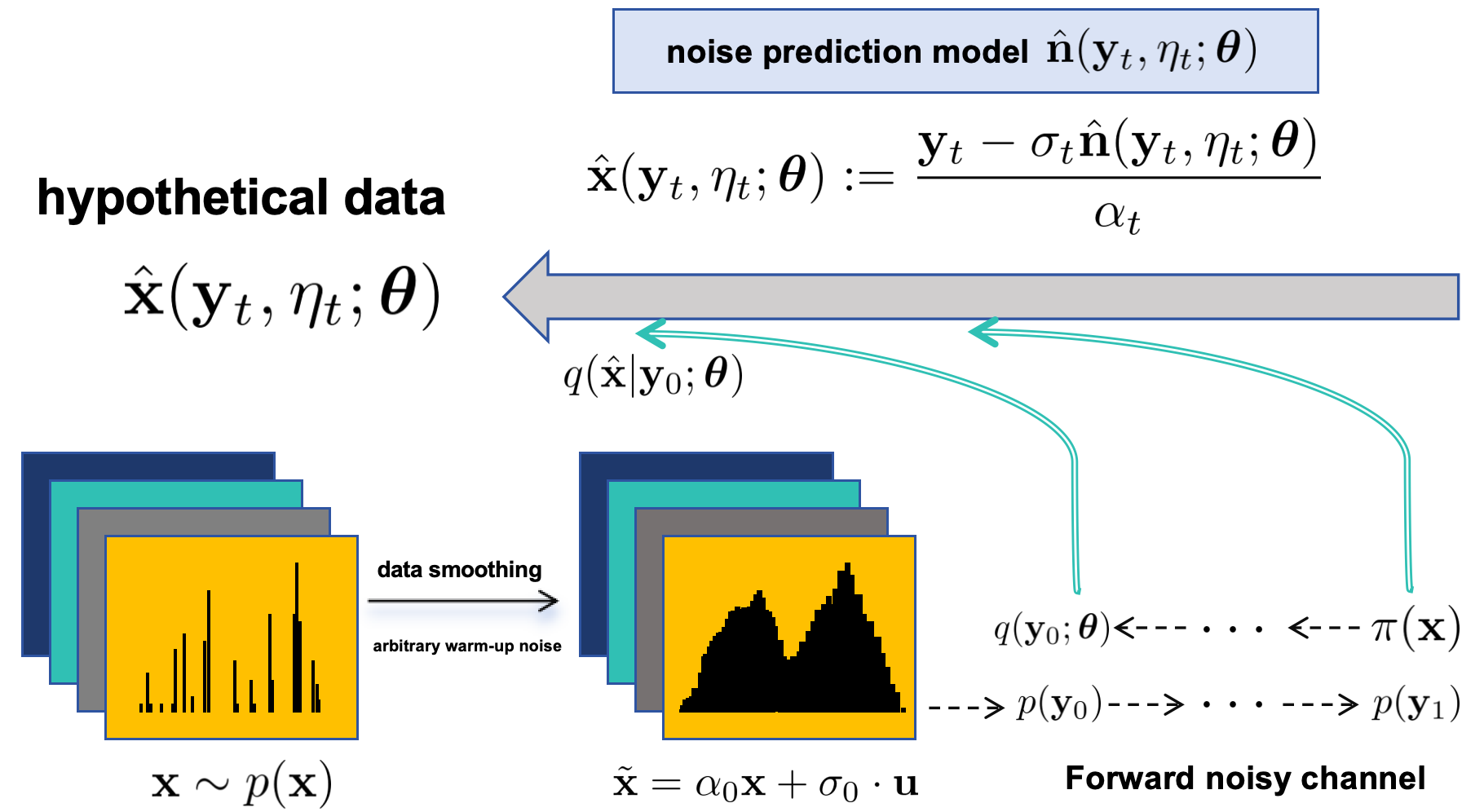}
    \end{minipage}
    \hspace{25pt}
    \begin{minipage}[t]{0.35\textwidth}
        \includegraphics[width=\columnwidth,height=0.78\columnwidth,keepaspectratio]{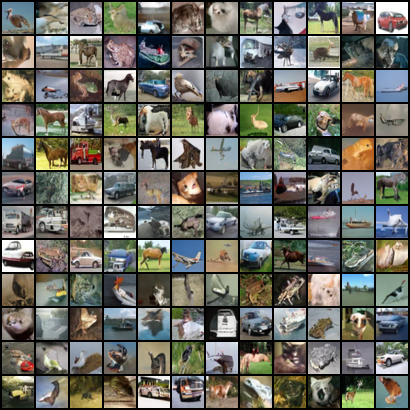}
    \end{minipage}
    \caption{Toy example illustrating our method (left) and samples generated on CIFAR-10 (right). We apply an identical \textit{warm-up} channel to both the data $p(\mathbf{x})$ and the model distribution $q(\mathbf{x};\boldsymbol{\theta})$: injecting arbitrary noise $\boldsymbol{\Psi}$ produces smoothed data $\tilde{\mathbf{x}}$ and a correspondingly perturbed model $\tilde{q}_\theta$. This results in \textit{two} variance-regime mismatched channels, a low-variance arbitrary-perturbation channel and a high-variance Gaussian noise regime, that share identical channel dynamics but differ in their priors, under which training aligns the Gaussian-perturbed distributions and learns the denoising map.} 
    \label{fig:pipeline}
\end{figure}

\subsection{Likelihood of Diffusion Models}\label{section:2.2}

Average log‑likelihood is widely recognized as the default metric for evaluating generative models. Previous work has largely prioritized perceptual quality, emphasizing coarse scale patterns and global consistency of generated images, with common metrics such as the Fréchet Inception Distance (FID) and the Inception Score (IS). In contrast, we optimise for likelihood of the model, a criterion that is inherently sensitive to fine‑scale details and the exact values of individual pixels.

 To evaluate the likelihood in diffusion models, we define a model distribution \( q(\hat{\mathbf{x}};\boldsymbol{\theta}) \), typically parameterized by a neural network with parameters \( \boldsymbol{\theta} \in \boldsymbol{\Theta} \), which aims to approximate the true data distribution \( p(\mathbf{x}) \). Given a sample \( \mathbf{x} \sim p(\mathbf{x}) \) and noise \( \mathbf{n} \sim \mathcal{N}(0, \mathbf{I}) \), a noisy observation \( \mathbf{y}_t \) is generated via the forward channel described in Equation~\eqref{channel eq}. Under the mismatched channel framework, the marginal distribution over \( \mathbf{y}_t \) induced by the model sample \( \hat{\mathbf{x}}\sim q(\hat{\mathbf{x}};\boldsymbol{\theta}) \) is given by:
\begin{equation}
    q(\mathbf{y}_t;\boldsymbol{\theta}) = \int_{\mathbb{R}^D} q(\hat{\mathbf{x}};\boldsymbol{\theta}) \, p(\mathbf{y}_t | \mathbf{x}) \, d\hat{\mathbf{x}}.
\end{equation}

In VP diffusion modeling, the stationary distribution at \( t = 1 \) is defined to be the stationary distribution \( \pi(\mathbf{x}) = \mathcal{N}(0, \mathbf{I}) \), which serves as the starting point for the sampling process. Since our primary focus is density estimation and probabilistic modeling, we defer implementation and details of sampling process to the Appendix.~\ref{sampling process} and optimization algorithms to future work.

Following \cite{song2021denoising}, we parameterize the model distribution using a noise prediction network. Specifically, the model \( \hat{\mathbf{n}}(\mathbf{y}_t, \eta_t;{\boldsymbol{\theta}}) \) is trained to predict the Gaussian noise \( \mathbf{n} \) that was added during the forward process, where \( \eta_t := -\log \mathrm{SNR}(t) \) defines the noise schedule in log-SNR space. A predicted hypothetical data $\hat{\mathbf{x}}$ is then obtained as:
\begin{equation}
    \hat{\mathbf{x}}(\mathbf{y}_t, \eta_t;{\boldsymbol{\theta}}) := \frac{\mathbf{y}_t - \sigma_t \hat{\mathbf{n}}(\mathbf{y}_t, \eta_t;{\boldsymbol{\theta}})}{\alpha_t}.
\end{equation}

Theoretically, the training objective of likelihood-based diffusion models is to minimize the KL divergence (see Def.~\ref{KL definition}) between the true data distribution $p(\mathbf{x})$ and the continuous model distribution $q(\hat{\mathbf{x}};\boldsymbol{\theta})$, i.e., $D_{\text{KL}}(p(\mathbf{x})\|q(\hat{\mathbf{x}};\boldsymbol{\theta}))$. However, as the noise variance $\sigma_t^2\rightarrow0$, the SNR diverges, leading to numerical instability during both training and sampling \cite{kim2022soft}. To address this, practical diffusion models typically start the forward process at a small positive time $t=\epsilon>0$ instead of $t=0$ for improved stability. Yet, this small time offset introduces additional perturbation, and the corresponding training objective becomes $D_{\text{KL}}(p(\mathbf{y}_{\epsilon})\|q(\mathbf{y}_{\epsilon};\boldsymbol{\theta}))$, which differs from the evaluation objective $D_{\text{KL}}(p(\mathbf{x})\|q(\hat{\mathbf{x}};\boldsymbol{\theta}))$. This discrepancy causes a mismatch between the training \textit{expected log-likelihood} ($\mathbb{E}_{p(\mathbf{y}_{\epsilon})}[\log q(\mathbf{y}_{\epsilon};\boldsymbol{\theta})]$) and testing \textit{expected log-likelihood} ($\mathbb{E}_{p(\mathbf{x})}[\log q(\hat{\mathbf{x}};\boldsymbol{\theta})]$).

\subsection{Dequantization for Density Estimation}\label{dequantization}
When modeling real-world data, care must be taken to ensure that the reported likelihood values are meaningful \cite{Theis2015ANO}. Since such data is typically discrete, using continuous density models directly can lead to arbitrarily high likelihoods due to singularities. To mitigate this issue, it is now standard practice to add real-valued noise to integer-valued inputs, a process known as \textit{dequantization} \cite{dinh2017density,song2021maximum,10.5555/2999792.2999855}. For example, in the case of 8-bit image data, the input values in $\{0,1,\dots,255\}$ are typically perturbed by uniform noise, yielding $\mathbf{v} = \mathbf{x} + \mathbf{u}$ where $\mathbf{u} \sim \mathcal{U}[0,1)^D$. With this transformation, training a continuous density model on the uniformly dequantized data $\mathbf{v}$ can be interpreted as maximizing a lower bound on the log-likelihood of a discrete model defined over the original quantized inputs \cite{pmlr-v97-ho19a,Theis2015ANO}. However, this introduces a training–test gap in diffusion models: during training, the model $q(\mathbf{v};\boldsymbol{\theta})$ is fitted to $p(\mathbf{y}_0)$, which corresponds to a Gaussian distribution centered at each discrete data point; during evaluation, however, $q$ is tested on uniformly dequantized data. Although this improves numerical stability in practice \cite{kim2022soft,song2021maximum}, the mismatch between training and evaluation degrades likelihood performance \cite{zheng2023improved}. Moreover, while variational methods \cite{pmlr-v97-ho19a} allow other forms of noise injection, they introduce an additional optimization stage that is computationally expensive and often difficult to train to optimality.

\section{Variance-Aware Likelihood Estimation of Diffusion models}

In this section, we present an information-theoretical framework for variance-aware likelihood estimation in diffusion models. While retaining the standard score matching objective~\cite{hyvarinen2005estimation,vincent2011connection}, our method introduces a tighter, pointwise upper bound on the negative log-likelihood. We also incorporate a modified forward process and an importance sampling scheme to reduce the variance of the Monte Carlo estimator. Notations and definitions are deferred to the Appendices.~\ref{notation} and~\ref{definition}.

\begin{figure*}[t]

\centering

% ---------- Left: Training ----------
\begin{minipage}[t]{0.47\textwidth}
\begin{algorithm}[H]
\scriptsize
\caption{Training}
\label{alg:training}

\textbf{Pre-processing:} \quad Wart Start\\
\For{each $\mathbf{x}$ in $\mathcal{X}$}{
    $\mathbf{u} \sim \Psi$;  \hspace{6em} \tcp{$\Psi$: arbitrary noise}
    $\tilde{\mathbf{x}} \leftarrow \alpha_0\cdot\mathbf{x} + \sigma(0) \cdot \mathbf{u}$;  \hspace{4em}
    \tcp{store $\tilde{\mathbf{x}}$}
}
\BlankLine
\vspace{3.5pt}
\textbf{Training:}\\
\Repeat{converged}{
    $\tilde{\mathbf{x}} \sim \mathcal{X}$, $\mathbf{n} \sim \mathcal{N}(0,\mathbf{I})$ \;
    $t \sim \mathcal{U}(0,1)$, 
$(\alpha_{t}^2, \sigma_{t}^2) \leftarrow t$\; 
    $\mathbf{y}_t \leftarrow \alpha_t \cdot \tilde{\mathbf{x}} + \sigma_t\cdot \mathbf{n}$\;
    % $\hat{\mathbf{n}}(\mathbf{y}_t,\eta_t;\boldsymbol{\theta}) \leftarrow f(\mathbf{y}_t, \eta_t;\boldsymbol{\theta})$\; \\
    Take gradient descent step on \\
    \hspace*{6em}$\nabla_{\boldsymbol{\theta}}\|\mathbf{n} - \hat{\mathbf{n}}(\mathbf{y}_t,t;\boldsymbol{\theta})\|^2$\;
}
\end{algorithm}
\end{minipage}
\hfill
% ---------- Right: Likelihood evaluation ----------
\begin{minipage}[t]{0.47\textwidth}
\begin{algorithm}[H]
\scriptsize
\caption{Likelihood Evaluation}
\label{alg:likelihood}
\KwIn{Trained model $f_{\boldsymbol{\theta}}$, test data $\mathbf{x}$}
\KwOut{Estimated loss $\mathcal{L}(\boldsymbol{\theta})$, NLL $\ell(\mathbf{x};\boldsymbol{\theta})$}

\For{each $\mathbf{x}$ in $\mathcal{X}$}{
    $\mathbf{u} \sim \Psi$; \hspace{6em} \tcp{$\Psi$: arbitrary noise}
    $\tilde{\mathbf{x}} \leftarrow \alpha_0\cdot\mathbf{x} + \sigma(0) \cdot \mathbf{u}$; \hspace{4em} 
    \tcp{store $\tilde{\mathbf{x}}$}
}
\BlankLine

\textbf{Evaluation:}\\
\hspace{2em}$\eta \sim \rho(\eta)$, $(\alpha_\eta^2, \sigma^2_\eta) \leftarrow \eta$\;
\hspace{2em}$\mathbf{n} \sim \mathcal{N}(0,\mathbf{I})$, \quad
$\hat{\mathbf{n}} \leftarrow f(\mathbf{y}_t, \eta_t; \boldsymbol{\theta})$\;
\hspace{2em}$\mathcal{L}(\sigma_\eta^2;\boldsymbol{\theta})\leftarrow 0.5\cdot Z\,\|\mathbf{n} - \hat{\mathbf{n}}(\mathbf{y}_t,\eta_t;\boldsymbol{\theta})\|^2$ \;
\quad \quad \quad \quad \quad \hspace{2em}\tcp{Z is a normalizing constant}

$\ell(\mathbf{x};\boldsymbol{\theta}) \leftarrow \mathcal{H}(p(\mathbf{y}_1|\mathbf{x}),\pi) + \mathcal{L}(\sigma_t^2;\boldsymbol{\theta})$\\
\Return $\ell(\mathbf{x};\boldsymbol{\theta})$\;
\end{algorithm}
\end{minipage}
\label{fig:isit_algorithms}
\end{figure*}

In specific, we analyze the impact of noise variance (schedule) on likelihood estimation through the lens of Fisher divergence via thermodynamic integration \cite{gelman1998simulating,ogata1989monte} along the entire noise variance-regime space, which includes a low-variance arbitrary noise regime ($0\leq\sigma_t^2<\sigma_0^2$) and a high-variance Gasussian channel ($\sigma_0^2\leq\sigma_t^2\leq\sigma_1^2$), providing a formal connection between score matching objectives to KL divergence under arbitrary isotropic noise perturbations. This analysis leads to an exact expression of the mismatched entropy \( \mathcal{H}(p(\mathbf{x}), q(\mathbf{x}; \boldsymbol{\theta})) \), which improves optimization process, numerical stability and remains compatible with probabilistic modeling under non-Gaussian noise setting. Our method toy example pipeline is illustrated in Fig.~\ref{fig:pipeline} and Algorithms~\ref{alg:training} and~\ref{alg:likelihood}.

\subsection{Relationship between Score Matching and
KL Divergence}\label{section 3.1}

Beyond the pathology discussed in Section~\ref{section:2.2}, real-world scenarios rarely align with the idealized Gaussian assumption. Sensor imperfections introduce structured noise such as Poisson, uniform quantization errors, or impulsive salt-and-pepper patterns. Hence, extending KL–Fisher relations beyond Gaussian perturbations is not merely a theoretical extension, but addresses a pressing practical need in robust generative modeling. Motivated by these practical considerations, we investigate how training objectives that incorporate noise-perturbed distributions relate to the classical maximum likelihood principle. Specifically, we highlight the role of Fisher divergence (see Def.~\ref{RFI definition}) in characterizing the first-order sensitivity of KL divergence under small additive noise. By doing so, we generalize prior results~\cite{lu2022maximum,10.5555/1795114.1795156,song2021maximum} from Gaussian to arbitrary isotropic noise distributions. Importantly, our derived relation ensures that score-matching losses remain consistent with the first-order KL term, thus preserving the maximum-likelihood interpretation for a broader class of models, including Poisson-flow~\cite{xu2022poisson} and Lévy-based diffusion frameworks~\cite{yoon2023scorebased}.

\begin{theorem}[Score Matching as the Small-Noise Limit of KL Divergence]\label{theorem 1}
Let \( \mathbf{X} \sim p(\mathbf{x}) \) be an arbitrary distributed random vector on \( \mathbb{R}^D \), and let \( q(\hat{\mathbf{x}}; \boldsymbol{\theta}) \) be a parametric model with \( \hat{\mathbf{X}} \sim q(\hat{\mathbf{x}}; \boldsymbol{\theta}) \). Define the perturbed observation
\[
\tilde{\mathbf{X}} := \alpha_t \mathbf{X} + \sigma_t \boldsymbol{\Psi},
\]
where \( \boldsymbol{\Psi} \) is a random vector independent of $\mathbf{X}$ , satisfying \( \mathbb{E}[\boldsymbol{\Psi}] = 0 \) and \( \mathrm{Cov}(\boldsymbol{\Psi}) = \mathbf{I} \). Let \( p_{\sigma_t^2} \) and \( q_{\sigma_t^2} \) denote the densities of $\tilde{\mathbf{X}}$ under \( p \) and \( q \) with noise variance $\sigma_t^2$, respectively. Suppose that the KL divergence \( D_{\mathrm{KL}}(p_{\sigma_t^2} \| q_{\sigma_t^2}) \) is finite for sufficiently small \( \sigma_t^2 \). Then the following limit holds:
\begin{equation}
    \left. \frac{d}{d\sigma_t^2} D_{\mathrm{KL}}(p_{\sigma_t^2} \| q_{\sigma_t^2}) \right|_{\sigma_t^2 \to 0^+}
= -\frac{1}{2} \int_{\mathbb{R}^D} p(\mathbf{x}) \left\| \nabla \log p(\mathbf{x}) - \nabla \log q(\hat{\mathbf{x}}; \boldsymbol{\theta}) \right\|^2 d\mathbf{x},
\end{equation}
i.e.,
\[
\left. \frac{d}{d\sigma_t^2} D_{\mathrm{KL}}(p(\tilde{\mathbf{x}}) \| q(\tilde{\mathbf{x}}; \boldsymbol{\theta})) \right|_{\sigma_t^2 \to 0^+}
= -\frac{1}{2} I\left(p(\mathbf{x}) \| q(\hat{\mathbf{x}};\boldsymbol{\theta})\right),
\]
where \( I(\cdot\|\cdot) \) denotes the Fisher divergence (equivalently, score matching objective) between \( p \) and \( q \).
\end{theorem}
\begin{proof}
    See Appendix.~\ref{proof:theorem1}.
\end{proof}

To further investigate the effect of additive noise on training objectives, we consider a second-order expansion of the KL divergence in terms of the noise variance \( \sigma_t^2 \). When both the true and model distributions are perturbed by small additive arbitrary noise, the KL divergence between them satisfies:
\begin{equation}
    D_{\mathrm{KL}}(p(\mathbf{x}) \| q(\hat{\mathbf{x}};\boldsymbol{\theta})) = D_{\mathrm{KL}}(p(\tilde{\mathbf{x}}) \| q(\tilde{\mathbf{x}}; \boldsymbol{\theta})) + \frac{\sigma_t^2}{2} I(p \| q) + o(\sigma_t^2).
\end{equation}

While maximum likelihood aims to minimize the KL divergence directly, according to Theorem.~\ref{theorem 1}, score matching, defined through the Fisher divergence~\cite{hyvarinen2005estimation}, does not directly minimize the KL divergence explicitly. The two objectives coincide only in the limit \( \sigma^2_t \to 0 \) where score matching captures the first-order sensitivity of KL divergence to additive noise. When \( \sigma^2_t \) is not infinitesimal, the Fisher term may dominate, potentially leading to biased or unstable solutions. This result significantly generalizes the observation that score matching seeks to eliminate its derivative in the scale space at $t=0$ into arbitrary noise settings. In contrast, it is known that variational methods are known to be highly sensitive  to noisy training data (see Appendix.~\ref{app:elbo_sensivitity} for more details), which may give rise to many false extreme values, whereas score matching tends to be more robust to small perturbation in training data, suggesting that it seeks parameters which lead to models robust to noisy setting.%This suggests the score matching that it seeks parameters that lead to models robust to noisy setting.

Therefore, the generalized result in Theorem~\ref{theorem 1}, formalizing its role as a local approximation for likelihood-based training beyond the Gaussian setting \cite{10.5555/1795114.1795156}. To build upon this foundation, we next consider the case where the forward process begins at a small but strictly positive variance, satisfying $0<\sigma_0^2\ll1$. This setting allows us to integrate from $\sigma_0^2$ rather than from zero, yielding a tractable and numerically stable lower bound while preserving consistency with the asymptotic result established above. We formalize this construction in the next section.

\subsection{Bounding the Mismatched Entropy with Thermodynamic Diffusion}

Building on Section~\ref{section 3.1}, we can now consider a practical setting where the forward diffusion process \eqref{channel eq} begins at a small but strictly positive noise level $0<\sigma_0^2\ll1$. In this regime, we derive the following approximation on the mismatched entropy (see Def.~\ref{cross-entropy}) based on denoising score matching objectives \cite{vincent2011connection} via thermodynamic integration \cite{gelman1998simulating,ogata1989monte} along the interval $[\sigma_0^2,\sigma_1^2]$.

\begin{proposition}[Thermodynamic Decomposition of Mismatched Entropy]\label{proposition}
    Consider the signal model \eqref{channel eq}, suppose $p(\tilde{\mathbf{x}})$ and $q(\tilde{\mathbf{x}};{\boldsymbol{\theta}})$ have continuous second-order derivatives and finite second moments. Denote by $p(\mathbf{y}_1)$ and $\pi(\mathbf{x})$ the output signals of channel at time $t=1$ when inputs are $p(\mathbf{x})$ and $q(\hat{\mathbf{x}};{\boldsymbol{\theta}})$ respectively. Assume $\pi(\mathbf{x})=\mathcal{N}(0,\mathbf{I})$, which is independent of $\boldsymbol{\theta}$. Let $p(\mathbf{y}_t|\mathbf{x})$ denote the channel \eqref{channel eq} for any $t\in[0,1]$, then for arbitrary datapoint $\mathbf{x}$ and small $\sigma_0^2$ with $0<\sigma_0^2\ll1$:
     \begin{equation}
         \mathcal{H}(p(\mathbf{x}),q(\hat{\mathbf{x}};{\boldsymbol{\theta}})) = \mathcal{H}(p(\mathbf{y}_1),\pi(\mathbf{x})) + \mathcal{J}_{\text{DSM}}(\boldsymbol{\theta};\sigma_t^2(\cdot)) - \frac{1}{2}\int_{\sigma_0^2}^{\sigma^2_1}\mathbb{E}\|\nabla_{\mathbf{y}_{t}}\log p(\mathbf{y}_{t}\vert\mathbf{x})\|^{2}d\sigma^2_t + o(\sigma_0^2).
     \end{equation}
     Here, the denoising score matching objective is defined as
     \begin{equation}
         \mathcal{J}_{\text{DSM}}(\boldsymbol{\theta};\sigma^2(\cdot))  := \frac{1}{2}\int_{\sigma_0^2}^{\sigma^2_1} \mathbb{E}\|\nabla_{\mathbf{y}_{t}}\log p(\mathbf{y}_{t}\vert  \mathbf{x}) - \hat{\boldsymbol{s}}(\mathbf{y}_{t};\boldsymbol{\theta})\|^{2}_2\,d\sigma^2_t, %\label{7}
     \end{equation}
     and \(\hat{\boldsymbol{s}}(\mathbf{y}_t; \boldsymbol{\theta}) = \nabla_{\mathbf{y}_t} \log q(\mathbf{y}_t; \boldsymbol{\theta}) := - \hat{\mathbf{n}}(\mathbf{y}_t,\eta_t;{\boldsymbol{\theta}})/\sigma_t\) is a score network estimator.
\end{proposition}
\begin{proof}
    Detailed proof defers in Appendix.~\ref{proof:prop}.
\end{proof}

The Proposition~\ref{proposition} preserves the theoretical guarantees of the Fisher–KL expansion while yielding a tractable and computable loss. The resulting bound remains asymptotically tight as $\sigma_0^2 \rightarrow 0$, up to an additive residual term of order $o(\sigma_0^2)$. The resulting bound remains asymptotically tight as Proposition~\ref{proposition} can be interpreted through the lens of the cost of mismatch in statistical inference \cite{jiao2017relations}. In this setting, the true data distribution is $p$, but the decoder employs an estimator optimized for a mismatched model $q_{\boldsymbol{\theta}}$. The resulting mismatched entropy $\mathcal{H}(p, q_{\boldsymbol{\theta}})$ quantifies the average code length incurred under this model discrepancy. It shows that this quantity decomposes into three parts: a mismatched output distribution loss, a score approximation error arising from the use of $\hat{\boldsymbol{s}}(\cdot; \boldsymbol{\theta})$ instead of the true score, and an irreducible term linked to the Fisher information of the channel.

\subsection{Bounding the log-likelihood on Individual Datapoints}\label{individual likelihood}

In many applications that benefit from likelihood optimization \cite{10.5555/3454287.3454635,hoogeboom2019integer}, it is desirable to evaluate the log-likelihood of individual data points, which can be highly sensitive to fine-scale variations and the precise values of pixel intensities~\cite{kingma2021variational}. To account for this sensitivity, we derive a pointwize lower bound that remains tractable and stable under finite noise injection.

\begin{theorem}\label{theorem ind}
    Let $p(\mathbf{y}_t \vert \mathbf{x})=\mathcal{N}(\alpha_t\mathbf{x},\sigma_t^2\mathbf{I})$ denote the Gaussian channel at any time $t\in[0,1]$. With the same notations and conditions in Proposition~\ref{proposition}, we have
    \begin{equation}
        -\log q(\hat{\mathbf{x}};{\boldsymbol{\theta}}) \leq \mathcal{H}(p(\mathbf{y}_{1}\vert\mathbf{x}),\mathbf{\pi}(\mathbf{x})) + \mathcal{L}_{\text{DSM}}(\sigma_t^2;\boldsymbol{\theta}), \label{7}
    \end{equation}
    in which $\mathcal{L}_{\text{DSM}}(\sigma_t^2;\boldsymbol{\theta})$ is defined as 
    \begin{equation}\label{9}
        \mathcal{L}_{\text{DSM}}(\sigma_t^2;\boldsymbol{\theta}) := \frac{1}{2}\int_{\sigma_0^2}^{\sigma_1^2} \mathbb{E}_{p(\mathbf{y}_t\vert\mathbf{x})}\|\nabla_{\mathbf{y}_{t}}\log p(\mathbf{y}_{t}\vert\mathbf{x}) - \hat{\boldsymbol{s}}(\mathbf{y}_{t};\boldsymbol{\theta})\|^{2}_2\,d\sigma^2_t, 
    \end{equation}
    and $\mathcal{H}(p(\mathbf{y}_{1}\vert\mathbf{x}),\mathbf{\pi}(\mathbf{x}))$ is given by 
    \begin{equation}
        \mathcal{H}(p(\mathbf{y}_{1}\vert\mathbf{x}),\mathbf{\pi}(\mathbf{x}))=D_{\text{KL}}(p(\mathbf{y}_{1}\vert\mathbf{x})\|\pi(\mathbf{x})) + \mathcal{H}( p(\mathbf{y}_{1}\vert\mathbf{x})) .
    \end{equation}
\end{theorem}
\begin{proof}
    Detailed proof defers in the Appendix.~\ref{proof:theorem2}.
\end{proof}

\subsection{Variance Reduction with Importance Sampling}\label{importance sampling}

A key challenge in evaluating diffusion models lies in accurately and efficiently estimating the loss integral of Equation~\eqref{9}, especially under a chosen noise variance schedule \( \sigma_t^2 \). In Appendix.~\ref{B.1} we show that during training, the noise schedule acts as an importance sampling distribution for loss estimation, and therefore plays a crucial role in efficient optimization. 

Among such, popular choices include linear \cite{ho2020denoising}, cosine \cite{nichol2021improved}, and learnable parameterizations \cite{kingma2021variational,sahoo2024diffusion}, with sigmoid-based schedules \( \sigma_t^2 = \text{sigmoid}(\eta(t)) \) remaining prevalent in state-of-the-art models. Accordingly, we employ a Monte Carlo estimator of this loss for evaluation and optimization. To further improve the efficiency and reduce variance in Monte Carlo estimation, we propose two importance sampling (IS) strategies (see Appendix.~\ref{VR:IS} for details).

Formally, the diffusion training objective can be expressed as an integral over the noise variance. It is shown in Appendix.~\ref{Variance-Aware Likelihood Bounds} that our denoising score matching loss \eqref{7} could be simplified to a noise prediction model $\hat{\mathbf{n}}(\mathbf{y}_t,\eta_t;{\boldsymbol{\theta}})$ that directly infers the noise $\mathbf{n}$ that was used to generate $\mathbf{y}_t$:
\begin{equation}
    \mathcal{L}_{\text{DSM}}(\sigma_t^2;\boldsymbol{\theta}) = \frac{1}{2} \int_{\sigma_0^2}^{\sigma_1^2} \mathbb{E}_{n\sim\mathcal{N}(0,\mathbf{I})}  \left[\sigma_t^{-2}\|\mathbf{n} - \hat{\mathbf{n}}(\mathbf{y}_t,\eta_t;{\boldsymbol{\theta}})\|_2^2 \right]\ d\sigma_t^2. \label{12}
\end{equation}

In practice, the evaluation of the integral is time-consuming, and Monte-Carlo methods are used to unbiasedly estimate it. In particular, the variance of the estimator has a direct impact on the optimization process, influencing both its stability and convergence speed. To mitigate this issue, it can be beneficial to decouple the integration variable used for loss estimation from the time variable employed during training. To this end, we introduce the negative log-SNR variable $\eta$ and rewrite the training objective \eqref{12} as an expectation over this reparameterized variable, as shown followed. By interpreting the integral above as an expectation over the $\eta$, the Monte Carlo estimation highlights the role of $\rho(\eta)$ as an IS distribution and $w(\eta)$ as a noise variance-dependent weighting function:
\begin{equation}
    \mathcal{L}_{\text{DSM}}(\sigma_t^2;\boldsymbol{\theta}) = \frac{1}{2}\mathbb{E}_{\mathbf{n}\sim\mathcal{N}(0,\mathbf{I}),\eta\sim\rho(\eta)}\left[\frac{w(\eta)}{\rho(\eta)}\|\mathbf{n} - \hat{\mathbf{n}}(\mathbf{y}_t,\eta_t;{\boldsymbol{\theta}}
    )\|_2^2\right].\label{13}
\end{equation}

    Equivalently, optimizing $\rho(\eta)$ can be viewed as learning a monotone mapping $\eta(t):[0,1]\rightarrow[\eta_0,\eta_1]$, which corresponds to the inverse cumulative distribution function (CDF) of $\rho(\eta)$. Thus, the Monte Carlo estimator achieves reduced variance and admits closed-form inverse-CDF sampling, which motivates the following intuitively designed proposals.

Here, we sample $\eta \sim \rho(\eta)$ with a continuous distribution $\rho(\eta) \propto w(\eta)$ that renders the weight in \eqref{13} time-invariant, so that the noise prediction error remains scale-consistent across all $\eta \in [\eta_0,\eta_1]$ values without any amplification or shrinkage during the estimation. This is analogous to the likelihood-weighted and flow matching methods which conduct uniform sampling for all $\eta$ \cite{kingma2023understanding,zheng2023improved}.

Furthermore, to further minimize the variance of the Monte Carlo estimator, we can optimize the sampling distribution $\rho(\eta)$ directly \cite{kingma2021variational,zheng2023improved}. Specifically, we parameterize $\eta(t)$ as a monotonic neural network and adjust it to minimize the variance of the diffusion loss (see Appendix.~\ref{learnedIS}). While this \textit{learned} proposal can effectively reduce Monte Carlo approximation variance, it introduces additional optimization overhead and potential instability. Empirically, we find that the hand-crafted proposal achieves better convergence speed and comparable NLL results, without requiring additional training objectives or neural components. As such, we adopt our designed IS as our default strategy.

\subsection{Training-Free Dequantization with Warm-up Noise}

Real-world datasets usually contain discrete input, such
as images or texts, and must be dequantized when training continuous‐density models. As discussed in Section~\ref{dequantization}, and more details in Appendix.~\ref{app:dequantization}, while dequantization makes diffusion models applicable to discrete data, the training–test discrepancy originates from the mismatch in noise injection between training and inference/testing. To reduce this discrepancy, we introduce an arbitrary isotropic \textit{warm-up} noise $\mathbf{u}\sim\mathbf{\Psi}$ and inject it into training, inference and testing, that fit seamlessly into continuous-density diffusion models, and preserves the original maximum‑likelihood objective without retraining.

By Theorem~\ref{theorem 1}, the generalized KL–Fisher identity holds for any isotropic noise, which means that dequantization schemes previously restricted to uniform and Gaussian noise can now be extended to, for instance, logistic and Laplacian perturbations. However, existing average log-likelihood bounds \cite{pmlr-v97-ho19a,Theis2015ANO} are derived under mismatched entropy, defined as the sum of the KL divergence and the differential entropy. This motivates the generalization of the differential entropy formulation to arbitrary noise distributions, as follows, a result known as the de Bruijn identity in information theory.

\begin{proposition}[de Bruijn identity with arbitrary noise \cite{rioul2010information}]\label{theorem 2}
    Let \( \tilde{\mathbf{X}} = \alpha_t \mathbf{X} + \sigma_t \boldsymbol{\Psi} \), where \( \mathbf{X}, \boldsymbol{\Psi} \in \mathbb{R}^D \) are independent random vectors, with \( \boldsymbol{\Psi} \) satisfying \( \mathbb{E}[\boldsymbol{\Psi}] = 0 \) and \( \operatorname{Cov}(\boldsymbol{\Psi}) = \mathbf{I} \). Assume the probability density function \( p(\mathbf{x}) \) of \( \mathbf{X} \) is twice continuously differentiable and decays sufficiently fast at infinity, and that the Fisher information \( \mathcal{J}(\mathbf{X}) \) exists and is finite. Then,
    \begin{equation}
        \frac{d}{d\sigma_t^2}\mathcal{H}\bigl(p(\tilde{\mathbf{x}})\bigr)\bigg|_{\sigma_t^2 \rightarrow  0^{+}} = \frac{1}{2} \mathcal{J}\big(p(\mathbf{x})\big), \label{5}
    \end{equation}
    where \( \mathcal{H}(\cdot) \) denotes the differential Shannon entropy (see Def.~\ref{def:entropy} and \( \mathcal{J}(\cdot) \) denotes the Fisher information (see Def.~\ref{def:fisher}).
\end{proposition}

The key implication of \eqref{5} is that the derivative is independent of the detailed statistics of the noise distribution. This suggests that the differential entropy forms a smooth manifold $\sigma_t \boldsymbol{\Psi}$, whose local geometry resembles an isotropic quadratic bowl. Combining Proposition~\ref{theorem 2} with the Theorem~\ref{theorem 1} gives a characterization of the gap between the true discrete cross-entropy and its continuous surrogate. Such dequantization noise can provide a tighter likelihood bound yet with no additional training costs. Our detailed empirical settings of dequantization noise could be found in Appendix.~\ref{dequantization}.

\subsection{Numerical Stability}

All theoretical developments in this work assume that the forward diffusion process begins at a small but strictly positive variance level \( \sigma_0^2 > 0 \). This choice avoids the pathological behaviour of the \( \sigma_0^2 \to 0 \) limit, which often leads to numerical instability in both training and evaluation~\cite{kim2022soft,kingma2021variational,song2021maximum}. Accordingly, we restrict the variance schedule to the interval \( [\sigma_0^2, \sigma_1^2] \) throughout all derivations and implementations. Truncating the lower limit of integration to \(\sigma_0^2\) introduces an approximation error of order \( o(\sigma_0^2) \), as discussed in Section~\ref{section 3.1}, which is negligible in both theory and practice. This ensures that our likelihood-weighted score matching objective remains a consistent estimator of the log-likelihood, while simultaneously improving numerical stability. Empirically, we observe that this design remains robust across diverse noise distributions, further confirming the practical reliability. 

\section{Related Works}

Most prior analyses of variance schedules \cite{albergobuilding,lipmanflow,nichol2021improved} in diffusion models focus on their role in reverse-time sampling \cite{bao2022analyticdpm,kingma2021variational,li2025evodiff,nielsen2024diffenc,sahoo2024diffusion,song2021denoising,zheng2024non,zheng2025bidirectional} (e.g., by approximating posteriors or solving reverse-time ODEs~\cite{lu2022maximum,zheng2023improved}). In contrast, our work examines how the forward variance schedule influences both robustness and likelihood estimation. Proposition~\ref{theorem 2} further generalizes the classical de Bruijn identity \cite{1057105} from Gaussian to arbitrary noise perturbations. Whereas prior information‐theoretic studies \cite{fan2025differential,wan2025enhancing,zou2025revisit} have characterized the connection of KL and Fisher divergences under Gaussian smoothing \cite{lu2022maximum,10.5555/1795114.1795156,song2021maximum,zou2025convexity}, we extend these relationships to non‐Gaussian noise families. Finally, \cite{kong2023informationtheoretic} studies relative‐entropy interpretations of denoising objectives and connects them to mean-square error, but does not address mismatched estimation \cite{jiao2017relations,verdu2010mismatched} in the sense of our framework.

\section{Experiments}
In this section, we present our training procedure and experiment settings, and our ablation studies to demonstrate how our techniques improve the likelihood of diffusion models. We report negative log-likelihood (NLL, in bits/dim) and Fréchet Inception Distance (FID) scores for all experiments, and compare convergence speed across models. Then we evaluate our model on a lossless compression benchmark against several competitive baselines in neural compression literature. Note that we focus here on pushing the state-of-the-art in density estimation, and while we report FID for completeness, we defer sample quality optimization to future work. Additional results can be found in Appendix~\ref{Appendix:experimental}.

\paragraph{Datasets and Implementation}  
We evaluate on CIFAR-10, anti-aliased ImageNet-32 dataset, ImageNet-64 and -128.  No data augmentation is applied.  We adopt the same architecture and hyperparameters as the VDM model \cite{kingma2021variational}, including a U-Net consisting of convolutional ResNet blocks without any downsampling. Unlike embeddings of diffusion time \( t \), we embed the log-SNR value \( \eta(t) \) or its reverse CDF time embedding instead. This modification better reflects the underlying noise scale and leads to improved likelihood estimation with importance sampling trick in practice.

\paragraph{Selection of Noise Schedules}  
We consider two representative choices:  
1. Variance-Preserving (VP):  
   \(\alpha_\eta^2 + \sigma_\eta^2 = 1\), ensuring unit marginal variance throughout the forward process.  
2. Straight-Path (SP):  
   \(\alpha_\eta + \sigma_\eta = 1\), corresponding to linear interpolation in data space as in \cite{albergobuilding,lipmanflow}.
We omit VE schedules, since in our preliminary runs the best VE model attained only 3.27 bits/dim on CIFAR-10, substantially worse than VP and SP. Moreover, to further illustrate the effects of the noise variance, we consider to evaluate various variance functions $\sigma_t^2$ shown in Appendix.~\ref{B.2}.

\subsection{Likelihood and Samples}

\begin{table}[t]
\caption{Negative log-likelihood (NLL) in bits per dimension (BPD), FID, and number of training iterations (Iter., in million) on CIFAR-10 and ImageNet-32 datasets. “/” indicates results not reported or not applicable. *Denotes results obtained on the original ImageNet-32 release. Boldface denotes the best performance within each column, and blue text marks the second-best.}
\label{tab:nll_fid}
\centering
\scriptsize
\begin{tabular}{@{}lccc@{\quad}ccc@{}}
\toprule
\textbf{Models}
& \multicolumn{3}{c}{CIFAR-10}
& \multicolumn{3}{c}{ImageNet-32} \\
\cmidrule(lr){2-4} \cmidrule(lr){5-7}
& NLL $\downarrow$ & FID $\downarrow$ & Iter.
& NLL $\downarrow$ & FID $\downarrow$ & Iter. \\

%PixelCNN~\cite{van2016pixel}               & 3.03 & / & / & 3.83\textsuperscript{*} & / & / \\
%ScoreSDE \cite{song2020score}              & 2.99 & 2.92 & 1.3  & / & / & / \\
\midrule \textit{(The models with IT methods)} \\

%VDM (+ aug.) \cite{kingma2021variational}  & 2.49 & / & 10  & / & / & / \\

ScoreODE (second order) \cite{lu2022maximum}           & 3.44 & \textbf{2.37} & 1.3  & 4.06\textsuperscript{*} & / & 1.3  \\
ScoreODE (third order) \cite{lu2022maximum}           & 3.38 & 2.95 & 1.3  & 4.04\textsuperscript{*} & / & 1.3  \\
ScoreFlow \cite{song2021maximum}           & 2.80 & 5.34 & 1.6  & 3.79\textsuperscript{*} & 11.20\textsuperscript{*} & 1.6  \\
%Soft Truncation \cite{kim2022soft}         & 3.01 & 3.96 & / & 3.90\textsuperscript{*} & 8.42\textsuperscript{*} & / \\
Flow Matching \cite{lipmanflow}            & 2.99 & 6.35 & 0.391  & 3.53 & \textbf{5.02} & 0.25  \\
Stoch. Interp.\ \cite{albergobuilding}     & 2.99 & 10.27 & 0.5  & 3.48\textsuperscript{*} & \textcolor{blue}{8.49} & 0.6  \\
%Image Transformer~\cite{parmar2018image}   & 2.90 & / & / & 3.77 & / & / \\
%Sparse Transformer~\cite{child2019generating}
%& 2.80 & / & / & / & / & / \\
i-DODE (SP with IS)~\cite{zheng2023improved}            & 2.56 & 11.20 & 6.2  & 3.44 / 3.69\textsuperscript{*} & 10.31 & 2.25 / 2.5\textsuperscript{*}  \\
\midrule \textit{(The models with variational methods)} \\
VDM \cite{kingma2021variational}           & 2.65 & 7.60 & 10  & 3.72\textsuperscript{*} & / & 2  \\
DiffEnc~\cite{nielsen2024diffenc}          & 2.62 & 11.20 & 8  & 3.46 & / & 8  \\
MuLAN~\cite{sahoo2024diffusion}            & 2.55 & 17.62 & 8  & 3.67\textsuperscript{*} & 13.19 & 2  \\
%DistAug~\cite{pmlr-v119-jun20a}          & 2.56 & 12.75 & / & 3.67 & 13.19 & / \\
BSI~\cite{lienen2025generative}            & 2.64 & / & 10  & 3.44 & / & 10   \\
%Reflected Diffusion\cite{lou2023reflected} & 2.68 & 2.72 & / & 3.74\textsuperscript{*} & / & / \\
W-PCDM \cite{li2024likelihood} (VDM \cite{kingma2021variational} weight)              & \textbf{2.35} & 6.23 & 2 & \textcolor{blue}{3.32} & / & 10 \\
W-PCDM \cite{li2024likelihood} (EDM \cite{10.5555/3600270.3602196} weight)            & 10.31 & \textcolor{blue}{2.42} & 2 & / & / & / \\
\midrule
Ours (SP with IS)                   & \textcolor{blue}{2.49} & / & 0.3  & 3.02 & / & 0.3  \\
Ours (VP with IS)                   & 2.50 & 10.18 & 0.3  & \textbf{3.01} & 14.76 & 0.3  \\
\bottomrule
\end{tabular}
\end{table}

Table~\ref{tab:nll_fid} summarizes experimental results on CIFAR-10 and ImageNet-32 (more details in Appendix.~\ref{app:experiment}). At the request of one of the reviewers we also ran our model on additional data sets of higher resolution images baseline (Table.~\ref{tab:comparative}). On four V100 GPUs with CIFAR-10, our model trains slightly faster (2.34 vs. 2.04 iterations/sec) compared to VDM, due to the removal of the additional networks. Typically, achieving benchmark likelihood estimation performance requires several million training iterations, and the training process usually takes a week, or even a month or longer. Table~\ref{tab:comparative} shows that we obtain state-of-the-art NLL on ImageNet-32/64/128 under directly comparable settings, while reducing training cost from millions of iterations in prior work to about 300 thousand.

\begin{table}[t]
  \centering
  \caption{Comparison between our proposed model and other competitive models in the literature in terms of expected negative log likelihood on the test set computed as bits per dimension (BPD). Results from existing models are taken from the literature. “/” indicates results not reported or not applicable. Boldface denotes the best performance within each column, and blue text marks the second-best.}
  \label{tab:comparative}
  \scriptsize
  \begin{tabular}{lcccc}
    \toprule
    \textbf{Model} & \textbf{Type} & \textbf{ImageNet-32} & \textbf{ImageNet-64} & \textbf{ImageNet-128} \\
    \midrule
    PixelCNN  \cite{van2016pixel} & Autoregressive & 3.83 & 3.57 & / \\
    Glow \cite{kingma2018glow} & Flow & 4.02 & 3.81 & / \\
    FLOW++  \cite{pmlr-v97-ho19a} & Flow & 3.86 & 3.69 & / \\
    Sparse Transformer  \cite{child2019generating} & Autoregressive & / & 3.44 & / \\
    Very deep VAE  \cite{child2021very} & Autoencoder  & 3.80 & 3.52 & / \\
    Improved DDPM  \cite{nichol2021improved} & Diffusion & / & 3.54 & / \\
    Routing Transformer \cite{roy2021efficient}  & Autoregressive & / & 3.43 & / \\
    Flow Matching \cite{lipmanflow} & Flow & 3.53 & 3.31 & 2.90 \\
    VDM \cite{kingma2021variational} & Autoencoder & 3.72 & 3.40 & / \\
    LP-PCDM \cite{li2024likelihood} & Diffusion & 3.52 & 3.12 & 2.91 \\
    W-PCDM \cite{li2024likelihood} & Diffusion & \textcolor{blue}{3.32} & \textcolor{blue}{2.95} & \textcolor{blue}{2.64} \\
    \textbf{Ours} & Diffusion & \textbf{3.01} & \textbf{2.91} & \textbf{2.59} \\
    \bottomrule
  \end{tabular}
\end{table}

\begin{table}[t]
  \centering
  \caption{Comparison of NLL on CIFAR-10 and ImageNet-32 with different warm-up noise.}
  \label{tab:nll_fid_compact}
  \scriptsize
  \begin{tabular}{lcccccccc}
    \toprule
    \textbf{Likelihood Estimation Bound} & \multicolumn{4}{c}{CIFAR-10} & \multicolumn{4}{c}{ImageNet-32} \\
    \cmidrule(lr){2-5} \cmidrule(lr){6-9}
                              & Gaussian & Laplace & Logistic & Uniform & Gaussian & Laplace & Logistic & Uniform \\
    \midrule
    ELBO                      & 2.69     & 2.70    & 2.71     & 2.71    & 3.53     & 3.53    & 3.54     & 3.55    \\
    \textbf{Our} (SP + IS)            & 2.49     & 2.49    & 2.50     & 2.52    & 3.01     & 3.00    & 3.02     & 3.09    \\
    \textbf{Our} (VP + IS)            & 2.50     & 2.51    & 2.51     & 2.53    & 3.00     & 3.01    & 3.03     & 3.08    \\
    \bottomrule
  \end{tabular}
\end{table}

\paragraph{Ablation Analysis and Additional Experiments}
Due to the high training cost, we conduct ablation studies only on CIFAR-10. In Table~\ref{tab:side_by_side}, we report both NLL and FID under different noise variance endpoints. We ablate the effect of the different IS weightings and reverse CDF embedding by comparing it to standard sinusoidal time encoding in Appendix.~\ref{app:experiment}. All models are trained for 300K steps with identical settings, except for the variance schedule. For fair comparison, we follow the VDM protocol~\cite{kingma2021variational} and evaluate models using the ELBO-based lower bound (see Appendix.~\ref{app:elbo_sensivitity}).

Our proposed bounds achieve state-of-the-art density estimation on ImageNet datasets and match the best reported results on CIFAR-10. We find that increasing the warm-up noise improves NLL but slightly degrades FID. This trade-off arises because larger noise suppresses pixel-level fluctuations and stabilizes training, while too little noise sufficiently regularize the likelihood objective, leading to worse likelihood estimates, despite improved sample quality. We provide details in the Appendix.~\ref{App:FID}.

We also observe that the gap between ELBO and our IT-based bound enlarges with increasing noise levels. This discrepancy arises from the reconstruction loss used in ELBO \cite{kingma2021variational,nielsen2024diffenc}, which assumes conditional independence of pixels given the latent variable \( \mathbf{y}_0 \). As noise increases, the conditional distribution \( p(\mathbf{y}_0|\mathbf{x}) \) becomes more diffuse around its mode, thereby amplifying the approximation error. These findings highlight a fundamental intuition between likelihood accuracy and perceptual quality under different noise configurations. We leave the further study of this direction in the future.

\subsection{Examining the Warm-up Noise Injection}

We further examine the effect of different noise distributions, Gaussian, Laplace, logistic, and Uniform, each scaled to equal variance. As shown in Table~\ref{tab:nll_fid_compact}, Gaussian noise performs best, followed by Laplace and logistic, while Uniform lags notably behind. This supports our intuition that heavier-tailed, exponential-family noises yield more stable training and improved likelihood estimation. Detailed theoretical discussion and additional analyses, including the connection to differential entropy, Fisher information and manifold hypothesis, are provided in Appendix~\ref{app:warm-up}.

\subsection{Lossless Progressive Coding}
\begin{wraptable}{r}{0.36\textwidth}  
  \vspace{-1.2cm}                      
  \centering
  \scriptsize
  \caption{Lossless compression performance on CIFAR-10 in bits/dim.}\label{tab:inline-bpd}  
  \begin{tabular}{lc}
    \toprule
    Model & Compression Rate (bits/dim) \\
    \midrule
    FLIF \cite{sneyers2016flif} & 4.14 \\
    LBB \cite{pmlr-v97-ho19a} & 3.12 \\
    IDF \cite{hoogeboom2019integer}  & 3.26 \\
    VDM \cite{kingma2021variational} & 2.72 \\
    ARDM \cite{hoogeboom2022autoregressive} & 2.71 \\
    W-PCDM \cite{li2024likelihood}  & \textbf{2.37} \\
    \textbf{Ours}  & 2.57 \\
    \bottomrule
  \end{tabular}
  \vspace{-6pt}                      %
\end{wraptable}
    As shown in prior work \cite{ho2020denoising,kingma2021variational}, likelihood-based generative models can be viewed as latent-variable models for neural lossless compression. We adopt this perspective and implement a Bits-Back Coding scheme \cite{pmlr-v97-ho19a} using our proposed model as the latent component. On CIFAR-10, our method achieves shorter average code lengths, measured in bits per dimension, compared to several strong baselines (see Table~\ref{tab:inline-bpd}). We leave this avenue of research for further work.

\section{Conclusion}

Our generalized KL–Fisher relationship transforms noise injection from a theoretical consideration into a widely applicable practical strategy. This principled framework validates existing approaches employing non-Gaussian perturbations and offers new theoretical tools for tackling real-world generative modeling challenges. Our analysis shows that, in the small-noise regime, the score matching objective asymptotically approximates maximum likelihood. Minimizing this objective yields consistent improvements in likelihood across diverse noise schedules, variance settings, and datasets. When critically combined with importance sampling, our approach achieves on-par likelihood on CIFAR-10 and state-of-the-art likelihood on ImageNet datasets. Our results also motivate future exploration of information-theoretic objectives in generative modeling.

\paragraph{Limitations} 
Our method improves likelihood estimation but does not construct a generative diffusion process under alternative noise. As a result, our method is limited to likelihood evaluation and cannot be directly used for sampling or generation. Dequantization bound, diffusion/drift coefficient and variance configurations are not fully explored. Due to resource limitations, we didn’t explore tuning of hyperparameters and network architectures, which are left for future work. We leave full discussion and future extensions to the Appendix.

\newpage
\bibliography{references}

\begin{thebibliography}{10}

\bibitem{albergobuilding}
Michael~Samuel Albergo and Eric Vanden-Eijnden.
\newblock Building normalizing flows with stochastic interpolants.
\newblock In {\em The Eleventh International Conference on Learning Representations}, 2023.

\bibitem{bao2022analyticdpm}
Fan Bao, Chongxuan Li, Jun Zhu, and Bo~Zhang.
\newblock Analytic-{DPM}: an analytic estimate of the optimal reverse variance in diffusion probabilistic models.
\newblock In {\em International Conference on Learning Representations}, 2022.

\bibitem{deconvergence}
Valentin~De Bortoli.
\newblock Convergence of denoising diffusion models under the manifold hypothesis.
\newblock {\em Transactions on Machine Learning Research}, 2022.
\newblock Expert Certification.

\bibitem{chen2024diffusion}
Huanran Chen, Yinpeng Dong, Shitong Shao, Zhongkai Hao, Xiao Yang, Hang Su, and Jun Zhu.
\newblock Diffusion models are certifiably robust classifiers.
\newblock In {\em The Thirty-eighth Annual Conference on Neural Information Processing Systems}, 2024.

\bibitem{chen2013mismatched}
Minhua Chen and John Lafferty.
\newblock Mismatched estimation and relative entropy in vector {Gaussian} channels.
\newblock In {\em 2013 IEEE International Symposium on Information Theory}, pages 2845--2849. IEEE, 2013.

\bibitem{chen2025dequantified}
Wei Chen, Shigui Li, Jiacheng Li, Junmei Yang, John Paisley, and Delu Zeng.
\newblock Dequantified {Diffusion-Schr\"odinger} bridge for density ratio estimation.
\newblock In {\em Forty-second International Conference on Machine Learning}, 2025.

\bibitem{chen2018autoencoder}
Zhaomin Chen, Chai~Kiat Yeo, Bu~Sung Lee, and Chiew~Tong Lau.
\newblock Autoencoder-based network anomaly detection.
\newblock In {\em 2018 Wireless telecommunications symposium (WTS)}, pages 1--5. IEEE, 2018.

\bibitem{child2021very}
Rewon Child.
\newblock Very deep {VAE}s generalize autoregressive models and can outperform them on images.
\newblock In {\em International Conference on Learning Representations}, 2021.

\bibitem{child2019generating}
Rewon Child, Scott Gray, Alec Radford, and Ilya Sutskever.
\newblock Generating long sequences with sparse transformers.
\newblock {\em arXiv preprint arXiv:1904.10509}, 2019.

\bibitem{choi2022density}
Kristy Choi, Chenlin Meng, Yang Song, and Stefano Ermon.
\newblock Density ratio estimation via infinitesimal classification.
\newblock In {\em International Conference on Artificial Intelligence and Statistics}, pages 2552--2573. PMLR, 2022.

\bibitem{pmlr-v97-cohen19c}
Jeremy Cohen, Elan Rosenfeld, and Zico Kolter.
\newblock Certified adversarial robustness via randomized smoothing.
\newblock In Kamalika Chaudhuri and Ruslan Salakhutdinov, editors, {\em Proceedings of the 36th International Conference on Machine Learning}, volume~97 of {\em Proceedings of Machine Learning Research}, pages 1310--1320. PMLR, 09--15 Jun 2019.

\bibitem{1057105}
M.~Costa.
\newblock A new entropy power inequality.
\newblock {\em IEEE Transactions on Information Theory}, 31(6):751--760, 1985.

\bibitem{1056983}
M.~Costa and T.~Cover.
\newblock On the similarity of the entropy power inequality and the brunn- minkowski inequality (corresp.).
\newblock {\em IEEE Transactions on Information Theory}, 30(6):837--839, 1984.

\bibitem{cover1999elements}
Thomas~M Cover.
\newblock {\em Elements of information theory}.
\newblock John Wiley \& Sons, 1999.

\bibitem{10.5555/3295222.3295397}
Zihang Dai, Zhilin Yang, Fan Yang, William~W. Cohen, and Ruslan Salakhutdinov.
\newblock Good semi-supervised learning that requires a bad {GAN}.
\newblock In {\em Proceedings of the 31st International Conference on Neural Information Processing Systems}, NIPS'17, page 6513–6523, Red Hook, NY, USA, 2017. Curran Associates Inc.

\bibitem{de2022riemannian}
Valentin De~Bortoli, Emile Mathieu, Michael Hutchinson, James Thornton, Yee~Whye Teh, and Arnaud Doucet.
\newblock Riemannian score-based generative modelling.
\newblock {\em Advances in neural information processing systems}, 35:2406--2422, 2022.

\bibitem{dinh2017density}
Laurent Dinh, Jascha Sohl-Dickstein, and Samy Bengio.
\newblock Density estimation using real {NVP}.
\newblock In {\em International Conference on Learning Representations}, 2017.

\bibitem{fan2025differential}
Luyao Fan, Jiayang Zou, Jiayang Gao, and Jia Wang.
\newblock Differential properties of information in jump-diffusion channels.
\newblock {\em arXiv preprint arXiv:2501.05708}, 2025.

\bibitem{Fisher_1925}
R.~A. Fisher.
\newblock Theory of statistical estimation.
\newblock {\em Mathematical Proceedings of the Cambridge Philosophical Society}, 22(5):700–725, 1925.

\bibitem{4623175}
Alessandro Foi, Mejdi Trimeche, Vladimir Katkovnik, and Karen Egiazarian.
\newblock Practical {Poissonian-Gaussian} noise modeling and fitting for single-image raw-data.
\newblock {\em IEEE Transactions on Image Processing}, 17(10):1737--1754, 2008.

\bibitem{gao2020flow}
Ruiqi Gao, Erik Nijkamp, Diederik~P Kingma, Zhen Xu, Andrew~M Dai, and Ying~Nian Wu.
\newblock Flow contrastive estimation of energy-based models.
\newblock In {\em Proceedings of the IEEE/CVF Conference on Computer Vision and Pattern Recognition}, pages 7518--7528, 2020.

\bibitem{gelman1998simulating}
Andrew Gelman and Xiao-Li Meng.
\newblock Simulating normalizing constants: From importance sampling to bridge sampling to path sampling.
\newblock {\em Statistical science}, pages 163--185, 1998.

\bibitem{gulrajani2023likelihood}
Ishaan Gulrajani and Tatsunori~B Hashimoto.
\newblock Likelihood-based diffusion language models.
\newblock {\em Advances in Neural Information Processing Systems}, 36:16693--16715, 2023.

\bibitem{guoscore}
Dongning Guo.
\newblock Relative entropy and score function: New information-estimation relationships through arbitrary additive perturbation.
\newblock In {\em 2009 IEEE International Symposium on Information Theory}, pages 814--818, 2009.

\bibitem{1412024}
Dongning Guo, S.~Shamai, and S.~Verdu.
\newblock Mutual information and minimum mean-square error in {Gaussian} channels.
\newblock {\em IEEE Transactions on Information Theory}, 51(4):1261--1282, 2005.

\bibitem{helminger2021lossy}
Leonhard Helminger, Abdelaziz Djelouah, Markus Gross, and Christopher Schroers.
\newblock Lossy image compression with normalizing flows.
\newblock In {\em Neural Compression: From Information Theory to Applications -- Workshop @ ICLR 2021}, 2021.

\bibitem{pmlr-v97-ho19a}
Jonathan Ho, Xi~Chen, Aravind Srinivas, Yan Duan, and Pieter Abbeel.
\newblock Flow++: Improving flow-based generative models with variational dequantization and architecture design.
\newblock In Kamalika Chaudhuri and Ruslan Salakhutdinov, editors, {\em Proceedings of the 36th International Conference on Machine Learning}, volume~97 of {\em Proceedings of Machine Learning Research}, pages 2722--2730. PMLR, 09--15 Jun 2019.

\bibitem{ho2020denoising}
Jonathan Ho, Ajay Jain, and Pieter Abbeel.
\newblock Denoising diffusion probabilistic models.
\newblock {\em Advances in neural information processing systems}, 33:6840--6851, 2020.

\bibitem{10.5555/3454287.3454635}
Jonathan Ho, Evan Lohn, and Pieter Abbeel.
\newblock Compression with flows via local bits-back coding.
\newblock {\em Advances in Neural Information Processing Systems}, 32, 2019.

\bibitem{ho2021anfic}
Yung-Han Ho, Chih-Chun Chan, Wen-Hsiao Peng, Hsueh-Ming Hang, and Marek Doma{\'n}ski.
\newblock Anfic: Image compression using augmented normalizing flows.
\newblock {\em IEEE Open Journal of Circuits and Systems}, 2:613--626, 2021.

\bibitem{hoogeboom2021learning}
Emiel Hoogeboom, Taco Cohen, and Jakub~Mikolaj Tomczak.
\newblock Learning discrete distributions by dequantization.
\newblock In {\em Third Symposium on Advances in Approximate Bayesian Inference}, 2021.

\bibitem{hoogeboom2022autoregressive}
Emiel Hoogeboom, Alexey~A. Gritsenko, Jasmijn Bastings, Ben Poole, Rianne van~den Berg, and Tim Salimans.
\newblock Autoregressive diffusion models.
\newblock In {\em International Conference on Learning Representations}, 2022.

\bibitem{hoogeboom2019integer}
Emiel Hoogeboom, Jorn Peters, Rianne Van Den~Berg, and Max Welling.
\newblock Integer discrete flows and lossless compression.
\newblock {\em Advances in Neural Information Processing Systems}, 32, 2019.

\bibitem{hyvarinen2005estimation}
Aapo Hyv{\"a}rinen and Peter Dayan.
\newblock Estimation of non-normalized statistical models by score matching.
\newblock {\em Journal of Machine Learning Research}, 6(4), 2005.

\bibitem{jiao2017relations}
Jiantao Jiao, Kartik Venkat, and Tsachy Weissman.
\newblock Relations between information and estimation in discrete-time {L{\'e}vy} channels.
\newblock {\em IEEE Transactions on Information Theory}, 63(6):3579--3594, 2017.

\bibitem{johnson1995continuous}
Norman~L Johnson, Samuel Kotz, and Narayanaswamy Balakrishnan.
\newblock {\em Continuous univariate distributions, volume 2}, volume~2.
\newblock John wiley \& sons, 1995.

\bibitem{10.5555/3600270.3602196}
Tero Karras, Miika Aittala, Samuli Laine, and Timo Aila.
\newblock Elucidating the design space of diffusion-based generative models.
\newblock In {\em Proceedings of the 36th International Conference on Neural Information Processing Systems}, NIPS '22, Red Hook, NY, USA, 2022. Curran Associates Inc.

\bibitem{kim2022soft}
Dongjun Kim, Seungjae Shin, Kyungwoo Song, Wanmo Kang, and Il-Chul Moon.
\newblock Soft truncation: A universal training technique of score-based diffusion model for high precision score estimation.
\newblock In {\em International Conference on Machine Learning}, pages 11201--11228. PMLR, 2022.

\bibitem{kingma2023understanding}
Diederik Kingma and Ruiqi Gao.
\newblock Understanding diffusion objectives as the {ELBO} with simple data augmentation.
\newblock {\em Advances in Neural Information Processing Systems}, 36:65484--65516, 2023.

\bibitem{kingma2021variational}
Diederik Kingma, Tim Salimans, Ben Poole, and Jonathan Ho.
\newblock Variational diffusion models.
\newblock {\em Advances in neural information processing systems}, 34:21696--21707, 2021.

\bibitem{Kingma2013AutoEncodingVB}
Diederik~P. Kingma and Max Welling.
\newblock Auto-encoding variational bayes.
\newblock In Yoshua Bengio and Yann LeCun, editors, {\em 2nd International Conference on Learning Representations, {ICLR} 2014, Banff, AB, Canada, April 14-16, 2014, Conference Track Proceedings}, 2014.

\bibitem{kingma2018glow}
Durk~P Kingma and Prafulla Dhariwal.
\newblock Glow: Generative flow with invertible {$1\times1$} convolutions.
\newblock {\em Advances in neural information processing systems}, 31, 2018.

\bibitem{kong2023informationtheoretic}
Xianghao Kong, Rob Brekelmans, and Greg~Ver Steeg.
\newblock Information-theoretic diffusion.
\newblock In {\em The Eleventh International Conference on Learning Representations}, 2023.

\bibitem{lastras-montaño2018information}
Luis~A. Lastras-Montaño.
\newblock Information theoretic lower bounds on negative log likelihood.
\newblock In {\em International Conference on Learning Representations}, 2019.

\bibitem{lee2022convergence}
Holden Lee, Jianfeng Lu, and Yixin Tan.
\newblock Convergence for score-based generative modeling with polynomial complexity.
\newblock {\em Advances in Neural Information Processing Systems}, 35:22870--22882, 2022.

\bibitem{li2024likelihood}
Henry Li, Ronen Basri, and Yuval Kluger.
\newblock Likelihood training of cascaded diffusion models via hierarchical volume-preserving maps.
\newblock In {\em The Twelfth International Conference on Learning Representations}, 2024.

\bibitem{li2025evodiff}
Shigui Li, Wei Chen, and Delu Zeng.
\newblock Evodiff: Entropy-aware variance optimized diffusion inference.
\newblock {\em arXiv preprint arXiv:2509.26096}, 2025.

\bibitem{lienen2025generative}
Marten Lienen, Marcel Kollovieh, and Stephan G{\"u}nnemann.
\newblock Generative modeling with bayesian sample inference.
\newblock {\em arXiv preprint arXiv:2502.07580}, 2025.

\bibitem{lipmanflow}
Yaron Lipman, Ricky T.~Q. Chen, Heli Ben-Hamu, Maximilian Nickel, and Matthew Le.
\newblock Flow matching for generative modeling.
\newblock In {\em The Eleventh International Conference on Learning Representations}, 2023.

\bibitem{lu2022maximum}
Cheng Lu, Kaiwen Zheng, Fan Bao, Jianfei Chen, Chongxuan Li, and Jun Zhu.
\newblock Maximum likelihood training for score-based diffusion odes by high order denoising score matching.
\newblock In {\em International Conference on Machine Learning}, pages 14429--14460. PMLR, 2022.

\bibitem{10.5555/1795114.1795156}
Siwei Lyu.
\newblock Interpretation and generalization of score matching.
\newblock In {\em Proceedings of the Twenty-Fifth Conference on Uncertainty in Artificial Intelligence}, UAI '09, page 359–366, Arlington, Virginia, USA, 2009. AUAI Press.

\bibitem{meng2021improved}
Chenlin Meng, Jiaming Song, Yang Song, Shengjia Zhao, and Stefano Ermon.
\newblock Improved autoregressive modeling with distribution smoothing.
\newblock In {\em International Conference on Learning Representations}, 2021.

\bibitem{narayanan2007thermodynamic}
Krishna~R Narayanan and Arun~R Srinivasa.
\newblock On the thermodynamic temperature of a general distribution.
\newblock {\em arXiv preprint arXiv:0711.1460}, 2007.

\bibitem{nichol2021improved}
Alexander~Quinn Nichol and Prafulla Dhariwal.
\newblock Improved denoising diffusion probabilistic models.
\newblock In {\em International conference on machine learning}, pages 8162--8171. PMLR, 2021.

\bibitem{nielsen2024diffenc}
Beatrix Miranda~Ginn Nielsen, Anders Christensen, Andrea Dittadi, and Ole Winther.
\newblock Diffenc: Variational diffusion with a learned encoder.
\newblock In {\em The Twelfth International Conference on Learning Representations}, 2024.

\bibitem{ogata1989monte}
Yosihiko Ogata.
\newblock A monte carlo method for high dimensional integration.
\newblock {\em Numerische Mathematik}, 55:137--157, 1989.

\bibitem{ouyang2024transfer}
Yidong Ouyang, Liyan Xie, Hongyuan Zha, and Guang Cheng.
\newblock Transfer learning for diffusion models.
\newblock {\em Advances in Neural Information Processing Systems}, 37:136962--136989, 2024.

\bibitem{palomar2005gradient}
Daniel~P Palomar and Sergio Verd{\'u}.
\newblock Gradient of mutual information in linear vector gaussian channels.
\newblock {\em IEEE Transactions on Information Theory}, 52(1):141--154, 2005.

\bibitem{5165186}
Miquel Payaro and Daniel~P. Palomar.
\newblock Hessian and concavity of mutual information, differential entropy, and entropy power in linear vector gaussian channels.
\newblock {\em IEEE Transactions on Information Theory}, 55(8):3613--3628, 2009.

\bibitem{rioul2010information}
Olivier Rioul.
\newblock Information theoretic proofs of entropy power inequalities.
\newblock {\em IEEE transactions on information theory}, 57(1):33--55, 2010.

\bibitem{roy2021efficient}
Aurko Roy, Mohammad Saffar, Ashish Vaswani, and David Grangier.
\newblock Efficient content-based sparse attention with routing transformers.
\newblock {\em Transactions of the Association for Computational Linguistics}, 9:53--68, 2021.

\bibitem{sahoo2024diffusion}
Subham Sahoo, Aaron Gokaslan, Christopher~M De~Sa, and Volodymyr Kuleshov.
\newblock Diffusion models with learned adaptive noise.
\newblock {\em Advances in Neural Information Processing Systems}, 37:105730--105779, 2024.

\bibitem{salimans2022progressive}
Tim Salimans and Jonathan Ho.
\newblock Progressive distillation for fast sampling of diffusion models.
\newblock In {\em International Conference on Learning Representations}, 2022.

\bibitem{shen2005fast}
Yirong Shen, Matthias Seeger, and Andrew Ng.
\newblock Fast {Gaussian} process regression using {KD}-trees.
\newblock {\em Advances in neural information processing systems}, 18, 2005.

\bibitem{singh2025squeezed}
Jyotirmai Singh, Samar Khanna, and James Burgess.
\newblock Squeezed diffusion models.
\newblock {\em arXiv preprint arXiv:2508.14871}, 2025.

\bibitem{sneyers2016flif}
Jon Sneyers and Pieter Wuille.
\newblock {FLIF}: Free lossless image format based on {MANIAC} compression.
\newblock In {\em 2016 IEEE international conference on image processing (ICIP)}, pages 66--70. IEEE, 2016.

\bibitem{10.5555/3045118.3045358}
Jascha Sohl-Dickstein, Eric~A. Weiss, Niru Maheswaranathan, and Surya Ganguli.
\newblock Deep unsupervised learning using nonequilibrium thermodynamics.
\newblock In {\em Proceedings of the 32nd International Conference on International Conference on Machine Learning - Volume 37}, ICML'15, page 2256–2265. JMLR.org, 2015.

\bibitem{song2021denoising}
Jiaming Song, Chenlin Meng, and Stefano Ermon.
\newblock Denoising diffusion implicit models.
\newblock In {\em International Conference on Learning Representations}, 2021.

\bibitem{song2021maximum}
Yang Song, Conor Durkan, Iain Murray, and Stefano Ermon.
\newblock Maximum likelihood training of score-based diffusion models.
\newblock {\em Advances in Neural Information Processing Systems}, 34:1415--1428, 2021.

\bibitem{song2020improved}
Yang Song and Stefano Ermon.
\newblock Improved techniques for training score-based generative models.
\newblock {\em Advances in neural information processing systems}, 33:12438--12448, 2020.

\bibitem{song2018pixeldefend}
Yang Song, Taesup Kim, Sebastian Nowozin, Stefano Ermon, and Nate Kushman.
\newblock Pixeldefend: Leveraging generative models to understand and defend against adversarial examples.
\newblock In {\em International Conference on Learning Representations}, 2018.

\bibitem{song2020score}
Yang Song, Jascha Sohl-Dickstein, Diederik~P Kingma, Abhishek Kumar, Stefano Ermon, and Ben Poole.
\newblock Score-based generative modeling through stochastic differential equations.
\newblock In {\em International Conference on Learning Representations}, 2021.

\bibitem{theis2022lossy}
Lucas Theis, Tim Salimans, Matthew~D Hoffman, and Fabian Mentzer.
\newblock Lossy compression with {Gaussian} diffusion.
\newblock {\em arXiv preprint arXiv:2206.08889}, 2022.

\bibitem{Theis2015ANO}
Lucas Theis, A{\"{a}}ron van~den Oord, and Matthias Bethge.
\newblock A note on the evaluation of generative models.
\newblock In Yoshua Bengio and Yann LeCun, editors, {\em 4th International Conference on Learning Representations, {ICLR} 2016, San Juan, Puerto Rico, May 2-4, 2016, Conference Track Proceedings}, 2016.

\bibitem{townsend2019practical}
J~Townsend, T~Bird, and D~Barber.
\newblock Practical lossless compression with latent variables using bits back coding.
\newblock In {\em 7th International Conference on Learning Representations, ICLR 2019}, volume~7. International Conference on Learning Representations (ICLR), 2019.

\bibitem{10.5555/2999792.2999855}
Benigno Uria, Iain Murray, and Hugo Larochelle.
\newblock Rnade: the real-valued neural autoregressive density-estimator.
\newblock In {\em Proceedings of the 27th International Conference on Neural Information Processing Systems - Volume 2}, NIPS'13, page 2175–2183, Red Hook, NY, USA, 2013. Curran Associates Inc.

\bibitem{van2016pixel}
A{\"a}ron Van Den~Oord, Nal Kalchbrenner, and Koray Kavukcuoglu.
\newblock Pixel recurrent neural networks.
\newblock In {\em International conference on machine learning}, pages 1747--1756. PMLR, 2016.

\bibitem{verdu2010mismatched}
Sergio Verd{\'u}.
\newblock Mismatched estimation and relative entropy.
\newblock {\em IEEE Transactions on Information Theory}, 56(8):3712--3720, 2010.

\bibitem{vincent2011connection}
Pascal Vincent.
\newblock A connection between score matching and denoising autoencoders.
\newblock {\em Neural computation}, 23(7):1661--1674, 2011.

\bibitem{wan2025enhancing}
Ben Wan, Tianyi Zheng, Zhaoyu Chen, and Jia Wang.
\newblock Enhancing the accuracy of generative adversarial networks with {Fokker--Planck} equations.
\newblock {\em Neurocomputing}, 638:130158, 2025.

\bibitem{wan2025pruning}
Ben Wan, Tianyi Zheng, Zhaoyu Chen, Yuxiao Wang, and Jia Wang.
\newblock Pruning for sparse diffusion models based on gradient flow.
\newblock In {\em ICASSP 2025-2025 IEEE International Conference on Acoustics, Speech and Signal Processing (ICASSP)}, pages 1--5. IEEE, 2025.

\bibitem{8007014}
Andre Wibisono, Varun Jog, and Po-Ling Loh.
\newblock Information and estimation in {Fokker-Planck} channels.
\newblock In {\em 2017 IEEE International Symposium on Information Theory (ISIT)}, pages 2673--2677, 2017.

\bibitem{wibisono2024optimal}
Andre Wibisono, Yihong Wu, and Kaylee~Yingxi Yang.
\newblock Optimal score estimation via empirical {Bayes} smoothing.
\newblock In {\em The Thirty Seventh Annual Conference on Learning Theory}, pages 4958--4991. PMLR, 2024.

\bibitem{10.5555/3495724.3497461}
Zhisheng Xiao, Qing Yan, and Yali Amit.
\newblock Likelihood regret: an out-of-distribution detection score for variational auto-encoder.
\newblock In {\em Proceedings of the 34th International Conference on Neural Information Processing Systems}, NIPS '20, Red Hook, NY, USA, 2020. Curran Associates Inc.

\bibitem{xu2022poisson}
Yilun Xu, Ziming Liu, Max Tegmark, and Tommi~S. Jaakkola.
\newblock Poisson flow generative models.
\newblock In Alice~H. Oh, Alekh Agarwal, Danielle Belgrave, and Kyunghyun Cho, editors, {\em Advances in Neural Information Processing Systems}, 2022.

\bibitem{yadin2024classification}
Shahar Yadin, Noam Elata, and Tomer Michaeli.
\newblock Classification diffusion models: Revitalizing density ratio estimation.
\newblock In {\em The Thirty-eighth Annual Conference on Neural Information Processing Systems}, 2024.

\bibitem{yang2024lossy}
Ruihan Yang and Stephan Mandt.
\newblock Lossy image compression with conditional diffusion models.
\newblock {\em Advances in Neural Information Processing Systems}, 36, 2024.

\bibitem{yoon2023scorebased}
Eunbi Yoon, Keehun Park, Sungwoong Kim, and Sungbin Lim.
\newblock Score-based generative models with {L\'evy} processes.
\newblock In {\em Thirty-seventh Conference on Neural Information Processing Systems}, 2023.

\bibitem{zheng2025direct}
Kaiwen Zheng, Yongxin Chen, Huayu Chen, Guande He, Ming-Yu Liu, Jun Zhu, and Qinsheng Zhang.
\newblock Direct discriminative optimization: Your likelihood-based visual generative model is secretly a {GAN} discriminator.
\newblock In {\em Forty-second International Conference on Machine Learning}, 2025.

\bibitem{zheng2023improved}
Kaiwen Zheng, Cheng Lu, Jianfei Chen, and Jun Zhu.
\newblock Improved techniques for maximum likelihood estimation for diffusion {ODEs}.
\newblock In {\em International Conference on Machine Learning}, pages 42363--42389. PMLR, 2023.

\bibitem{zheng2024non}
Tianyi Zheng, Cong Geng, Peng-Tao Jiang, Ben Wan, Hao Zhang, Jinwei Chen, Jia Wang, and Bo~Li.
\newblock Non-uniform timestep sampling: Towards faster diffusion model training.
\newblock In {\em Proceedings of the 32nd ACM International Conference on Multimedia}, pages 7036--7045, 2024.

\bibitem{zheng2025bidirectional}
Tianyi Zheng, Jiayang Zou, Peng-Tao Jiang, Hao Zhang, Jinwei Chen, Jia Wang, and Bo~Li.
\newblock Bidirectional {Beta}-tuned diffusion model.
\newblock {\em IEEE Transactions on Pattern Analysis and Machine Intelligence}, 2025.

\bibitem{zou2025convexity}
Jiayang Zou, Luyao Fan, Jiayang Gao, and Jia Wang.
\newblock Convexity of mutual information along the {Fokker-Planck} flow.
\newblock {\em arXiv preprint arXiv:2501.05094}, 2025.

\bibitem{zou2025revisit}
Jiayang Zou, Luyao Fan, Jiayang Gao, and Jia Wang.
\newblock A revisit to rate-distortion theory via optimal weak transport.
\newblock {\em arXiv preprint arXiv:2501.09362}, 2025.

\end{thebibliography}
\bibliographystyle{plain}

%%%%%%%%%%%%%%%%%%%%%%%%%%%%%%%%%%%%%%%%%%%%%%%%%%%%%%%%%%%%

%%%%%%%%%%%%%%%%%%%%%%%%%%%%%%%%%%%%%%%%%%%%%%%%%%%%%%%%%%%%

\newpage
\appendix
\clearpage
\section*{Contents of the Appendix}
\addcontentsline{toc}{section}{Contents of the Appendix}

\begingroup
  \renewcommand{\contentsname}{}  % 防止 \tableofcontents 插入“Contents”标题
  \tableofcontents
\endgroup

\newpage
\section{Preliminaries and Reviews}

We summarize the key notations and assumptions used in our theorems. The data distribution is \( p(\mathbf{x}) \), and the model is \( q(\hat{\mathbf{x}}; \boldsymbol{\theta}) \), where $\boldsymbol{\theta}$ is restricted to a parameter space $\mathbf{\Theta}$. The vector Gaussian channel follows \( \mathbf{Y}_t = \alpha_t \mathbf{X} + \sigma_t \mathbf{N} \), where \( \sigma_t: \mathbb{R} \to \mathbb{R} \) controls the time-dependent coefficient, \( t \in [0,1] \) represents the time horizon, and \( \mathbf{N} \sim \mathcal{N}(\mathbf{n}; 0, \mathbf{I}) \) is Gaussian noise. The input \( \mathbf{x} \sim p(\mathbf{x}) \), and the output \( \mathbf{y}_t \sim p(\mathbf{y}_t) \), with both as column vectors of appropriate dimensions.

\subsection{Notations}\label{notation}

In this paper, we are working on the Euclidean space $\mathbb{R}^D$ for some $D \geq 1$. We denote the $\ell_2$-inner product between vectors $\mathbf{u} = (u_1, \dots, u_d), \mathbf{v} = (v_1, \dots, v_d) \in \mathbb{R}^D$ as $\mathbf{u}^\top \mathbf{v} = \langle \mathbf{u}, \mathbf{v} \rangle = \sum_{i=1}^{D} u_i v_i$.

For a symmetric matrix $ \mathbf{A} \in \mathbb{R}^{D \times D}$, the notation $\mathbf{A} \succeq 0$ means $A$ is positive semidefinite, i.e., $\mathbf{u}^\top \mathbf{A} \mathbf{u} \geq 0$ for all $\mathbf{u} \in \mathbb{R}^R$. For symmetric matrices $\mathbf{A}, \mathbf{B} \in \mathbb{R}^{D \times D}$, the notation $\mathbf{A} \succeq \mathbf{B}$ means $\mathbf{A} - \mathbf{B} \succeq 0$ is positive semidefinite. Throughout, let $\mathbf{I} \in \mathbb{R}^{D \times D}$ denote the identity matrix.

For a differentiable function $f: \mathbb{R}^D \to \mathbb{R}$, let $\nabla f(\mathbf{x}) \in \mathbb{R}^D$ denote the gradient vector at $\mathbf{x} \in \mathbb{R}^D$ of the partial derivatives: $(\nabla f(\mathbf{x}))_i = \frac{\partial f(\mathbf{x})}{\partial \mathbf{x}_i}$. Let $\nabla^2 f(\mathbf{x}) \in \mathbb{R}^{D \times D}$ be the Hessian matrix of second partial derivatives: $(\nabla^2 f(\mathbf{x}))_{i,j} = \frac{\partial^2 f(\mathbf{x})}{\partial \mathbf{x}_i \partial \mathbf{x}_j}$. Let $\Delta f(\mathbf{x}) = \operatorname{Tr}(\nabla^2 f(\mathbf{x}))$ be the Laplacian. We use $\mathcal{C}$ to denote all continuous functions, and let $\mathcal{C}^k$ denote the family of functions with continuous $k$-th order derivatives.

For a vector field $\mathbf{v}: \mathbb{R}^D \to \mathbb{R}^D$ with $\mathbf{v}(\mathbf{x}) = (\mathbf{v}_1(\mathbf{x}), \dots, \mathbf{v}_d(\mathbf{x})) \in \mathbb{R}^D$, let $\nabla \mathbf{v}: \mathbb{R}^D \to \mathbb{R}^{D \times D}$ be the Jacobian matrix of mixed partial derivatives: $(\nabla \mathbf{v}(\mathbf{x}))_{i,j} = \frac{\partial \mathbf{v}_i(\mathbf{x})}{\partial \mathbf{x}_j}$. Let $\nabla \cdot \mathbf{v}: \mathbb{R}^D \to \mathbb{R}$ be the divergence of $\mathbf{v}$, defined by
\[
(\nabla \cdot \mathbf{v})(\mathbf{x}) = \sum_{i=1}^{D} \frac{\partial \mathbf{v}_i(\mathbf{x})}{\partial \mathbf{x}_i} = \operatorname{Tr}(\nabla \mathbf{v}(\mathbf{x})).
\]

Let \(V_r\subset \mathbb{R}^D\) be the region (an \(D\)-dimensional ball) bounded by the closed, piecewise-smooth, oriented surface \(S_r\), which is the \emph{\(D\)-sphere} of radius \(r\) centered at the origin. At any point \(\mathbf{y}\in S_r\), the symbol \(\mathbf{e}_{S_r}(\mathbf{y})\) denotes the outward-pointing unit normal vector to \(S_r\). Under the notation $d\mathbf{s}_{r} = \|d\mathbf{s}_{r}\|\mathbf{e}_{S_r}(\mathbf{y})$.

\subsection{Definitions}\label{definition}

Let $\mathcal{P}(\mathbb{R}^D)$ denote the space of probability distributions $\rho$ over $\mathbb{R}^D$ which are absolutely continuous with respect to the Lebesgue measure and have a finite second moment $\mathbb{E}_{\rho}[\|\mathbf{X}\|^2] < \infty$. We identify a probability distribution $\rho \in \mathcal{P}(\mathbb{R}^D)$ with its probability density function with respect to the Lebesgue measure, which we also denote by $\rho$: $\mathbb{R}^D \to \mathbb{R}$, so $\rho(x) > 0$ and $\int_{\mathbb{R}^D} \rho(x) dx = 1$. 

We say $\rho$ is absolutely continuous with respect to another distribution $\nu$, denoted by $\rho \ll \nu$, if $\nu(\mathbf{A}) = 0$ implies $\rho(\mathbf{A}) = 0$ for any $\mathbf{A} \subseteq \mathbb{R}^D$; if $\rho$ and $\nu$ both have density functions, then $\rho \ll \nu$ means $\nu(\mathbf{x}) = 0$ implies $\rho(\mathbf{x}) = 0$ for all $\mathbf{x} \in \mathbb{R}^D$.

\begin{definition}\label{def:entropy}
    Let $\mathcal{H}(\cdot):\mathcal{P}(\mathbb{R}^D) \rightarrow \mathbb{R}$ be the differential Shannon entropy:
    \begin{equation}
        \mathcal{H}(\rho) = - \mathbb{E}_{\rho}[\log \rho] = - \int_{\mathbb{R}^D} \rho(\mathbf{x}) \log \rho(\mathbf{x}) \,d\mathbf{x}.
    \end{equation}
\end{definition}

\begin{definition}\label{def:fisher}
    Let $\mathcal{J}(\cdot): \mathcal{P}(\mathbb{R}^D) \to \mathbb{R}$ be the \textit{Fisher information}:
    \begin{equation}\label{fisher information}
        \mathcal{J}(\rho) = \mathbb{E}_{\rho} \left[ \|\nabla \log \rho\|^2 \right] = -\mathbb{E}_{\rho} [\Delta \log \rho],
    \end{equation}
    and we define $\mathcal{J}(\rho) = +\infty$ if $\rho$ does not have a differentiable density. The second equality in the definition of $\mathcal{J}(\rho)$ above follows by integration by parts.
\end{definition}
Note that the expression in \eqref{fisher information} is a special case (with respect to a translation parameter) which does not involve an explicit parameter as in its most general definition.\footnote{The parameterized Fisher information matrix is defined with respect to a parameter \(\boldsymbol{\theta} \in \mathbf{\Theta}\) by \(\mathcal{J}(\rho) \equiv \mathbb{E}_{\rho}[\nabla_{\boldsymbol{\theta}}\log \rho(\mathbf{x};\boldsymbol{\theta})\,\nabla_{\boldsymbol{\theta}}^\top\log \rho(\mathbf{x};\boldsymbol{\theta})]\).}

\begin{definition}\label{KL definition}
    For probability distributions $\rho \ll \nu$ on $\mathbb{R}^D$, the \textit{Kullback-Leibler (KL) divergence} or the \textit{relative entropy} of $\rho$ with respect to $\nu$ is defined by:
    \begin{equation}
        D_{\text{KL}}(\rho \| \nu) = \mathbb{E}_{\rho} \left[ \log \frac{\rho}{\nu} \right] = \int_{\mathbb{R}^D} \rho(\mathbf{x}) \log \frac{\rho(\mathbf{x})}{\nu(\mathbf{x})} \,d\mathbf{x}.
    \end{equation}
\end{definition}

\begin{definition}\label{cross-entropy}
    For probability distributions $\rho \ll \nu$ on $\mathbb{R}^D$, the \textit{mismatched entropy} or the \textit{cross entropy} of $\rho$ with respect to $\nu$ is defined by:
    \begin{equation}
        \mathcal{H}(\rho,\nu) = - \mathbb{E}_{\rho}\left[ \log \nu \right] = \int_{\mathbb{R}^D} \rho(\mathbf{x}) \frac{1}{\log \nu(\mathbf{x})} d\mathbf{x}.
    \end{equation}
\end{definition}

\begin{definition}\label{RFI definition}
    If $\rho$ and $\nu$ have differentiable density functions, then the \textit{relative Fisher information} of $\rho$ with respect to $\nu$ is defined by:
    \begin{equation}
        I(\rho \| \nu) = \mathbb{E}_{\rho} \left[ \left\| \nabla \log \frac{\rho}{\nu} \right\|^2 \right] = \int_{\mathbb{R}^D} \rho(\mathbf{x}) \left\| \nabla \log \frac{\rho(\mathbf{x})}{\nu(\mathbf{x})} \right\|^2 \,d\mathbf{x}.
    \end{equation}
\end{definition}

We recall that the KL divergence, \( D_{\text{KL}}(\cdot\|\cdot) \), corresponds to the Bregman divergence of the negative entropy. Similarly, the relative Fisher information (RFI), \( I(\cdot\|\cdot) \), can be viewed as the Bregman divergence of the Fisher information. Moreover, half of the RFI is equivalent to score matching, as defined in \cite{hyvarinen2005estimation}.

\subsection{Standard Diffusion Models}\label{sampling process}

Consider a Gaussian diffusion process \cite{ho2020denoising}, which starts from clean data \( \mathbf{x} \) and defines a sequence of progressively noisier versions, denoted by channel outputs \( \mathbf{y}_t \), where \( t \) runs from 0 (least noisy) to 1 (most noisy). In the sampling process, given \( T \), we uniformly discretise the time interval into \( T \) timesteps, each of width \( 1/T \). Let \( t(i) = i/T \) denote the current timestep and \( s(i) = (i-1)/T \) the preceding one \cite{kingma2021variational}.

\paragraph{Forward Process.}
The forward process is defined by a conditional Gaussian distribution:
\[
p(\mathbf{y}_t|\mathbf{x}) = \mathcal{N}(\alpha_t\mathbf{x},\sigma_t^2\mathbf{I}),
\]
and, by the Markov property,
\[
p(\mathbf{y}_t|\mathbf{y}_s) = \mathcal{N}(\alpha_{t|s}\mathbf{y}_s, \sigma_{t|s}^2\mathbf{I}),
\]
where
\[
\alpha_{t|s} = \frac{\alpha_t}{\alpha_s}, \quad
\sigma_{t|s}^2 = \sigma_t^2 - \frac{\alpha_{t|s}^2}{\sigma_s^2}.
\]

\paragraph{Reverse Process.}
As established in prior work \cite{ho2020denoising,kingma2021variational,song2021denoising}, the conditional distribution \( p(\mathbf{y}_{s}|\mathbf{y}_{t}, \mathbf{x}) \) is also Gaussian:
\[
p(\mathbf{y}_{s}|\mathbf{y}_{t}, \mathbf{x}) = \mathcal{N}\left(
\frac{\alpha_{t|s}\sigma_{s}^2}{\sigma_t^2} \mathbf{y}_t +
\frac{\sigma_{t|s}^2\alpha_s}{\sigma_t^2} \mathbf{x},\ 
\frac{\sigma_{s}^2\sigma_{t|s}^2}{\sigma_t^2} \mathbf{I}
\right).
\]
As the ground truth \( \mathbf{x} \) is not available during the reverse process, it is approximated by a neural network \( \hat{\mathbf{x}}(\mathbf{y}_t, t;{\boldsymbol{\theta}}) \), parameterised by \( \boldsymbol{\theta} \). The learned reverse kernel becomes:
\[
p(\mathbf{y}_{s} | \mathbf{y}_{t}; \boldsymbol{\theta}) = \mathcal{N}\left(
\frac{\alpha_{t|s}\sigma_{s}^2}{\sigma_t^2} \mathbf{y}_t +
\frac{\sigma_{t|s}^2\alpha_s}{\sigma_t^2}\hat{\mathbf{x}}(\mathbf{y}_t, t;{\boldsymbol{\theta}}) ,\ 
\frac{\sigma_{s}^2\sigma_{t|s}^2}{\sigma_t^2} \mathbf{I}
\right).
\]

\subsection{Lemmas}\label{lemmas}
In this section, we introduce all Lemmas in order to prove the Theorems in the next section.

\begin{lemma}\label{arbitrary noise heat equation}

Let $\mathbf{X}\in \mathbb{R}^D$ be a random vector with density $p(\mathbf{x})$.  
Suppose $\boldsymbol{\Psi}\in \mathbb{R}^D$ is an arbitrary independent random vector with zero mean and identity matrix $\mathbf{I}$.  

Define $\tilde{\mathbf{X}}= \alpha_t\mathbf{X} + \sigma_t\,\boldsymbol{\Psi}$, and let $p(\tilde{\mathbf{x}})$ denote the density of $\tilde{\mathbf{X}}$. Then for every $\tilde{\mathbf{x}} \in \mathbb{R}^{D}$, as $\sigma^2_t \rightarrow 0^{+}$, we have:
\begin{equation}\label{17}
    \frac{d}{d\sigma_t^2} p(\tilde{\mathbf{x}})\bigg|_{\sigma^2_t = 0^{+}} = \frac{1}{2} \Delta_{} p(\mathbf{x}),
\end{equation}
where $\displaystyle \Delta_{\mathbf{x}}=\sum_{j=1}^D \frac{\partial^2}{\partial (x_{j})^2}$ is the usual Laplacian in $\mathbb{R}^D$.
\end{lemma}

Formula \eqref{17} allows the derivative w.r.t. the energy of the perturbation $\sigma_t^2$ to be transformed to the second derivative of the original pdf. In what followed we provide the proof for Lemma \ref{arbitrary noise heat equation} which is slightly different than that in \cite{narayanan2007thermodynamic}. Note that Lemma \ref{arbitrary noise heat equation} does not require the distribution of the perturbation to be symmetric as is required in \cite{narayanan2007thermodynamic}.

\begin{proof}
Let \( p(\tilde{\mathbf{x}}) \) denote the probability density function of the random vector \( \tilde{\mathbf{X}} \), and define its characteristic function as follows:
\begin{equation}
    \phi(\mathbf{k},\sigma_t^2) = \mathbb{E}\left[\exp\left(i\,\mathbf{k}^\top \tilde{\mathbf{X}} \right)\right], \quad \mathbf{k} \in \mathbb{R}^D.
\end{equation}

Given that \( \tilde{\mathbf{X}}  \) is defined by
\begin{equation}
    \tilde{\mathbf{X}} = \alpha_t \mathbf{X} + \sigma_t \boldsymbol{\Psi},
\end{equation}
where \( \mathbf{X} \) and \( \boldsymbol{\Psi} \) are independent, the characteristic function factorises as follows:
\begin{align}
    \phi(\mathbf{k}, \sigma_t^2) &= \mathbb{E} \left[ e^{i \mathbf{k}^\top (\alpha_t \mathbf{X} + \sigma_t \boldsymbol{\Psi})} \right] \\
    &= \mathbb{E} \left[ e^{i \mathbf{k}^\top \alpha_t \mathbf{X}} \right] \mathbb{E} \left[ e^{i \mathbf{k}^\top \sigma_t \boldsymbol{\Psi}} \right] \\
    &= \underbrace{\phi(\alpha_t \mathbf{k}, 0)}_{\text{characteristic function of } \mathbf{X}} \times \underbrace{\mathbb{E} \left[e^{i \mathbf{k}^\top \sigma_t \boldsymbol{\Psi}} \right]}_{\text{characteristic function of } \sigma_t \boldsymbol{\Psi}}.
\end{align}

Expanding the exponential function via Taylor’s theorem gives
\begin{equation}
    e^{i \mathbf{k}^\top \sigma_t \boldsymbol{\Psi}} = 1 + i \sigma_t (\mathbf{k}^\top \boldsymbol{\Psi}) - \frac{(\sigma_t \mathbf{k}^\top \boldsymbol{\Psi})^2}{2!} + O(\sigma_t^3).
\end{equation}
Since \( \mathbb{E}[\boldsymbol{\Psi}] = \mathbf{0} \) and \( \mathrm{Cov}[\boldsymbol{\Psi}] = \mathbf{I} \), it follows that
\begin{align}
    \mathbb{E}[\mathbf{k}^\top \boldsymbol{\Psi}] &= 0, \\
    \mathbb{E}[(\mathbf{k}^\top \boldsymbol{\Psi})^2] &= \mathbf{k}^\top \mathbb{E}[\boldsymbol{\Psi} \boldsymbol{\Psi}^\top] \mathbf{k} = \|\mathbf{k}\|^2.
\end{align}
Taking expectations, we obtain
\begin{equation}
    \mathbb{E} \left[ e^{i \mathbf{k}^\top \sigma_t \boldsymbol{\Psi}} \right] = 1 - \frac{\sigma_t^2}{2} \|\mathbf{k}\|^2 + o(\sigma_t^2).
\end{equation}
Thus, the characteristic function satisfies
\begin{equation}
    \phi(\mathbf{k}, \sigma_t^2) = \phi(\alpha_t \mathbf{k}, 0) \left( 1 - \frac{\sigma_t^2}{2} \|\mathbf{k}\|^2 + o(\sigma_t^2) \right).
\end{equation}

Recalling the inverse Fourier transform, the probability density function is given by
\begin{equation}
    p_{\sigma_t^2}(\tilde{\mathbf{x}}) = \frac{1}{(2\pi)^D} \int_{\mathbb{R}^D} e^{-i\mathbf{k}^\top \tilde{\mathbf{x}}} \phi(\mathbf{k}, \sigma_t^2) d\mathbf{k}.
\end{equation}
Substituting the expansion of \( \phi(\mathbf{k}, \sigma_t^2) \), we obtain
\begin{equation}
    p_{\sigma_t^2}(\tilde{\mathbf{x}}) = \frac{1}{(2\pi)^D} \int_{\mathbb{R}^D} e^{-i\mathbf{k}^\top \tilde{\mathbf{x}}} \phi(\alpha_t \mathbf{k}, 0) \left( 1 - \frac{\sigma_t^2}{2} \|\mathbf{k}\|^2 + o(\sigma_t^2) \right) d\mathbf{k}.
\end{equation}

Since
\begin{equation}
    p_0(\tilde{\mathbf{x}}) = \frac{1}{(2\pi)^D} \int_{\mathbb{R}^D} e^{-i\mathbf{k}^\top \tilde{\mathbf{x}}} \phi(\alpha_t \mathbf{k},0) d\mathbf{k},
\end{equation}
it follows that the difference \( p_{\sigma_t^2}(\tilde{\mathbf{x}}) - p_0(\tilde{\mathbf{x}}) \) corresponds to multiplying \( \phi(\alpha_t \mathbf{k}, 0) \) by 
\[
-\frac{\sigma_t^2}{2} \|\mathbf{k}\|^2 + o(\sigma_t^2).
\]

A standard result from Fourier analysis states that multiplication by \( -\|\mathbf{k}\|^2 \) in the Fourier domain corresponds to the application of the Laplacian \( \Delta_{} \) in the spatial domain, defined as
\[
    \Delta_{\mathbf{x}} = \sum_{j=1}^{D} \frac{\partial^2}{\partial x_{j}^2}.
\]

Thus, for small \( \sigma_t^2 \), we obtain
\begin{equation}
    p_{\sigma_t^2}(\tilde{\mathbf{x}}) = p_0(\tilde{\mathbf{x}}) + \frac{\sigma_t^2}{2} \Delta_{\tilde{\mathbf{x}}} p_0(\tilde{\mathbf{x}}) + o(\sigma_t^2).
\end{equation}
Since \( p_0(\tilde{\mathbf{x}}) \) corresponds to \( p(\mathbf{x}) \) in the absence of noise, this completes the proof of Lemma \ref{arbitrary noise heat equation}.
\end{proof}

\begin{lemma}[Vanishing boundary flux]\label{integration vanish}
For an arbitrary input distribution $p(\mathbf{x})$, an assumed input distribution $q(\hat{\mathbf{x}};\boldsymbol{\theta})$ and a vector Gaussian channel $p(\mathbf{y}_t\vert\mathbf{x})=\mathcal{N}(\alpha_t\mathbf{x},\sigma_t^2\mathbf{I})$ in \eqref{channel eq}, denote the corresponding output densities by $p(\mathbf{y}_t)$ and $q(\mathbf{y}_t;\boldsymbol{\theta})$ at time $t\in[0,1]$.
For $r>0$, let $B_r:=\{\mathbf{y}\in\mathbb{R}^D:\|\mathbf{y}\|\le r\}$ be the closed $D$-ball and let $S_r:=\partial B_r$ be the $(D\!-\!1)$-sphere with outward unit normal $\mathbf{e}_{S_{r}}(\mathbf{y})=\mathbf{y}/\|\mathbf{y}\|$. 
Assume $p(\mathbf{y}_t)$ and $q(\mathbf{y}_t;\boldsymbol{\theta})$ are $\mathcal{C}^1$ with well-defined $\nabla_{\mathbf{y}}$ and $\Delta_{\mathbf{y}}$, and that $\mathbb{E}_{p}\!\big[\log q(\mathbf{y}_t;\boldsymbol{\theta})\big]<\infty$. 
Then the boundary integrals
\begin{align}
L_1 &:= \lim_{r\to\infty}\int_{S_r} \log q(\mathbf{y}_t;\boldsymbol{\theta})\, \nabla_{\mathbf{y}_t} p(\mathbf{y}_t)\cdot \mathbf{e}_{S_{r}}(\mathbf{y}_t)\, \mathrm{d}S = 0, \label{eq:L1}\\
L_2 &:= \lim_{r\to\infty}\int_{S_r} \frac{1}{2}\,\frac{p(\mathbf{y}_t)}{q(\mathbf{y}_t;\boldsymbol{\theta})}\, \nabla_{\mathbf{y}_t} q(\mathbf{y}_t;\boldsymbol{\theta})\cdot \mathbf{e}_{S_{r}}(\mathbf{y}_t)\, \mathrm{d}S = 0. \label{eq:L2}
\end{align}
\end{lemma}

\begin{proof}
First, note that the limits $L_{1}$ and $L_{2}$ obviously exist, as each of these limits can be expressed as the sum of two converging integrals. 
We prove \eqref{eq:L1}; the argument for \eqref{eq:L2} is analogous. 
For brevity, omit the subscript $t$ and write $\mathbf{y}$ for $\mathbf{y}_t$. Let $\mathbf{e}_{S_{r}}(\mathbf{y})$ to be unit vector normal to $S_r$ at the point $\mathbf{y}$. Under this notation $d{\mathbf{s_r}} = ||d{\mathbf{s_r}}||\mathbf{e}_{S_{r}}(\mathbf{y})$. We integrate over $r \geq 0$ the surface integral in \eqref{eq:L1} and apply Green's identity to find the relations. Set
\[
f(r):=\int_{S_r} \log q(\mathbf{y})\,\nabla p(\mathbf{y})\cdot d{\mathbf{s_r}} .
\]
By the coarea formula (spherical coordinates), 
\begin{equation}
    \int_0^\infty f(r)\,\mathrm{d}r 
= \int_{\mathbb{R}^D}  \nabla p(\mathbf{y})\cdot \left(\log q(\mathbf{y})\,\mathbf{e}_{S_{r}}(\mathbf{y})\right) \,\mathrm{d}\mathbf{y}.
\end{equation}
% $\nabla\!\cdot\!\big(p\,\log q\,\mathbf{n}\big)=\nabla p\cdot(\log q\,\mathbf{n})+p\,\nabla\!\cdot\!\big(\log q\,\mathbf{n}\big)$
Using the product rule and the divergence theorem on $B_r$,
\begin{align}
&\int_{B_r} \nabla p(\mathbf{y})\cdot \big(\log q(\mathbf{y})\,\mathbf{e}_{S_{r}}(\mathbf{y})\big)\,\mathrm{d}\mathbf{y} \notag \\
&= \lim_{r \rightarrow \infty}\int_{S_r} p(\mathbf{y})\log q(\mathbf{y})\,\mathbf{e}_{S_{r}}(\mathbf{y})  \cdot d\mathbf{s_r}
\;-\; \int_{B_r} p(\mathbf{y})\, \nabla\!\cdot\!\big(\log q(\mathbf{y})\,\mathbf{e}_{S_{r}}(\mathbf{y})\big)\,\mathrm{d}\mathbf{y}. \label{eq:div-ibp}
\end{align}
Letting $r\to\infty$ we bound the two terms on the right-hand side.

\paragraph{First term.}
By the coarea formula,
\[
\int_r^\infty \int_{S_r} \big|p(\mathbf{y})\log q(\mathbf{y})\big|\,\mathrm{d}S\,\mathrm{d}r
= \int_{\|\mathbf{y}\|>r} \big|p(\mathbf{y})\log q(\mathbf{y})\big|\,\mathrm{d}\mathbf{y}
\xrightarrow[r\to\infty]{} 0,
\]
because $\mathbb{E}_{p}[|\log q(\mathbf{y})|]<\infty$. Hence $\int_{S_r} p(\mathbf{y})\log q(\mathbf{y})\,\mathrm{d}S\to 0$ along the full sequence $r\to\infty$ (e.g., via a Cesàro argument).

\paragraph{Second term.}
Now we note that the absolute value of the divergence in the second term satisfies the relation
\begin{align}
\left| \nabla \cdot \left(\log q(\mathbf{y})\mathbf{e}_{S_{r}}(\mathbf{y})\right) \right| &= \frac{\left| \nabla q(\mathbf{y}) \cdot \mathbf{e}_{S_{r}}(\mathbf{y}) \right|}{q(\mathbf{y})} \notag \\
&\leq \frac{\left\| \nabla q(\mathbf{y}) \right\|}{q(\mathbf{y})}  \label{eq:eq1}
\end{align}
Hence, we have
\begin{align}
    \frac{\|\nabla q(\mathbf{y})\|}{q(\mathbf{y})}
    &= \left( \sum_{i=1}^{n} \left[ \int_{\mathbb{R}^{D}} \frac{q(\mathbf{x})}{q(\mathbf{y})} (2\pi \sigma^2_{t})^{-\frac{D}{2}}\bigl(\frac{\|y_{i}-\alpha_tx_{i}\|}{\sigma_t^2}\bigr) \exp\! \ \bigl(-\frac{\|\mathbf{y}_t-\alpha_t\mathbf{x}\|^{2}}{2\sigma_t^2}\bigr) \,d\mathbf{x}. \right]^2 \right)^{\frac{1}{2}} \notag \\
    &= \left(\sum_{i=1}^{n} \left[\mathbb{E}\left(\frac{Y_{i}-\alpha_tX_{i}}{\sigma_t^2} \Big| \mathbf{Y} = \mathbf{y} \right) \right]^2\right)^{\frac{1}{2}} \notag \\
    &\leq \left(\sum_{i=1}^{n} \left(\mathbb{E} \left(\left(\frac{Y_{i} - \alpha_tX_{i}}{\sigma_t^2}\right)^2 \Big| \mathbf{Y} = \mathbf{y} \right) \right)\right)^{\frac{1}{2}} \notag \\
    &= \left(\mathbb{E}\left( \left(\frac{\|\mathbf{Y} - \alpha_t\mathbf{X}\|}{\sigma_t^2}\right)^2 \Big| \mathbf{Y} = \mathbf{y} \right) \right)^{\frac{1}{2}}.
\end{align}
Integrating w.r.t.\ $p(\mathbf{y})\,\mathrm{d}\mathbf{y}$ and applying Jensen, we can write the chain of inequalities
\begin{align}
    \int_{\mathbb{R}^n} p(\mathbf{y}) \frac{\|\nabla q(\mathbf{y})\|}{q(\mathbf{y})} d\mathbf{y} &\leq \mathbb{E}_{p} \left \{ \left(\mathbb{E}\left( \left(\frac{\|\mathbf{Y} - \alpha_t\mathbf{X}\|}{\sigma_t^2}\right)^2 \Big| \mathbf{Y} = \mathbf{y} \right) \right)^{\frac{1}{2}} \right\} \notag \\
    &\leq \left \{\mathbb{E}_{p}  \left(\mathbb{E}\left( \left(\frac{\|\mathbf{Y} - \alpha_t\mathbf{X}\|}{\sigma_t^2}\right)^2 \Big| \mathbf{Y} = \mathbf{y} \right) \right) \right\}^{\frac{1}{2}} \notag \\
    &= \left[\mathbb{E}_{p}\left(\left(\frac{\|\mathbf{N}\|}{\sigma_t^2}\right)^2\right)\right]^{\frac{1}{2}} \notag \\
    &< \infty,
\end{align}
which means that
\[
\mathbb{E}_{p}\big[\|\nabla \log q(\mathbf{y})\|\big]
\le \Big(\mathbb{E}_{p}\big[\|\mathbf{N}\|^{2}/\sigma_t^{4}\big]\Big)^{1/2}<\infty.
\]
%For the second expectation, on $\{\|\mathbf{y}\|>1\}$ we have $\|\mathbf{y}\|^{-1}\le 1$, hence it is bounded by $(D-1)\,\mathbb{E}_{p}[|\log q(\mathbf{Y})|]<\infty$.
Therefore the right-hand side of \eqref{eq:div-ibp} is finite, so
\[
\int_0^\infty f(r)\,\mathrm{d}r < \infty.
\]

Finally, $f$ is locally absolutely continuous in $r$ (by the smoothness of $p$ and $q$ and the coarea formula), hence $\lim_{r\to\infty} f(r)$ exists. Since $\int_0^\infty f(r)\,\mathrm{d}r<\infty$, this limit must equal $0$. Thus $L_1=0$. The proof of $L_2=0$ follows by the same steps with $\tfrac{1}{2}\,\frac{p}{q}\nabla q$ in place of $\log q\,\nabla p$.
\end{proof}

\begin{lemma}\label{lemma 4}
Fix $t\in[0,1]$ and consider the Gaussian channel in~\eqref{channel eq}. 
Let $p(\mathbf{y}_t)$ and $q(\mathbf{y}_t;\boldsymbol{\theta})$ be the output densities induced by inputs $p(\mathbf{x})$ and $q(\mathbf{x};\boldsymbol{\theta})$, respectively. 
Assume $p(\cdot)$ and $q(\cdot;\boldsymbol{\theta})$ are sufficiently smooth and integrable so that differentiation under the integral sign is justified and the boundary terms in Lemma~\ref{integration vanish} vanish. 
Then the derivative of the mismatched output cross-entropy
\[
\mathcal{H}\big(p(\mathbf{y}_t),q(\mathbf{y}_t;\boldsymbol{\theta})\big)
:= -\!\int_{\mathbb{R}^D} p(\mathbf{y}_t)\,\log q(\mathbf{y}_t;\boldsymbol{\theta})\,\mathrm{d}\mathbf{y}_t
\]
with respect to the noise variance $\sigma_t^2$ is
\begin{equation}\label{eq:lemma-variance-deriv}
\frac{\mathrm{d}}{\mathrm{d}\sigma_t^2}\,\mathcal{H}\big(p(\mathbf{y}_t),q(\mathbf{y}_t;\boldsymbol{\theta})\big)
= \frac{1}{2}\int_{\mathbb{R}^{D}} \left(
\frac{\nabla_{} q(\mathbf{y}_t;\boldsymbol{\theta}) \cdot \nabla_{} p(\mathbf{y}_t)}{q(\mathbf{y}_t;\boldsymbol{\theta})}
\;+\;
\nabla_{}\!\left(\frac{p(\mathbf{y}_t)}{q(\mathbf{y}_t;\boldsymbol{\theta})}\right)
\cdot \nabla_{} q(\mathbf{y}_t;\boldsymbol{\theta})
\right)\mathrm{d}\mathbf{y}_t,
\end{equation}
where $\mathbf{a}\cdot\mathbf{b}=\langle \mathbf{a},\mathbf{b}\rangle$ denotes the Euclidean inner product.
\end{lemma}

\begin{proof}
By definition of cross-entropy in Def.~\ref{cross-entropy},
\[
\mathcal{H}\big(p(\mathbf{y}_t),q(\mathbf{y}_t;\boldsymbol{\theta})\big)
= -\int_{\mathbb{R}^D} p(\mathbf{y}_t)\,\log q(\mathbf{y}_t;\boldsymbol{\theta})\,\mathrm{d}\mathbf{y}_t.
\]
Differentiating w.r.t.\ $\sigma_t^2$ (justified by the assumed regularity) yields
\begin{align}\label{eq:diff-start}
\frac{\mathrm{d}}{\mathrm{d}\sigma_t^2}\Big(-\!\int p\log q\Big)
= -\int \log q\,\frac{\partial p}{\partial \sigma_t^2}\,\mathrm{d}\mathbf{y}_t
\;-\;\int p\,\frac{\partial}{\partial \sigma_t^2}\log q\,\mathrm{d}\mathbf{y}_t,
\end{align}
where we abbreviate $p=p(\mathbf{y}_t)$ and $q=q(\mathbf{y}_t;\boldsymbol{\theta})$.

Under Gaussian smoothing (heat equation),
\[
\frac{\partial p}{\partial \sigma_t^2}=\frac{1}{2}\,\Delta_{\mathbf{y}_t}p,
\qquad
\frac{\partial q}{\partial \sigma_t^2}=\frac{1}{2}\,\Delta_{\mathbf{y}_t}q,
\qquad
\frac{\partial}{\partial \sigma_t^2}\log q
=\frac{1}{q}\frac{\partial q}{\partial \sigma_t^2}
=\frac{1}{2}\,\frac{\Delta_{\mathbf{y}_t} q}{q}.
\]
We now recall Green's identity: If $\boldsymbol{\phi}(\mathbf{x})$ and $\boldsymbol{\psi}(\mathbf{x})$ are twice continuously differentiable functions in $\mathbb{R}^D$ and $V$ is any set bounded by a piecewise smooth, closed, oriented surface $S$ in $\mathbb{R}^D$, then

\begin{equation}
\int_V \boldsymbol{\phi} \nabla^2 \boldsymbol{\psi} \, dV = \int_S \boldsymbol{\phi} \nabla \boldsymbol{\psi} \cdot d\mathbf{s} - \int_V \nabla \boldsymbol{\phi} \cdot \nabla \boldsymbol{\psi} \, dV, \label{green identity}
\end{equation}

where $\nabla \boldsymbol{\phi}$ denotes the gradient of $\boldsymbol{\phi}$, $d\mathbf{s}$ denotes the elementary area vector, and $\nabla \boldsymbol{\phi} \cdot d\mathbf{s}$ is the inner product of these two vectors. This identity plays the role of integration by parts in $\mathbb{R}^D$.

To apply Green’s identity to \eqref{eq:diff-start}, we let $V_r$ be the $D$-sphere of radius $r$ centered at the origin and having surface $S_r$. Then we use Green’s identity on $V_r$ and $S_r$ with $\boldsymbol{\phi}(\mathbf{y}_t) = \log p(\mathbf{y}_t)$ and $\boldsymbol{\psi}(\mathbf{y}_t) = p(\mathbf{y}_t)$ and take the limit as $r \to \infty$. In Lemma~\ref{integration vanish} the surface integral over $S_r$ is shown to vanish in the limit. Hence, by applying heat equation \cite{1057105} and Green's identity for the first term on right hide side, we obtain:
\begin{align}
    -\int \log q(\mathbf{y}_t;\boldsymbol{\theta}) \frac{d}{d\sigma_t^2} p(\mathbf{y}_t) \, d\mathbf{y}_t &= - \frac{1}{2} \int \log q(\mathbf{y}_t;\boldsymbol{\theta}) \Delta_{\mathbf{y}_t} p(\mathbf{y}_t) \, d\mathbf{y}_t \\
    &= \frac{1}{2} \int \nabla_{\mathbf{y}_t} \log q(\mathbf{y}_t;\boldsymbol{\theta}) \cdot \nabla_{\mathbf{y}_t} p(\mathbf{y}_t) \, d\mathbf{y}_t \\
    &= \frac{1}{2}\int_{\mathbb{R}^{D}} \frac{\nabla_{} q(\mathbf{y}_t;\boldsymbol{\theta}) \cdot \nabla_{} p(\mathbf{y}_t)}{q(\mathbf{y}_t;\boldsymbol{\theta})} \, d\mathbf{y}_t,
\end{align}
where we used integration by parts in high dimension and the vanishing boundary terms from Lemma~\ref{integration vanish}. For the second term,
\begin{align}
    - \int  p(\mathbf{y}_t) \frac{d}{d\sigma_t^2} \log q(\mathbf{y}_t;\boldsymbol{\theta}) \, d\mathbf{y}_t&= - \int p(\mathbf{y}_t) \frac{\nabla_{\sigma_t^2}q(\mathbf{y}_t;\boldsymbol{\theta})}{q(\mathbf{y}_t;\boldsymbol{\theta})}\,d\mathbf{y}_t \\ 
    %&= - \frac{1}{2} \int p(\mathbf{y}_t) \Delta_{\mathbf{y}_t} \log q(\mathbf{y}_t;\boldsymbol{\theta}) \, d\mathbf{y}_t \\
    &= - \frac{1}{2} \int \left(\frac{p(\mathbf{y}_t)}{q(\mathbf{y}_t;\boldsymbol{\theta})}\right) \Delta_{\mathbf{y}_t} q(\mathbf{y}_t;\boldsymbol{\theta}) \, d\mathbf{y}_t \\
    &= \frac{1}{2}\int_{\mathbb{R}^{D}} \nabla_{}\left(\frac{p(\mathbf{y}_t)}{q(\mathbf{y}_t)}\right) \cdot \nabla_{} q(\mathbf{y}_t) \, d\mathbf{y}_t,
\end{align}

where we again applied Green’s identity and used the vanishing of the boundary flux.

Summing the two contributions yields~\eqref{eq:lemma-variance-deriv}.
\end{proof}

\begin{lemma}\label{KLinvariance}
    Let $(X, \mathcal{S}, \mu_i)$ be probability spaces, and let $T: X \to Y$ be a measurable transformation inducing probability measures $\nu_i$ on $(Y, \mathcal{T})$ such that
\[
\nu_i(G) = \mu_i(T^{-1}(G)), \quad \forall G \in \mathcal{T}, \quad i=1,2.
\]
Denote the Radon-Nikodym derivatives of $\nu_i$ with respect to a common reference measure $\gamma$ as
\[
g_i(y) = \frac{d\nu_i}{d\gamma}.
\]
Then, the KL divergence remains invariant under the transformation $T$, i.e.,
\[
D_{\mathrm{KL}}(\nu_1 || \nu_2) = D_{\mathrm{KL}}(\mu_1 || \mu_2),
\]
where
\[
D_{\mathrm{KL}}(\nu_1 || \nu_2) = \int_Y g_1(y) \log \frac{g_1(y)}{g_2(y)} d\gamma(y).
\]
\end{lemma}

\begin{proof}
    Let $\lambda$ be a reference measure on $X$ such that $\mu_i$ has densities $f_i$ with respect to $\lambda$, i.e., 
\[
f_i(x) = \frac{d\mu_i}{d\lambda}(x).
\]
By the change of variables under $T$, the density functions transform as
\[
g_i(y) = \frac{d\nu_i}{d\gamma} (y) = \frac{d\mu_i}{d\lambda}(T^{-1}(y)) J_T^{-1}(y),
\]
where $J_T(y) = \left| \det \frac{dT}{dx} \right|$ is the Jacobian determinant of $T$.

Substituting into the definition of KL divergence:
\[
D_{\mathrm{KL}}(\nu_1 || \nu_2) = \int_Y g_1(y) \log \frac{g_1(y)}{g_2(y)} d\gamma(y),
\]
we expand:
\[
= \int_Y \left( f_1(T^{-1}(y)) J_T^{-1} \right) \log \frac{f_1(T^{-1}(y)) J_T^{-1}}{f_2(T^{-1}(y)) J_T^{-1}} d\gamma(y).
\]
Since the Jacobian terms cancel, we obtain:
\[
= \int_Y f_1(T^{-1}(y)) \log \frac{f_1(T^{-1}(y))}{f_2(T^{-1}(y))} d\gamma(y).
\]
By the change of variables $z = T^{-1}(y)$, we rewrite this as:
\[
= \int_X f_1(z) \log \frac{f_1(z)}{f_2(z)} d\lambda(z) = D_{\mathrm{KL}}(\mu_1 || \mu_2).
\]
Thus, the KL divergence is invariant under the transformation $T$, completing the proof. 
\end{proof}

\subsection{Proof and Remark of Theorems}\label{proof}

\subsubsection{Proof of Theorem.~\ref{theorem 1}}
\textbf{Theorem ~\ref{theorem 1}.}\label{proof:theorem1} \textit{Let \( \mathbf{X} \sim p(\mathbf{x}) \) be an arbitrary distributed random vector on \( \mathbb{R}^D \), and let \( q(\hat{\mathbf{x}}; \boldsymbol{\theta}) \) be a parametric model with \( \hat{\mathbf{X}} \sim q(\hat{\mathbf{x}}; \boldsymbol{\theta}) \). Define the perturbed observation
\[
\tilde{\mathbf{X}} := \alpha_t \mathbf{X} + \sigma_t \boldsymbol{\Psi},
\]
where \( \boldsymbol{\Psi} \) is a random vector independent of $\mathbf{X}$ , satisfying \( \mathbb{E}[\boldsymbol{\Psi}] = 0 \) and \( \mathrm{Cov}(\boldsymbol{\Psi}) = \mathbf{I} \). Let \( p_{\sigma_t^2} \) and \( q_{\sigma_t^2} \) denote the densities of $\tilde{\mathbf{X}}$ under \( p \) and \( q \) with noise variance $\sigma_t^2$, respectively. Suppose that the KL divergence \( D_{\mathrm{KL}}(p_{\sigma_t^2} \| q_{\sigma_t^2}) \) is finite for sufficiently small \( \sigma_t^2 \). Then the following limit holds:
\begin{equation}
    \left. \frac{d}{d\sigma_t^2} D_{\mathrm{KL}}(p_{\sigma_t^2} \| q_{\sigma_t^2}) \right|_{\sigma_t^2 \to 0^+}
= -\frac{1}{2} \int_{\mathbb{R}^D} p(\mathbf{x}) \left\| \nabla \log p(\mathbf{x}) - \nabla \log q(\hat{\mathbf{x}}; \boldsymbol{\theta}) \right\|^2 d\mathbf{x},
\end{equation}
i.e.,
\[
\left. \frac{d}{d\sigma_t^2} D_{\mathrm{KL}}(p(\tilde{\mathbf{x}}) \| q(\tilde{\mathbf{x}}; \boldsymbol{\theta})) \right|_{\sigma_t^2 \to 0^+}
= -\frac{1}{2} I\left(p(\mathbf{x}) \| q(\hat{\mathbf{x}};\boldsymbol{\theta})\right),
\]
where \( I(\cdot\|\cdot) \) denotes the Fisher divergence (equivalently, score matching objective) between \( p \) and \( q \).}

\begin{proof}
    According to the definition of relative entropy (See Def.~\ref{KL definition}), we have
    \begin{align}
        &\frac{d}{d\sigma_t^2}D_{\text{KL}}(p(\tilde{\mathbf{x}})\|q(\tilde{\mathbf{x}};{\boldsymbol{\theta}})) = \frac{d}{d\sigma_t^2}\int_{\mathbb{R}^D} p(\tilde{\mathbf{x}}) \log(\frac{p(\tilde{\mathbf{x}})}{q(\tilde{\mathbf{x}};\boldsymbol{\theta})}) \,d\tilde{\mathbf{x}} \\
        &= \int \frac{\partial p(\tilde{\mathbf{x}})}{\partial\sigma_t^2}  \log(\frac{p(\tilde{\mathbf{x}})}{q(\tilde{\mathbf{x}};\boldsymbol{\theta})}) + p(\tilde{\mathbf{x}}) \frac{\partial}{\partial\sigma_t^2} \log(\frac{p(\tilde{\mathbf{x}})}{q(\tilde{\mathbf{x}};\boldsymbol{\theta})}) \,d\tilde{\mathbf{x}} \\
        &= \int\frac{\partial p(\tilde{\mathbf{x}})}{\partial \sigma_t^2}\log(\frac{p(\tilde{\mathbf{x}})}{q(\tilde{\mathbf{x}};\boldsymbol{\theta})}) + p(\tilde{\mathbf{x}})\left(\frac{\partial p(\tilde{\mathbf{x}})}{\partial\sigma_t^2}/p(\tilde{\mathbf{x}}) - \frac{\partial q(\tilde{\mathbf{x}};\boldsymbol{\theta})}{\partial\sigma_t^2}/q(\tilde{\mathbf{x}};\boldsymbol{\theta})\right)    \,d\tilde{\mathbf{x}}. \label{21}
    \end{align}
    Invoking Lemma \ref{arbitrary noise heat equation} on \eqref{21} yields
    \begin{equation}
        \frac{d}{d\sigma_t^2}D_{\text{KL}}(p(\tilde{\mathbf{x}})\|q(\tilde{\mathbf{x}};\boldsymbol{\theta}))\bigg|_{\sigma_t^2 = 0^{+}} = \frac{1}{2} \int \left(p''(\tilde{\mathbf{x}})\log (\frac{p(\tilde{\mathbf{x}})}{q(\tilde{\mathbf{x}};\boldsymbol{\theta})}) + p''(\tilde{\mathbf{x}}) - \frac{q''(\tilde{\mathbf{x}};\boldsymbol{\theta})p(\tilde{\mathbf{x}})}{q(\tilde{\mathbf{x}};\boldsymbol{\theta})} \right)\,d\tilde{\mathbf{x}}. \label{22}
    \end{equation}
    For convenience, define:
    \[\upsilon(\tilde{\mathbf{x}}) = \log \frac{p(\tilde{\mathbf{x}})}{q(\tilde{\mathbf{x}};\boldsymbol{\theta})}= \log p(\tilde{\mathbf{x}}) - \log q(\tilde{\mathbf{x}};\boldsymbol{\theta})\]
    Recall the Green's identity:
    \begin{equation}
        \int_{\mathbb{R}^D} \nu\nabla^2 \mu\ d\tilde{\mathbf{x}} = \int_{\partial\mathbb{R}^D} \nu \frac{\partial\mu}{\partial n} dS - \int_{\mathbb{R}^D} \left(\nabla\nu \cdot \nabla\mu \right)\,d\tilde{\mathbf{x}}.
    \end{equation}
    Setting $\nu = p$ and $\mu = \upsilon = \log \frac{p}{q}$, we obtain:
    \begin{equation}
        \int_{\mathbb{R}^D}p''\upsilon  \ \, d  \tilde{\mathbf{x}} = \int_{\partial\mathbb{R}^D} p\frac{\partial\upsilon}{\partial n} dS - \int_{\mathbb{R}^D} \left(\nabla\upsilon \cdot \nabla p \right)\ d\tilde{\mathbf{x}}.
    \end{equation}
    Since $p(\tilde{\mathbf{x}})$ has finite differential Shannon entropy, the boundary integral vanishes, leaving
    \begin{equation}
        \int_{\mathbb{R}^D} p'' \upsilon \ d\tilde{\mathbf{x}} = - \int_{\mathbb{R}^D} \left( \nabla \upsilon \cdot \nabla p \right) \ d\tilde{\mathbf{x}}.
    \end{equation}
    For the term $q''p/q$, we apply the same identity:
    \begin{equation}
        \int_{\mathbb{R}^D} \frac{q''p}{q} d\tilde{\mathbf{x}} = - \int_{\mathbb{R}^D} \left(p \nabla \log q \cdot \nabla \log q \right)\ d\tilde{\mathbf{x}}.
    \end{equation}
    Since $\int_{\mathbb{R}^D} p''(\tilde{\mathbf{x}}) d\tilde{\mathbf{x}} = 0 $ under appropriate boundary conditions, it does not contribute. Similar technique was used in \cite{chen2013mismatched,palomar2005gradient,5165186}.
    Using \( \upsilon = \log p - \log q \), we compute:
    \begin{equation}
        \nabla \upsilon = \nabla \log p - \nabla \log q.
    \end{equation}
    Thus, 
    \begin{equation}
        \nabla p \cdot \nabla \upsilon = \nabla p \cdot (\nabla \log p - \nabla \log q).
    \end{equation}
Since \( \nabla p = p \nabla \log p \), we substitute:
\begin{equation}
    \nabla p \cdot \nabla \upsilon = p \nabla \log p \cdot (\nabla \log p - \nabla \log q).
\end{equation}
Expanding,
\begin{equation}
    \nabla \log p \cdot (\nabla \log p - \nabla \log q) = \|\nabla \log p\|^2 - \nabla \log p \cdot \nabla \log q.
\end{equation}
Thus,
\begin{equation}
    \int_{\mathbb{R}^D} p'' \upsilon \, d\tilde{\mathbf{x}} = - \int_{\mathbb{R}^D} p \left( \|\nabla \log p\|^2 - \nabla \log p \cdot \nabla \log q \right) d\tilde{\mathbf{x}}.
\end{equation}
Using the result for \( q'' p / q \):
\begin{equation}
    \int_{\mathbb{R}^D} \frac{q'' p}{q} d\tilde{\mathbf{x}} = -\int_{\mathbb{R}^D} \left(p \nabla \log q \cdot \nabla \log q \, \right)d\tilde{\mathbf{x}}.
\end{equation}
Summing both contributions,
\begin{equation}
    \frac{1}{2} \int_{\mathbb{R}^D} \left(- p \left( \|\nabla \log p\|^2 - \nabla \log p \cdot \nabla \log q \right) - p \|\nabla \log q\|^2 + p \nabla \log p \cdot \nabla \log q \right) \,d\tilde{\mathbf{x}}.
\end{equation}
Rearranging, we conclude:
\begin{equation}
    \frac{d}{d\sigma_t^2}D_{\text{KL}}(p(\tilde{\mathbf{x}})\|q(\tilde{\mathbf{x}};\boldsymbol{\theta}))\bigg|_{\sigma_t^2 = 0^{+}} = -\frac{1}{2} \int_{\mathbb{R}^D} p(\mathbf{x}) \|\nabla \log p(\mathbf{x}) - \nabla \log q(\mathbf{x};{\boldsymbol{\theta}})\|_2^2\, d\mathbf{x},
\end{equation}
which completes the proof.
\end{proof}

Similar to the classical result in \cite{10.5555/1795114.1795156,verdu2010mismatched}, the relation in Theorem~\ref{theorem 1} holds because both sides quantify the error induced by a mismatch between the true distribution \( p \) and the prior \( q \) provided to the estimator. Naturally, when \( p = q \), both sides vanish; otherwise, the derivative is strictly negative, indicating that perturbations reduce the relative entropy. This observation also yields the relative entropy version of the data processing inequality:
\begin{equation}
    D_{\text{KL}}(\rho\|\nu) \geq D_{\text{KL}}(\bar{\rho}\|\bar{\nu}),
\end{equation}
where
\[
\bar{\rho} = \int W(\mathbf{y}\vert\mathbf{x})\,d\rho(\mathbf{x}), \qquad
\bar{\nu} = \int W(\mathbf{y}\vert\mathbf{x})\,d\nu(\mathbf{x}).
\]
Here, \( W \) denotes a noisy channel. This inequality asserts that the KL divergence between two distributions decreases under the action of a common channel, which is consistent with the data processing argument used in \cite{song2021maximum}, where, the channel \( W \) corresponds to time-reversed Brownian motion, which can be viewed as a continuous analogue of the Gaussian channel \cite{8007014}. Under the assumption that \( \bar{\rho}(\mathbf{x}) = \bar{\nu}(\mathbf{x}) = \pi(\mathbf{x}) \), the result further implies that the diffusion process smooths the discrepancy between inputs. Moreover, the neural network acts as a channel simulator, enabling efficient sampling by integrating its output into the neural ODE solvers \cite{10.5555/3600270.3602196}.

\subsubsection{Proof of Proposition.~\ref{proposition}}\label{proof:prop}

\textbf{Proposition \ref{proposition}.} \textit{Consider the signal model \eqref{channel eq}, suppose $p(\tilde{\mathbf{x}})$ and $q(\tilde{\mathbf{x}};{\boldsymbol{\theta}})$ have continuous second-order derivatives and finite second moments. Denote by $p(\mathbf{y}_1)$ and $\pi(\mathbf{x})$ the output signals of channel at time $t=1$ when inputs are $p(\mathbf{x})$ and $q(\hat{\mathbf{x}};{\boldsymbol{\theta}})$ respectively. Assume $\pi(\mathbf{x})=\mathcal{N}(0,\mathbf{I})$, which is independent of $\boldsymbol{\theta}$. Let $p(\mathbf{y}_t|\mathbf{x})$ denote the channel \eqref{channel eq} for any $t\in[0,1]$, then for arbitrary datapoint $\mathbf{x}$ and small $\sigma_0^2$ with $0<\sigma_0^2\ll1$:}
\begin{equation}
         \mathcal{H}(p(\mathbf{x}),q(\hat{\mathbf{x}};{\boldsymbol{\theta}})) = \mathcal{H}(p(\mathbf{y}_1),\pi(\mathbf{x})) + \mathcal{J}_{\text{DSM}}(\boldsymbol{\theta};\sigma_t^2(\cdot)) - \frac{1}{2}\int_{\sigma_0^2}^{\sigma^2_1}\mathbb{E}\|\nabla_{\mathbf{y}_{t}}\log p(\mathbf{y}_{t}\vert\mathbf{x})\|^{2}d\sigma^2_t + o(\sigma_0^2).
\end{equation}
     \textit{Here, the denoising score matching objective is defined as
     \begin{equation}
         \mathcal{J}_{\text{DSM}}(\boldsymbol{\theta};\sigma^2(\cdot))  := \frac{1}{2}\int_{\sigma_0^2}^{\sigma^2_1} \mathbb{E}\|\nabla_{\mathbf{y}_{t}}\log p(\mathbf{y}_{t}\vert  \mathbf{x}) - \hat{\boldsymbol{s}}(\mathbf{y}_{t};\boldsymbol{\theta})\|_2^{2}\,d\sigma^2_t,
     \end{equation}
     and \(\hat{\boldsymbol{s}}(\mathbf{y}_t; \boldsymbol{\theta}) = \nabla_{\mathbf{y}_t} \log q(\mathbf{y}_t; \boldsymbol{\theta}) := - \hat{\mathbf{n}}(\mathbf{y}_t,\eta_t;{\boldsymbol{\theta}})/\sigma_t\) is a score network estimator.}

\begin{proof}
Recall the thermodynamic integration commonly used in statistical physics:
    \begin{equation}
        \int_{0}^{\Xi}\frac{d}{d\xi}f(\xi)d\xi = f(\Xi) - f(0).
    \end{equation}
    Let \(f(\sigma_t^2) := D_{\text{KL}}\bigl(p(\mathbf{y})\|q(\mathbf{y};{\boldsymbol{\theta}})\bigr)\), where $\mathbf{Y}_t = \alpha_t \mathbf{X} + \sigma_t \mathbf{N}$, i.e.,

    \begin{equation}
        \int_{0}^{\Xi}\frac{d}{d\xi}f(\xi)d\xi = f(\Xi) - f(0),
    \end{equation} 
    becomes equivalent to
    \begin{equation}
        \int_{0}^{\sigma_1^2}\frac{d}{d\sigma_t^2}f(\sigma_t^2)d\sigma_t^2 = f(\sigma_1^2) - f(0),
    \end{equation}
    which yields the following expression,
    \begin{equation}
        \int_0^{\sigma_1^2} \frac{d}{d\sigma_t^2} D_{\text{KL}}(p(\mathbf{y})\|q(\mathbf{y};{\boldsymbol{\theta}})) = D_{\text{KL}}(p(\mathbf{y}_1)\|q(\mathbf{y}_1;{\boldsymbol{\theta}})) - D_{\text{KL}}(p(\mathbf{x})\|q(\hat{\mathbf{x}};{\boldsymbol{\theta}})). \label{85}
    \end{equation}
This occurs because the KL divergence remains invariant under measurable transformations shown in Lemma~\ref{KLinvariance}. In particular, probability density transformation serves as an illustrative example. Let \( \mathbf{X} \sim \rho \) be a random vector on the measure space with a probability density function \( p(\mathbf{x}) \). We define the transformation \( \mathbf{X}' = T(\mathbf{X}) \). Consequently, the probability density function of \( \mathbf{X}' \), denoted as \( p(\mathbf{x}') \), can be determined using:
\begin{equation}
    p(\mathbf{x}') = p(T^{-1}(\mathbf{x}')) \left| \frac{d}{d\mathbf{x}'} T^{-1}(\mathbf{x}') \right|.
\end{equation}

Furthermore, the KL divergence satisfies the following property:
\begin{equation}
    D_{\text{KL}}(\rho(\alpha\mathbf{x})\|\nu(\alpha\mathbf{x})) = D_{\text{KL}}(\rho(\mathbf{x})\|\nu(\mathbf{x})) \quad \forall \alpha > 0,
\end{equation}

which implies that the KL divergence depends solely on the relative shape of the two distributions rather than their absolute scale. Invoke the Lemma~\ref{lemma 4} and results in \cite{10.5555/1795114.1795156} on \eqref{85} yields
    \begin{equation}
        D_{\text{KL}}(p(\mathbf{x})\|q(\hat{\mathbf{x}};{\boldsymbol{\theta}})) = D_{\text{KL}}(p(\mathbf{y}_1)\|q(\mathbf{y}_1;{\boldsymbol{\theta}})) + \frac{1}{2}\int_0^{\sigma_1^2} I(p(\mathbf{y}_t)\|q(\mathbf{y}_t;{\boldsymbol{\theta}})) d\sigma_t^2.
    \end{equation}
With the fact that $q(\mathbf{y}_1;{\boldsymbol{\theta}}) := \pi(\mathbf{x})$, which is independent of $\boldsymbol{\theta}$, we have
    \begin{equation}\label{64}
        D_{\text{KL}}(p(\mathbf{x})\|q(\hat{\mathbf{x}};{\boldsymbol{\theta}})) = D_{\text{KL}}(p(\mathbf{y}_1)\|\pi(\mathbf{x})) + \frac{1}{2}\int_0^{\sigma_1^2} I(p(\mathbf{y}_t)\|q(\mathbf{y}_t;{\boldsymbol{\theta}})) \,d\sigma_t^2.
    \end{equation}
    The mismatched entropy between data and model becomes equivalent to
    \begin{equation}
        \mathcal{H}(p(\mathbf{x}),q(\hat{\mathbf{x}};{\boldsymbol{\theta}})) = D_{\text{KL}}(p(\mathbf{x})\|q(\hat{\mathbf{x}};{\boldsymbol{\theta}})) + \mathcal{H}(p(\mathbf{x})). \label{00}
    \end{equation}
    Invoking \eqref{64} on \eqref{00} yields
    \begin{equation}
        \mathcal{H}(p(\mathbf{x}),q(\hat{\mathbf{x}};{\boldsymbol{\theta}})) =  D_{\text{KL}}(p(\mathbf{y}_1)\|\pi(\mathbf{x})) + \frac{1}{2} \int_0^{\sigma_{1}^{2}} I(p(\mathbf{y}_t) \| q(\mathbf{y}_t;{\boldsymbol{\theta}})) \, d\sigma_t^2 + \mathcal{H}(p(\mathbf{x})).
    \end{equation}
    Recall de Bruijin's identity \cite{1057105}:
    \begin{equation}
        \frac{d}{d\sigma_t^2} \mathcal{H}(p(\mathbf{y}_t)) = \frac{1}{2} \mathcal{J}(p(\mathbf{y}_t)).
    \end{equation}
    Using thermodynamic integration again, we have
    \begin{align}
        \int_0^{\sigma_1^2} \frac{d}{d\sigma_t^2} \mathcal{H}(p(\mathbf{y}_t)) d\sigma_t^2 &= \mathcal{H}(p(\mathbf{y}_1)) - \mathcal{H}(p(\alpha_0\mathbf{x})),
    \end{align}
    which is equivalent to
    \begin{align}
         \mathcal{H}(p(\alpha_0\mathbf{x})) &= \mathcal{H}(p(\mathbf{y}_1)) - \frac{1}{2}\int_0^{\sigma_1^2}\mathcal{J}(p(\mathbf{y}_t)) d\sigma_t^2.
    \end{align}
    Recall the property of differential entropy (see Theorem 8.6.4 (8.71) in \cite{cover1999elements})
    \begin{equation}
        \mathcal{H}(a\mathbf{X}) = \mathcal{H}(\mathbf{X}) + \log |\det(a)|,
    \end{equation}
    and the connection with denoising score matching (see (11) in \cite{vincent2011connection})
    \begin{equation}
        I(p(\mathbf{y}_t)\|q(\mathbf{y}_t;{\boldsymbol{\theta}})) = \mathbb{E}\left[\|\nabla\log p(\mathbf{y}_{t}\vert  \mathbf{x}) - \boldsymbol{s}(\mathbf{y}_{t};\boldsymbol{\theta})\|^{2}\right] + \mathcal{J}(p(\mathbf{y}_t)) - \mathbb{E}\left[\|\nabla\log p(\mathbf{y}_{t}\vert\mathbf{x})\|^{2}\right].
    \end{equation}
    Finally, we have
    \begin{align}
        \mathcal{H}&(p(\mathbf{x}),q(\hat{\mathbf{x}};{\boldsymbol{\theta}})) = \Bigg(D_{\text{KL}}(p(\mathbf{y}_1)\|\pi(\mathbf{x})) + \mathcal{H}(p(\mathbf{y}_1)) - \frac{1}{2}\int_{0}^{\sigma_1^2} \mathcal{J}(p(\mathbf{y}_t))\,d\sigma_t^2  - D\log |\alpha_0| \notag \\
        &+ \frac{1}{2} \bigg(\int_{0}^{\sigma_1^2}\mathbb{E}\left[\|\nabla\log p(\mathbf{y}_{t}\vert  \mathbf{x}) - \hat{\boldsymbol{s}}(\mathbf{y}_{t};\boldsymbol{\theta})\|^{2}\right] + \mathcal{J}(p(\mathbf{y}_t)) - \mathbb{E}\left[\|\nabla\log p(\mathbf{y}_{t}\vert\mathbf{x})\|^{2}\right]d\sigma_t^2\bigg)\Bigg) \\
        &= \mathcal{H}(p(\mathbf{y}_1),\pi(\mathbf{x})) + \mathcal{J}_{\text{DSM}}(\boldsymbol{\theta};\sigma_t^2(\cdot)) - \frac{1}{2}\int_{0}^{\sigma^2_1}\mathbb{E}\|\nabla_{\mathbf{y}_{t}}\log p(\mathbf{y}_{t}\vert\mathbf{x})\|^{2}d\sigma^2_t - D\log|\alpha_0|. 
    \end{align}
    Since the identity \( \sigma_t^2 + \alpha_t^2 = 1 \) holds for all \( t \), setting the lower limit of integration such that \( \sigma_0^2 = 0 \) implies \( \alpha_0^2 = 1 \). Consequently, \( |\alpha_0| = 1 \), and thus: $\log |\alpha_0| = \log 1 = 0$. And invoking Theorem~\ref{theorem 1}, we finally get
    \begin{equation}
        \mathcal{H}(p(\mathbf{x}),q(\hat{\mathbf{x}};{\boldsymbol{\theta}})) = \mathcal{H}(p(\mathbf{y}_1),\pi(\mathbf{x})) + \mathcal{J}_{\text{DSM}}(\boldsymbol{\theta};\sigma_t^2(\cdot)) - \frac{1}{2}\int_{\sigma_0^2}^{\sigma^2_1}\mathbb{E}\|\nabla_{\mathbf{y}_{t}}\log p(\mathbf{y}_{t}\vert\mathbf{x})\|^{2}d\sigma^2_t + o(\sigma_0^2).
    \end{equation}
    
\end{proof}

\subsubsection{Proof of Theorem~\ref{theorem ind}}\label{proof:theorem2}
\textbf{Theorem \ref{theorem ind}.} \textit{ Let $p(\mathbf{y}_t \vert \mathbf{x})=\mathcal{N}(\alpha_t\mathbf{x},\sigma_t^2\mathbf{I})$ denote the Gaussian channel at any time $t\in[0,1]$. With the same notations and conditions in Proposition~\ref{proposition}, we have} \label{proof of theorem 4}
    \begin{equation}
        -\log q(\hat{\mathbf{x}};{\boldsymbol{\theta}}) \leq \mathcal{H}(p(\mathbf{y}_{1}\vert\mathbf{x}),\mathbf{\pi}(\mathbf{x})) + \mathcal{L}_{\text{DSM}}(\sigma_t^2;\boldsymbol{\theta}), 
    \end{equation}
    \textit{in which $\mathcal{L}_{\text{DSM}}(\sigma_t^2;\boldsymbol{\theta})$ is defined as} \[\mathcal{L}_{\text{DSM}}(\sigma_t^2;\boldsymbol{\theta}) := \frac{1}{2}\int_{\sigma_0^2}^{\sigma_1^2} \mathbb{E}_{p(\mathbf{y}_t\vert\mathbf{x})}\|\nabla_{\mathbf{y}_{t}}\log p(\mathbf{y}_{t}\vert\mathbf{x}) - \boldsymbol{s}(\mathbf{y}_{t};\boldsymbol{\theta})\|^{2}d\sigma^2_t, \]
    \textit{and $\mathcal{H}(p(\mathbf{y}_{1}\vert\mathbf{x}),\mathbf{\pi}(\mathbf{x}))$ is given by} \[\mathcal{H}(p(\mathbf{y}_{1}\vert\mathbf{x}),\mathbf{\pi}(\mathbf{x}))=D_{\text{KL}}(p(\mathbf{y}_{1}\vert\mathbf{x})\|q(\mathbf{y}_1)) + \mathcal{H}( p(\mathbf{y}_{1}\vert\mathbf{x})) .\]

\begin{proof}
    From Proposition~\ref{proposition} we have
    \begin{align}
        -\mathbb{E}_{p(\mathbf{x})} [\log q(\hat{\mathbf{x}};{\boldsymbol{\theta}})] &= -\mathbb{E}_{p(\mathbf{y}_1)} [\log \pi(\mathbf{x})]  + \frac{1}{2}\int_{\sigma_0^2}^{\sigma^2_1} \mathbb{E}_{p(\mathbf{y}_t\vert\mathbf{x}) p(\mathbf{x})}\|\nabla_{\mathbf{y}_{t}}\log p(\mathbf{y}_{t}\vert  \mathbf{x}) - \boldsymbol{s}(\mathbf{y}_{t};\boldsymbol{\theta})\|^{2}d\sigma^2_t \notag \\
        &-\frac{1}{2}\int_{\sigma_0^2}^{\sigma^2_1}\mathbb{E}_{p(\mathbf{y}_t\vert\mathbf{x}) p(\mathbf{x})}\|\nabla_{\mathbf{y}_{t}}\log p(\mathbf{y}_{t}\vert\mathbf{x})\|^{2}d\sigma^2_t \\
        &= -\int_{\mathbb{R}^D} \int_{\mathbb{R}^D}p(\mathbf{y}_1\vert\mathbf{x})p(\mathbf{x})d\mathbf{x} \log\pi(\mathbf{x}) d\mathbf{y}_1 \notag \\ 
        &-\frac{1}{2}\int_{\sigma_0^2}^{\sigma^2_1}\iint_{\mathbb{R}^D}{p(\mathbf{y}_t\vert\mathbf{x})p(\mathbf{x})}\|\nabla_{\mathbf{y}_{t}}\log p(\mathbf{y}_{t}\vert\mathbf{x})\|^{2} d\mathbf{x} d\mathbf{y}_t d\sigma^2_t \notag \\
        &+ \frac{1}{2}\int_{\sigma_0^2}^{\sigma^2_1}\iint_{\mathbb{R}^D}{p(\mathbf{y}_t\vert\mathbf{x})p(\mathbf{x})}\|\nabla_{\mathbf{y}_{t}}\log p(\mathbf{y}_{t}\vert  \mathbf{x}) - \boldsymbol{s}(\mathbf{y}_{t};\boldsymbol{\theta})\|^{2}d\mathbf{x}d\mathbf{y}_t d\sigma^2_t \label{108}
    \end{align}
    Given a fixed channel $p(\mathbf{y}_t\vert\mathbf{x})$, we can easily see that $\int_{\mathbb{R}^D}p(\mathbf{x})d\mathbf{x}$ in both sides of \eqref{108} can be canceled to get
    \begin{align}
        -\log q(\hat{\mathbf{x}};{\boldsymbol{\theta}}) &= -\int_{\mathbb{R}^D} p(\mathbf{y}_1\vert\mathbf{x}) \log \pi(\mathbf{x})d\mathbf{y}_1 - \frac{1}{2}\int_{\sigma_0^2}^{\sigma^2_1}\int_{\mathbb{R}^D}{p(\mathbf{y}_t\vert\mathbf{x})}\|\nabla_{\mathbf{y}_{t}}\log p(\mathbf{y}_{t}\vert\mathbf{x})\|^{2}d\mathbf{y}_t d\sigma^2_t \notag \\
        &+ \frac{1}{2}\int_{\sigma_0^2}^{\sigma^2_1}\int_{\mathbb{R}^D}{p(\mathbf{y}_t\vert\mathbf{x})}\|\nabla_{\mathbf{y}_{t}}\log p(\mathbf{y}_{t}\vert  \mathbf{x}) - \boldsymbol{s}(\mathbf{y}_{t},\boldsymbol{\theta})\|^{2}d\mathbf{y}_t d\sigma^2_t . \label{109}
    \end{align}
    The second term in \eqref{109} corresponds to one-half of the integrated Fisher information of the Gaussian distribution, which is strictly non-negative \cite{Fisher_1925}, we obtain:
    %The second term in \eqref{109} is actually half the integral of Fisher information w.r.t given Gaussian, which is always positive \cite{Fisher_1925}. Also, with the minimum Fisher property of Gaussian given the entropy constraint, we get:
    \begin{equation}
        -\log q(\hat{\mathbf{x}};{\boldsymbol{\theta}}) \leq \mathcal{H}(p(\mathbf{y}_{1}\vert\mathbf{x}),\mathbf{\pi}(\mathbf{x}))+ \frac{1}{2}\int_{\sigma_0^2}^{\sigma^2_1}\int_{\mathbb{R}^D}{p(\mathbf{y}_t\vert\mathbf{x})}\|\nabla_{\mathbf{y}_{t}}\log p(\mathbf{y}_{t}\vert  \mathbf{x}) - \boldsymbol{s}(\mathbf{y}_{t},\boldsymbol{\theta})\|^{2}d\mathbf{y}_t d\sigma^2_t ,
    \end{equation}
    which finishes the proof.
\end{proof}
\subsubsection{Proof of Proposition~\ref{theorem 2}}
\textbf{Proposition \ref{theorem 2}.}
    \textit{Let \( \tilde{\mathbf{X}} = \alpha_t \mathbf{X} + \sigma_t \boldsymbol{\Psi} \), where \( \mathbf{X}, \boldsymbol{\Psi} \in \mathbb{R}^D \) are independent random vectors, with \( \boldsymbol{\Psi} \) satisfying \( \mathbb{E}[\boldsymbol{\Psi}] = 0 \) and \( \operatorname{Cov}(\boldsymbol{\Psi}) = \mathbf{I} \). Assume the probability density function \( p(\mathbf{x}) \) of \( \mathbf{X} \) is twice continuously differentiable and decays sufficiently fast at infinity, and that the Fisher information \( \mathcal{J}(\mathbf{X}) \) exists and is finite. Then,
    \begin{equation}
        \frac{d}{d\sigma_t^2}\mathcal{H}\bigl(p(\tilde{\mathbf{x}})\bigr)\bigg|_{\sigma_t^2 \rightarrow  0^{+}} = \frac{1}{2} \mathcal{J}\big(p(\mathbf{x})\big),
    \end{equation}
    where \( \mathcal{H}(\cdot) \) denotes the differential Shannon entropy and \( \mathcal{J}(\cdot) \) denotes the Fisher information.}
\begin{proof}
By the smoothing properties established in Lemma~\ref{arbitrary noise heat equation}, the distribution \( p(\tilde{\mathbf{x}}) \) is differentiable with respect to \( \sigma_t^2 \). Since the integrand in equation~\eqref{5} is both continuous and differentiable in \( \sigma_t^2 \), we may interchange the order of differentiation and integration to obtain:
\begin{align}
    \frac{d}{d\sigma_t^2} \mathcal{H}\big(p(\tilde{\mathbf{x}})\big) 
    &= - \int_{\mathbb{R}^D} \frac{d}{d\sigma_t^2} p(\tilde{\mathbf{x}}) \cdot \left(1 + \log p(\tilde{\mathbf{x}}) \right) \, d\tilde{\mathbf{x}}\\
    &= -\frac{1}{2} \int_{\mathbb{R}^D} \left( \Delta_{\tilde{\mathbf{x}}} p(\tilde{\mathbf{x}}) \right) \log p(\tilde{\mathbf{x}}) \, d\tilde{\mathbf{x}}.
\end{align}
Following the same argument as in Theorem~\ref{theorem 1}, this establishes the desired expression, which recovers the special case of isotropic additive noise discussed in \cite{rioul2010information}.
\end{proof}

\section{Improve the Likelihood Estimation Bounds}

Our theoretical analysis suggests that different variance functions can lead to varying performance in likelihood estimation. To validate this empirically, we evaluate the effect of different noise schedules, variance functions and datasets on likelihood estimation performance.

However, in practice, exact numerical evaluation of the integral is generally intractable. A common approach is to use Monte Carlo method to estimate the variance integral in \(\mathcal{L}_{\text{DSM}}(\sigma_t^2;\boldsymbol{\theta})\) during both training and evaluation. Moreover, reducing the variance of the Monte Carlo estimator for the continuous loss objective generally improves the efficiency of optimization. We next summarize our empirical observations on likelihood estimation across different settings.

\subsection{Variance-Aware Likelihood Bounds}\label{Variance-Aware Likelihood Bounds}
As mentioned in Section~\ref{individual likelihood}, denoising score matching \cite{vincent2011connection} is typically used as the training objective for diffusion models, serving as a surrogate for approximating the likelihood function $q(\hat{\mathbf{x}};\boldsymbol{\theta})$. Specifically, recall that out channel \eqref{channel eq} is defined as $p(\mathbf{y}_t\vert\mathbf{x}) = \mathcal{N}(\alpha_t\mathbf{x},\sigma_t^2\mathbf{I})$, such that  $\nabla_{\mathbf{y}_t} \log p(\mathbf{y}_t\vert\mathbf{x})= -\mathbf{n}/\sigma_t$. In the sense that with the score model $\hat{\boldsymbol{s}}(\mathbf{y}_t; \boldsymbol{\theta}) := - \hat{\mathbf{n}}(\mathbf{y}_t,\eta_t;{\boldsymbol{\theta}})/\sigma_t$, the objective takes the following integral form, where the loss is evaluated under varying noise levels:
\begin{equation}
    \mathcal{L}_{\text{DSM}}(\sigma_t^2;\boldsymbol{\theta}) = \frac{1}{2} \int_{\sigma_0^2}^{\sigma_1^2} \mathbb{E}_{p(\mathbf{y}_t\vert\mathbf{x})}  \left[\sigma_t^{-2}\|\mathbf{n} - \hat{\mathbf{n}}(\mathbf{y}_t,\eta_t;{\boldsymbol{\theta}})\|_2^2 \right]\ d\sigma_t^2.\label{loss}
\end{equation}

We begin by approximating the expectation in $\mathcal{L}_{\text{DSM}}(\sigma_t^2;\boldsymbol{\theta})$ via Monte Carlo estimation. This involves drawing samples $\mathbf{y}_t \sim p(\mathbf{y}_t\vert\mathbf{x})$ from a tractable Gaussian kernel, which enables efficient estimation of the denoising objective. To estimate the integral over the variance schedule, we apply a second Monte Carlo approximation by uniformly sampling $\sigma_t^2$ from the interval $[\sigma_0^2,\sigma_1^2]$. However, we empirically observe that such an estimator can significantly degrade the optimization process, potentially due to increased variance or poor convergence behaviour. To address this, we reparameterize the objective in terms of a smoother coordinate, namely, the negative log signal-to-noise ratio (log-SNR), which facilitates both numerical stability and analytical tractability.

Even under constrained noise schedules, such as the Variance Preserving (VP) formulation~\cite{ho2020denoising}, which enforces $\alpha^2_t + \sigma_t^2 = 1$, there remains freedom in how $\alpha_t$, $\sigma_t$ evolve over time $t$. To abstract away from specific parameterizations such as VP or VE schedules, we adopt a log-SNR parameterization as the default: $\eta_t:=- \log \text{SNR}(t) = - \log \frac{\alpha_t^2}{\sigma_t^2}$. This reparameterization simplifies the analysis and unifies various diffusion formulations under a single coordinate system. Under this reparameterization, the change of variable from time $t$ to negative log-\text{SNR} $\eta$ follows:
\begin{equation}
    \frac{d\eta_t}{dt} =  \frac{1}{\sigma_t^2}\frac{d\sigma_t^2}{dt} - \frac{1}{\alpha_t^2}\frac{d\alpha_t^2}{dt}.
\end{equation}
For the VP schedule, this simplifies to:
\begin{equation}
    \frac{d\eta_t}{dt} = \frac{1}{\alpha_t^2\sigma_t^2}\frac{d\sigma_t^2}{dt}.
\end{equation}

By setting the noisy process to follow an EDM-like variance-exploding (VE) schedule~\cite{10.5555/3600270.3602196}, we let $\alpha_t^2 \equiv 1$ and parameterize the variance as $\sigma_t^2 := \exp(\eta_t)$. This configuration is equivalent to that of NCSNv2~\cite{song2020improved} when $\sigma_t$ follows a geometric sequence interpolating between $0.01$ and $50$, i.e., $\eta(t) = 2\log(0.01) + 2\log(5000)\,t$. Under this specification, the objective in~\eqref{loss} reduces to the continuous-time loss proposed in~\cite{kingma2021variational}:
\begin{equation}\label{15}
    \mathcal{L}_{\infty}(\mathbf{x}) := \frac{1}{2} \mathbb{E}_{\mathbf{n}\sim \mathcal{N}(0,\mathbf{I}), t\sim\mathcal{U}(0,1)}[\eta'(t)\|\mathbf{n}-\hat{\mathbf{n}}_{\boldsymbol{\theta}}(\mathbf{y}_t;t)\|_2^2],
\end{equation}
where $\eta'(t) = d\eta_t/dt$ acts as a natural weighting function. In the above case, this means that $\eta'(t) = 2\log[5000]$ and thus that weighted function is a constant. We leave the VE schedule optimization for future works.

Similarly, in the information-theoretic VE formulation~\cite{kong2023informationtheoretic}, i.e., $\mathbf{z}_{\gamma} =\sqrt{\gamma} \ \mathbf{x} + \mathbf{n}$, our bound recovers the loss integrates over the signal-to-noise ratio $\gamma$:
\begin{equation}
    \mathcal{L}:=\frac{1}{2}\int_{\text{SNR}_{\text{min}}}^{\text{SNR}_{\text{max}}}\mathbb{E}_{p(\mathbf{z}_\gamma,\mathbf{x})}\bigl[\|\mathbf{x}-\hat{\mathbf{x}}_{\boldsymbol{\theta}}(\mathbf{z}_\gamma,\gamma)\|_2^2\bigr] \ d\gamma,
\end{equation}
where instead of integrating w.r.t. time $t$, now integrate w.r.t. the signal-to-noise ratio $\gamma$, and where ${\text{SNR}_{\text{max}}} =\gamma(1)$, ${\text{SNR}_{\text{min}}}=\gamma(0)$.

In a more general form, the objective in \eqref{loss} recovers the weighted continuous-time formulation of Eq.(66)~\cite{kingma2021variational}, where the weighting function is defined as $w(t):=\alpha_t^2$; equivalently, this corresponds to scaling the signal coefficient by $\alpha_t$ across time:
\begin{equation}%\label{15}
    \mathcal{L}^{w}_{\infty}(\mathbf{x}) := \frac{1}{2} \mathbb{E}_{\mathbf{n}\sim \mathcal{N}(0,\mathbf{I}), t\sim\mathcal{U}(0,1)}[\eta'(t)\alpha_t^{2}\|\mathbf{n}-\hat{\mathbf{n}}_{\boldsymbol{\theta}}(\mathbf{y}_t;t)\|_2^2],
\end{equation}

Moreover, \eqref{loss} can be further unified under the stochastic differential equation (SDE) framework~\cite{song2020score}:
\begin{equation}
    d\mathbf{x} = \boldsymbol{f}(\mathbf{x},t)\,dt +g(t) \,d\mathbf{w},
\end{equation}
where $\boldsymbol{f}(\mathbf{x},t):\mathbb{R}^D\rightarrow\mathbb{R}^D$, $g(t)\in\mathbb{R}$ are manually designed noise schedules and $\mathbf{w} \in \mathbb{R}^D$ is a standard Wiener process. When the variance function satisfies $d\sigma_t^2/dt =g^2(t)$, the corresponding loss aligns with the likelihood-weighted score-matching objective~\cite{song2021maximum}:
\begin{align}
    \mathcal{J}_{\text{SM}}(\boldsymbol{\theta}):= \int_0^T \frac{g^2(t)}{2\sigma_t^2}\mathbb{E}_{\mathbf{x}\sim p(\mathbf{x}),\,\mathbf{n}\sim\mathcal{N}(0,\mathbf{I})}\bigl[\|\sigma_t\hat{\boldsymbol{s}}(\mathbf{y}_t;{\boldsymbol{\theta}})+\mathbf{n}\|_2^2\bigr]\,dt.
\end{align}
The same principle extends to deterministic flow-based formulations, including continuous-time ODEs and higher-order probability flows~\cite{lu2022maximum,song2021maximum, zheng2023improved} with specific design of the $g^2(t)$. In this case, the training objective becomes:
\begin{align}
    \mathcal{J}_{\text{SM}}^{\eta}(\boldsymbol{\theta}) := \frac{1}{2}\int_{\eta_0}^{\eta_1} \mathbb{E}_{\mathbf{x}\sim p(\mathbf{x}),\, \mathbf{n}\sim\mathcal{N}(0,\mathbf{I})}\bigl[\|\sigma_t\hat{\boldsymbol{s}}(\mathbf{y}_{\eta};{\boldsymbol{\theta}})+\mathbf{n}\|_2^2\bigr]\,d\eta.
\end{align}

This reformulation highlights that the loss landscape is intimately linked to the choice of noise schedule, particularly the functional form of \( \sigma_t^2 \) as a function of log-SNR \( \eta \). Since both the weighting in the integrand and the resulting likelihood bounds depend explicitly on \( \sigma_t^2 \), different scheduling strategies can lead to substantial differences in training dynamics, gradient variance, and overall model performance. In the following, we provide both theoretical and empirical analyses to quantify how the design of \( \sigma_t^2 \) influences estimation accuracy, stability, and convergence behaviour.

\subsection{The Noise Schedule Matters}\label{B.1}

Many practical objectives, including the ones above, are special cases of the likelihood–weighting loss
\cite{kingma2023understanding,kingma2021variational,song2021maximum}. Let
$\eta=g(t)$ be a differentiable, monotone noise parameterization of time
$t\in[0,1]$ with fixed endpoints $\eta_{\min}=g(0)$ and $\eta_{\max}=g(1)$.
For a per-example loss, define
\begin{equation}\label{eq:weighted-loss-eta}
\mathcal{L}_{w}(\mathbf{x})
:= \frac{1}{2}\int_{\eta_{\min}}^{\eta_{\max}}
w(\eta)\ \mathbb{E}_{\mathbf n\sim\mathcal N(0,\mathbf I)}
\!\left[\left\|\mathbf n-\hat{\mathbf n}\bigl(\mathbf y_{\eta},\eta;\boldsymbol\theta\bigr)\right\|_2^2\right]\, d\eta,
\end{equation}
where $\mathbf y_{\eta}$ denotes the noisy input at noise level $\eta$ (e.g.,
via a log-SNR parameterization). As shown in \cite{kingma2021variational}, the
\emph{integral} in \eqref{eq:weighted-loss-eta} is invariant to smooth
reparameterizations of $\eta$ (i.e., to the choice of schedule) so long as the
endpoints $\eta_{\min},\eta_{\max}$ are kept fixed. This invariance extends to
the weighted diffusion loss family \cite{kingma2023understanding} because the
integrand is measured per unit $\eta$.

However, this invariance does \emph{not} carry over to the Monte Carlo (MC)
estimator used in evaluation. If we sample $t\sim\mathcal U(0,1)$ and
$\mathbf n\sim\mathcal N(0,\mathbf I)$, then by change of variables,
\begin{equation}\label{eq:mc-unif-t}
\mathcal{L}_{w}(\mathbf{x})
= \frac{1}{2}\ \mathbb{E}_{t\sim\mathcal U(0,1),\,\mathbf {n}\sim\mathcal{N}(0,\mathbf{I})}\!
\left[\, w\bigl(\eta(t)\bigr)\ \Bigl|\tfrac{d\eta}{dt}(t)\Bigr|\
\left\|\mathbf n-\hat{\mathbf n}\!\left(\mathbf y_{\eta(t)},\eta(t);\boldsymbol\theta\right)\right\|_2^2 \right].
\end{equation}
More generally, if we sample directly $\eta\sim\rho(\eta)$ over
$[\eta_{\min},\eta_{\max}]$ (with $\rho>0$ a.e.), then
\begin{equation}\label{eq:mc-importance-eta}
\mathcal{L}_{w}(\mathbf{x})
= \frac{1}{2}\ \mathbb{E}_{\eta\sim\rho,\ \mathbf n}\!
\left[\, \frac{w(\eta)}{\rho(\eta)}\
\left\|\mathbf n-\hat{\mathbf n}\bigl(\mathbf y_{\eta},\eta;\boldsymbol\theta\bigr)\right\|_2^2 \right],
\end{equation}
revealing that the \emph{noise schedule} induces an \emph{importance-sampling
distribution} $\rho(\eta)$ over $\eta$ and a corresponding weighting function $w(\eta)$ related to $\sigma_t^2(\eta)$, simultaneously. Consequently, while the population loss is invariant
to reparameterization, the \emph{variance} of the MC estimator (and of its
gradients) depends on $\rho$. Choosing $\rho$ to better match the integrand can
substantially reduce estimator variance and thus accelerate optimization. In
particular, the variance-minimizing density satisfies the
heuristic proportionality
$\rho^\star(\eta)\propto \mathbb E_{\mathbf{x}\sim\mathcal{X}, \mathbf{n}\sim\mathcal{N}(0,\mathbf{I})}\!\left[w(\eta)\left\|\mathbf n-\hat{\mathbf n}(\mathbf y_{\eta},\eta;\boldsymbol\theta)\right\|_2^2\right]\,
$.

\subsection{Variance Reduction with Importance Sampling}\label{VR:IS}

Monte Carlo approximation offers a computationally efficient alternative to exact integration, but typically introduces variance into the training objective. To mitigate this, we adopt both designed and learned importance sampling (IS) strategies, as outlined in Section~\ref{importance sampling}. As mentioned in Section.~\ref{B.1}, the diffusion model is conducted for all $\eta$ in $[\eta_0,\eta_1]$ through an integral. In practice, the evaluation of the integral is time-consuming, and Monte-Carlo methods are used to unbiasedly estimate the objective by uniformly sampling $\eta$. Thus, a continuous importance distribution $\rho(\eta)$ can be proposed for variance reduction. Denote \[\mathcal{L}(\mathbf{x},\mathbf{n},\eta;\boldsymbol{\theta}) :=  \frac{\alpha_t^{2}\|\mathbf{n} - \hat{\mathbf{n}}(\mathbf{y}_\eta,\eta;{\boldsymbol{\theta}})\|_2^2}{2}, \]
then
\begin{equation}
    \mathcal{L}_{w}(\mathbf{x};\boldsymbol{\theta}) = \mathbb{E}_{\eta\sim\rho(\eta)}\mathbb{E}_{\mathbf{x},\mathbf{n}}\left[\frac{\mathcal{L}(\mathbf{x},\mathbf{n},\eta;\boldsymbol{\theta})}{\rho(\eta)}\right].
\end{equation}
We propose to use two types of importance sampling (IS), and empirically compare them for faster convergence. 

\subsubsection{Designed IS}

 In particular, we propose a continuous proposal distributions $\rho(\eta)$ over \(\eta\), which is proportional $\alpha_\eta^2$. Since we have explicit expressions for the density, we utilize inverse transform sampling to design a sampling procedure. Concretely, we take uniform samples of a number $t\in[0,1]$, and solve the following equation about $\eta_t$:
\begin{equation}
    \frac{1}{Z}\int_{\eta_0}^{\eta_1} \alpha_{\eta}^{2} \,d\eta = t, \quad  Z= \int_{\eta_0}^{\eta_1} \alpha_{\eta}^{2} \,d\eta,
\end{equation}
where we have defined maximum time $t=1$, and $Z$ is a normalizing constant.

\paragraph{Uniform Weighting in \(\eta\)-Space.}  
As a canonical example, we consider uniform weighting in the log-SNR domain. In this case, the proposal distribution becomes \(\rho(\eta) \propto \alpha_\eta^2\), yielding:
\[
\quad Z = \int_{\eta_0}^{\eta_1} \text{sigmoid}(-\eta)\, d\eta = \log \left( \frac{1 + e^{-\eta_0}}{1 + e^{-\eta_1}} \right).
\]
Also, for squashed hyperbolic tangent, we have:
\[\quad Z = \int_{\eta_0}^{\eta_1} \text{sigmoid}(-2\eta)\, d\eta = \frac{1}{2} \log \left( \frac{1 + e^{-2\eta_0}}{1 + e^{-2\eta_1}} \right).\]
To enable inverse transform sampling, we define the antiderivative \(g(\eta) := \log(1 + e^{-\eta})\), and set \(l_0 = g(\eta_0)\), \(l_1 = g(\eta_1)\). The corresponding cumulative distribution function (CDF) is then:
\[
F(\eta) = \frac{l_0 - g(\eta)}{Z}, \quad \eta \in [\eta_0, \eta_1].
\]

Sampling is performed by drawing \(u \sim \mathcal{U}(0,1)\) and solving for \(\eta\) such that \(F(\eta) = u\). Since \(F(\eta)\) is smooth and strictly monotonic, this inverse can be computed via root-finding or a precomputed lookup table. The method is similar to the implementation in \cite{song2021maximum,zheng2023improved}.

\paragraph{Uniform weighting in $t$-space.}
As discussed in Section \ref{Variance-Aware Likelihood Bounds}, drawing $t$ uniformly from the interval $[0,1]$ amounts to using the proposal density $\rho(\eta) = 1/\alpha_{\eta(t)}^2$.  With this baseline choice the objective in Eq.~\eqref{15} can be optimized in its native form: each Monte-Carlo sample is simply re-weighted by the factor $w(\eta)=\alpha_{t}^{2}\,\eta'(t)$, which bundles the forward‐process weight $\alpha_{t}^{2}$ and the Jacobian $\eta'(t)$.  Although the resulting estimator exhibits higher variance than the importance-sampling schemes introduced earlier, it provides a clear reference point for measuring the effectiveness of less sophisticated designs.

\subsubsection{Learned IS}\label{learnedIS}
The variance of the Monte-Carlo estimator depends on the learned network $\hat{\mathbf{n}}(\cdot;\boldsymbol{\theta})$. To minimize the variance, we can parameterize the IS with another network and treat the variance as an objective \cite{kingma2021variational,zheng2023improved}. Actually, learning $\rho(\eta)$ is equivalent to learning a monotone mapping. Thus, we can uniformly sample $t$, and regard the IS as change-of-variable from $\eta$ to $t$:
\begin{equation}
    \mathcal{L}_{w}(\mathbf{x};\boldsymbol{\psi},\boldsymbol{\theta}) = \mathbb{E}_{t\sim\mathcal{U}(0,1)}\mathbb{E}_{\mathbf{x},\mathbf{n}}\left[\eta'(t;\boldsymbol{\psi})\,\mathcal{L}(\mathbf{x},\mathbf{n},\eta;\boldsymbol{\theta})\right],
\end{equation}
where $\eta'(t;\boldsymbol{\psi}) = d\eta(t;\boldsymbol{\psi})/dt $ and $\eta(t;\boldsymbol{\psi})$ is a monotonic neural network with parameters $\boldsymbol{\psi}$. Since the variance $\text{Var}_{\mathbf{x},\mathbf{n},t}[\eta'(t;\boldsymbol{\psi})\,\mathcal{L}(\mathbf{x},\mathbf{n},\eta;\boldsymbol{\theta})] = \mathbb{E}_{\mathbf{x},\mathbf{n},t}[\big(\eta'(t;\boldsymbol{\psi})\,\mathcal{L}(\mathbf{x},\mathbf{n},\eta;\boldsymbol{\theta})\big)^2] - \big(\mathcal{L}_{w}(\mathbf{x};\boldsymbol{\psi},\boldsymbol{\theta})\big)^2$, where $\mathcal{L}_{w}(\mathbf{x};\boldsymbol{\psi},\boldsymbol{\theta})$ is proved \cite{kingma2021variational} invariant to $\eta(t;\boldsymbol{\psi})$, we can minimize $\mathbb{E}_{\mathbf{x},\mathbf{n},t}[\big(\eta'(t;\boldsymbol{\psi})\,\mathcal{L}(\mathbf{x},\mathbf{n},\eta;\boldsymbol{\theta})\big)^2]$ for variance reduction. At this point, we parameterize $\eta(t;\boldsymbol{\psi})$ similar to \cite{kingma2021variational}. where the network consists of 3 linear layers with weights that are restricted to be positive $l_1, l_2,l_3$, which are composed as $\tilde{\eta}(t;\boldsymbol{\psi}) = l_1(t) + l_3(\text{sigmoid}(l_2(l_1(t))))$. The $l_2$ layer has 1024 outputs, where $l_1$ and $l_3$ have a single output.

We therefore postprocess the monotonic neural network as
\begin{equation}
    \eta(t;\boldsymbol{\psi}) = \eta_0 + (\eta_1 - \eta_0)\frac{\tilde{\eta}(t;\boldsymbol{\psi})-\tilde{\eta}(0;\boldsymbol{\psi})}{\tilde{\eta}(1;\boldsymbol{\psi})-\tilde{\eta}(0;\boldsymbol{\psi})}
\end{equation}
where the constants $\eta_0 = -\log(\text{SNR}_\text{MAX})$, $\eta_1 = -\log(\text{SNR}_\text{MIN})$ define the target SNR range. This construction ensures that $\eta(t;\boldsymbol{\psi})$ remains monotonic and bounded. By adjusting the network parameters $\boldsymbol{\psi}$ we adapt the time-to-noise mapping $t \mapsto \eta(t)$ to match regions with higher loss variance, thereby allocating more samples to informative regions of the diffusion trajectory.

\paragraph{Overhead.}
While this approach seeks the optimal IS, learning $\eta$ adds a lightweight auxiliary network training steps and complex gradient operation through $\eta(t;\boldsymbol{\psi})$, but no extra score-network passes. Hence, we use it mainly to benchmark the optimality of hand-designed IS; more aggressive adaptive schemes following \cite{kingma2023understanding,zheng2023improved} are compatible.

\subsection{Log-SNR-Timed Channel with Different Variance Functions}\label{B.2}

Given the central role of \(\sigma_t^2\) in shaping both the weighting of the loss and the structure of the forward process, it is natural to ask how different functional forms of \(\sigma_t^2(\eta)\) affect the overall behaviour of the model. In particular, we investigate how alternative variance functions under the log-SNR parameterization influence the resulting likelihood bound, optimization dynamics, and numerical stability. This section provides both analytical insights and empirical comparisons across several representative schedules. As state in Section~\ref{Variance-Aware Likelihood Bounds}, using $\eta$ timing and denoising score matching estimator, the variance weighted objectives are reformulated as:
\begin{equation}
    \mathcal{L}_{\text{DSM}}(\sigma_t^2;\boldsymbol{\theta}) = \frac{1}{2} \int_{\sigma_0^2}^{\sigma_1^2} \mathbb{E}_{\mathbf{n}}  \left[\sigma_t^{-2}\|\mathbf{n} - \hat{\mathbf{n}}(\mathbf{y}_t,\eta_t;{\boldsymbol{\theta}})\|_2^2 \right]\ d\sigma_t^2.
\end{equation}
Due to the $\eta$-timed schedule properties, we replace the time subscript with $\eta$, $\alpha_\eta$ and $\sigma_\eta$ are deterministic functions of $\eta$ without any hyperparameters.

\paragraph{Logistic Sigmoid.}  
We begin with the benchmark variance function \( \sigma_\eta^2 = \mathrm{sigmoid}(\eta) := \frac{1}{1 + e^{-\eta}} \), which is widely adopted in VP-based diffusion models~\cite{kingma2021variational}. This schedule maps variance smoothly into the unit interval, with most of the variation concentrated around \( \sigma_\eta^2 = 0 \) \cite{zheng2023improved}. Its symmetry and boundedness make it a natural default choice, but also raise questions about its flexibility in capturing tail behaviour and class imbalance. We have the analytic form of the derivative of $d\sigma_\eta^2= (1-\sigma_\eta^2) \sigma_\eta^2d\eta$:
\begin{equation}
    \mathcal{L}_{\text{DSM}}(\sigma_t^2;\boldsymbol{\theta}) = \frac{1}{2} \int_{\eta_0}^{\eta_1} \mathbb{E}_{\mathbf{n}}  \left[(1-\sigma_\eta^2)\|\mathbf{n} - \hat{\mathbf{n}}(\mathbf{y}_t,\eta_t;{\boldsymbol{\theta}})\|_2^2 \right]\ d\eta.
\end{equation}

\paragraph{Generalized Logistic Sigmoid.}
Then we consider the type I generalized logistic sigmoid function \cite{johnson1995continuous} $\sigma_\eta^2 = \text{sigmoid}^{a}(
\eta
) = (\frac{1}{1+e^{-\eta}})^{a}$ with $a$ is strict positive coefficient. We have $d\sigma_\eta^2= a (1- \text{sigmoid}(\eta)) \sigma_\eta^2d\eta$:
\begin{equation}
    \mathcal{L}_{\text{DSM}}(\sigma_t^2;\boldsymbol{\theta}) = \frac{a}{2} \int_{\eta_0}^{\eta_1} \mathbb{E}_{\mathbf{n}}  \left[\text{sigmoid}(-\eta)\|\mathbf{n} - \hat{\mathbf{n}}(\mathbf{y}_t,\eta_t;{\boldsymbol{\theta}})\|_2^2 \right]\ d\eta.
\end{equation}

\paragraph{Squashed Hyperbolic tangent.} Then we consider the Tanh squash function $\sigma_\eta^2 = \frac{1}{2}(\text{tanh}(\eta) +1)$, we have the analytic form of the derivative of $d\sigma_\eta^2= 2(1-\sigma_\eta^2)  \sigma_\eta^2d\eta$:
\begin{equation}
    \mathcal{L}_{\text{DSM}}(\sigma_t^2;\boldsymbol{\theta}) = \int_{\eta_0}^{\eta_1} \mathbb{E}_{\mathbf{n}}  \left[(1-\sigma_\eta^2)\|\mathbf{n} - \hat{\mathbf{n}}(\mathbf{y}_t,\eta_t;{\boldsymbol{\theta}})\|_2^2 \right]\ d\eta.
\end{equation}

\begin{table}[t]
  \caption{Specification of related values and objectives under VP schedule. With $s(\eta)=\text{sigmoid}(\eta)$, $s_{-} = \text{sigmoid}(-\eta)$, $\kappa = s^{a/2}$.}
  \label{tab:vp_spec}
  \centering
  \small 
  \resizebox{\textwidth}{!}{%
  \begin{tabular}{@{}lcc@{}}
    \toprule
    \textbf{Formula} & \textbf{Generalized Logistic Sigmoid} & \textbf{Squashed Hyperbolic tangent} \\ 
    \midrule
    $\alpha_\eta$ 
      & $\displaystyle\sqrt{1-\left(\frac{1+\text{tanh}(\eta/2)}{2}\right)^a}$ 
      & $\displaystyle\sqrt{\frac{1}{1+\exp(2\eta)}}$ \\[4pt]

    $\sigma_\eta$ 
      & $\displaystyle\kappa(\eta)$ 
      & $\displaystyle\sqrt{\frac{1}{1+\exp(-2\eta)}}$ \\[4pt]

    $\mathcal{L}_{\text{DSM}}$ 
      & $\dfrac{a}{2}\displaystyle\int_{\eta_{0}}^{\eta_{1}}\!
         s_{-}(\eta)\,
         \mathbb{E}_{\mathbf{n}}[
         \|\mathbf{n} - \hat{\mathbf{n}}(\mathbf{y}_t,\eta;{\boldsymbol{\theta}})\|_2^2]\,
         d\eta$
      & $\displaystyle\int_{\eta_{0}}^{\eta_{1}}\!
         s_{-}(2\eta)\,
         \mathbb{E}_{\mathbf{n}}\,[
         \|\mathbf{n} - \hat{\mathbf{n}}(\mathbf{y}_t,\eta;{\boldsymbol{\theta}})\|_2^2]\,
         \,\mathrm{d}\eta$ \\[14pt]

    \bottomrule
  \end{tabular}%
  }
\end{table}

\begin{table}[t]
  \caption{Specification of related values and objectives under SP schedule. With $s(\eta)=\text{sigmoid}(\eta)$, $s_{-} = \text{sigmoid}(-\eta)$.}
  \label{tab:sp_spec}
  \centering
  \small  
  \resizebox{\textwidth}{!}{%
  \begin{tabular}{@{}lcc@{}}
    \toprule
    \textbf{Formula} & \textbf{Generalized Logistic Sigmoid} & \textbf{Squashed Hyperbolic tangent} \\ 
    \midrule
    $\alpha_\eta$ 
      & $\displaystyle1-\sqrt{\text{sigmoid}(\eta)^{a}}$ 
      & $1-\sqrt{\frac{1}{1+\exp(-2\eta)}}$ \\[4pt]

    $\sigma_\eta$ 
      & $\displaystyle\sqrt{\text{sigmoid}(\eta)^{a}}$ 
      & $\sqrt{\frac{1}{1+\exp(-2\eta)}}$ \\[4pt]

    $\mathcal{L}_{\text{DSM}}$ 
      & $\dfrac{a}{2}\displaystyle\int_{\eta_{0}}^{\eta_{1}}\!
         s_{-}(\eta)\,
         \mathbb{E}_{\mathbf{n}}[
         \|\mathbf{n} - \hat{\mathbf{n}}(\mathbf{y}_t,\eta;{\boldsymbol{\theta}})\|_2^2]\,
         d\eta$
      & $\displaystyle\int_{\eta_{0}}^{\eta_{1}}\!
         s_{-}(2\eta)\,
         \mathbb{E}_{\mathbf{n}}\,[
         \|\mathbf{n} - \hat{\mathbf{n}}(\mathbf{y}_t,\eta;{\boldsymbol{\theta}})\|_2^2]\,
         \mathrm{d}\eta$ \\[14pt]

    \bottomrule
  \end{tabular}%
  }
\end{table}

Thus, we can derive their specific objectives and equivalent predictors using the formula for general noise schedules. We summarize them in Table~\ref{tab:vp_spec} and Table~\ref{tab:sp_spec}.

\subsection{Empirical Results}\label{app:experiment}

\begin{table}[t]
  \centering
  \caption{Comparison of NLLs and FID on CIFAR-10 under different noise variance settings.}
  \label{tab:side_by_side}
  \begin{subtable}[t]{0.45\textwidth}
    \centering
    \small
    \caption{NLL bounds with different noise variances}
    \begin{tabular}{lcc}
      \toprule
      \textbf{Likelihood Bounds} & \(\eta_0=-8.7\) & \(\eta_0=-13.3\) \\
      \midrule
      ELBO                       & 3.99           & 2.75            \\
      \textbf{Ours} (SP with IS)          & 2.49           & 2.79            \\
      \textbf{Ours} (VP with IS)          & 2.50           & 2.78            \\
      \bottomrule
    \end{tabular}
    \label{tab:full_bounds}
  \end{subtable}
  \hfill
  \begin{subtable}[t]{0.45\textwidth}
    \centering
    \small
    \caption{FID of CIFAR-10 with different losses}
    \begin{tabular}{lcc}
      \toprule
      \textbf{Loss}            & \(\eta_0=-8.7\) & \(\eta_0=-13.3\) \\
      \midrule
      ELBO                  & 14.60           & 11.9           \\
      \textbf{Ours} (VP)                  & 10.18           & 9.42            \\
      \bottomrule
    \end{tabular}
    \label{tab:endpoint_bounds}
  \end{subtable}
\end{table}

Tables~\ref{tab:schedule_ablation} and~\ref{tab:schedule} report various configurations of the noise variance function and its endpoints in the VP setting. Empirically, increasing the noise variance at the initial stage of the diffusion process yields consistently improved likelihood bounds.

We begin with the benchmark configuration adopted in~\cite{kingma2021variational}, where  $\sigma_t^2 = \mathrm{sigmoid}(\eta(t))$ and the initial negative log-SNR value is set to $\eta_0 = -13.3$.  Under this setting, the model achieves an NLL of 2.81 bits/dim after 300K iterations, which is comparable to the ELBO result of 2.75 reported on CIFAR-10 with the same number of iterations. 

When using $a=0.5$, i.e., $\sigma_t^2 = \mathrm{sigmoid}(\eta)^{1/2}$, the results remain promising. However, for $\eta_0=-8.7$, we note that $\mathrm{sigmoid}(-8.7)^{1/2}\approx1.29\times10^{-2}$, which falls outside the regime where the assumption $\sigma_0^2\!\ll\!1$ holds strictly, thus violating the theoretical conditions underlying our analysis. Nevertheless, this configuration can be regarded as a practical compromise, yielding strong  empirical performance despite the theoretical violation.

We further observe that increasing the warm-up noise variance generally leads to degraded FID scores. This motivates a more systematic investigation of such extreme configurations, as summarized in Tables~\ref{tab:schedule}, \ref{tab:schedule_SP_1}, and~\ref{tab:schedule_SP_2}. Empirically, the likelihood performance appears insensitive to the detailed shape of the variance schedule, but depends primarily on the endpoint values $\sigma_0^2$ and $\sigma_1^2$. These findings suggest that careful tuning of the endpoint variances can have a larger effect than modifying the overall schedule shape, particularly for likelihood-oriented objectives, an observation consistent with prior analyses in~\cite{ho2020denoising,kim2022soft,kingma2021variational}.

\begin{table}[t]
\centering
\caption{
Performance comparison across generalized sigmoid schedules with varying exponent \(a\), initial log-SNR \(\eta_0\), and importance sampling strategies on CIFAR-10. Bold denotes the best NLL and FID within each group.
}
\label{tab:schedule_ablation}
\renewcommand{\arraystretch}{1.2}
\setlength{\tabcolsep}{6pt}
\begin{tabular}{@{}cllccl@{}}
\toprule
\textbf{Schedule \(a\)} & \textbf{\(\eta_0\)} & \textbf{Sampling Strategy} & \textbf{NLL (↓)} & \textbf{FID (↓)} \\
\midrule

$a=1$
  & $-13.3$
  & Uniform in \(t\)            & 2.81 & 9.79 &  \\
  &                            & IS: Learned Importance    & \textbf{2.75} & \textbf{9.12}  \\
  &                            & IS: Uniform in \(\eta\)         & 2.78 & 9.42\\

\midrule

$a=1$
  & $-8.7$
  & Uniform in \(t\)            & 2.62 & 14.9  \\
  &                            & IS: Learned Importance                & \textbf{2.50} & \textbf{10.05} \\
  &                            & IS: Uniform in \(\eta\)           & 2.50 & 10.18  \\
\midrule$a=2$
  & $-13.3$
  & Uniform in \(t\)            & 3.62 & /  \\
  &                            & IS: Learned Importance                & \textbf{3.47} & /  \\
  &                            & IS: Uniform in \(\eta\)           & 3.51 & / \\
\midrule
$a=2$
  & $-8.7$
  & Uniform in \(t\)            & 3.16 & /  \\
  &                            & IS: Learned Importance                & 3.05 & / \\
  &                            & IS: Uniform in \(\eta\)           & \textbf{3.04} & /  \\

\midrule$a=0.5$
  & $-13.3$
  & Uniform in \(t\)            & 2.54 & / \\
  &                            & IS: Learned Importance                & \textbf{2.45} & /  \\
  &                            & IS: Uniform in \(\eta\)           & 2.47 & / \\

\midrule

$a=0.5$
  & $-8.7$
  & Uniform in \(t\)            & 2.61 & / \\
  &                            & IS: Learned Importance                & \textbf{2.31} & / \\
  &                            & IS: Uniform in \(\eta\)           & 2.32 & / \\

\bottomrule
\end{tabular}
\end{table}

\begin{table}[t]
\centering
\caption{
Performance comparison across squashed Hyperbolic tangent schedules with log-SNR endpoints \(\eta_0\), $\eta_1$, and importance sampling strategies on CIFAR-10 with VP. Bold denotes the best NLL within each group.
}
\label{tab:schedule}
\renewcommand{\arraystretch}{1.2}
\setlength{\tabcolsep}{6pt}
\begin{tabular}{@{}cllccl@{}}
\toprule
\textbf{\(\eta_0\)} & \textbf{\(\eta_1\)} & \textbf{Sampling Strategy} & \textbf{NLL (↓)}  \\
\midrule

$-13.3$
  &$5$
  & Uniform in \(t\)            & 3.83 &  \\
  &                            & IS: Learned Importance    & \textbf{3.77}  \\
  &                            & IS: Uniform in \(\eta\)         & 3.79 \\

\midrule
$-8.7$
  & $5$
  & Uniform in \(t\)            & 3.35  &  \\
  &                            & IS: Learned Importance    & \textbf{3.15} \\
  &                            & IS: Uniform in \(\eta\)         & 3.18 \\
\midrule
$-4.33$
  & $2.5$
  & Uniform in \(t\)            & 2.61  &  \\
  &                            & IS: Learned Importance    & \textbf{2.51} \\
  &                            & IS: Uniform in \(\eta\)         & 2.53 \\

\bottomrule
\end{tabular}
\end{table}

\begin{table}[t]
\centering
\caption{
Performance comparison across generalized sigmoid schedules with varying exponent \(a\), initial log-SNR \(\eta_0\), and importance sampling strategies on CIFAR-10 with SP. Bold denotes the best NLL and FID within each group.
}
\label{tab:schedule_SP_1}
\renewcommand{\arraystretch}{1.2}
\setlength{\tabcolsep}{6pt}
\begin{tabular}{@{}cllccl@{}}
\toprule
\textbf{\(a\)} & \textbf{\(\eta_0\)} & \textbf{Sampling Strategy} & \textbf{NLL (↓)}  \\
\midrule

$a=1$
  & $-13.3$
  & Uniform in \(t\)            & 2.82 &  \\
  &                            & IS: Learned Importance    & \textbf{2.77}  \\
  &                            & IS: Uniform in \(\eta\)         & 2.79 \\

\midrule
$a=1$
  & $-8.7$
  & Uniform in \(t\)            & 2.62 &  \\
  &                            & IS: Learned Importance    & \textbf{2.49}  \\
  &                            & IS: Uniform in \(\eta\)         & 2.50 \\

\midrule
$a=0.5$
  & $-13.3$
  & Uniform in \(t\)            & 2.53  &  \\
  &                            & IS: Learned Importance    & 2.47 \\
  &                            & IS: Uniform in \(\eta\)         & \textbf{2.45} \\
\midrule
$a=0.5$
  & $-8.7$
  & Uniform in \(t\)            & 2.61  &  \\
  &                            & IS: Learned Importance    & \textbf{2.31} \\
  &                            & IS: Uniform in \(\eta\)         & 2.33 \\

\bottomrule
\end{tabular}
\end{table}

\begin{table}[t]
\centering
\caption{
Performance comparison across squashed Hyperbolic tangent schedules with log-SNR endpoints \(\eta_0\), $\eta_1$, and importance sampling strategies on CIFAR-10 with SP. Bold denotes the best NLL within each group.
}
\label{tab:schedule_SP_2}
\renewcommand{\arraystretch}{1.2}
\setlength{\tabcolsep}{6pt}
\begin{tabular}{@{}cllccl@{}}
\toprule
\textbf{\(\eta_0\)} & \textbf{\(\eta_1\)} & \textbf{Sampling Strategy} & \textbf{NLL (↓)}  \\
\midrule

$-13.3$
  & $5$
  & Uniform in \(t\)            & 3.82 &  \\
  &                            & IS: Learned Importance    & 3.76  \\
  &                            & IS: Uniform in \(\eta\)         & \textbf{3.73} \\

\midrule
$-8.7$
  & $5$
  & Uniform in \(t\)            & 3.32  &  \\
  &                            & IS: Learned Importance    & \textbf{3.14} \\
  &                            & IS: Uniform in \(\eta\)         & 3.18 \\
\midrule
$-4.33$
  & $2.5$
  & Uniform in \(t\)            & 2.60  &  \\
  &                            & IS: Learned Importance    & \textbf{2.51} \\
  &                            & IS: Uniform in \(\eta\)         & 2.51 \\

\bottomrule
\end{tabular}
\end{table}

\subsection{Likelihood Estimation Comparison for IT-bound and ELBO on CIFAR-10}\label{app:elbo_sensivitity}

\begin{table}[t]
\centering
\caption{Decomposition of ELBO and Information-Theoretic Bound (IT) with initial endpoint $\eta_0=-13.3$. Values reported on CIFAR-10 in 310K iterations (↓ lower is better).}
\label{tab:elbo_it_decomposition_1}
\renewcommand{\arraystretch}{1.3}
\begin{tabular}{lll}
\toprule
\textbf{Term} & \textbf{Description} & \textbf{Value (bits/dim)} \\
\midrule
& \textit{Evidence Lower Bound (ELBO)} & \textit{Total = 2.794} \\
\midrule
Term 1 & KL divergence between $p(\mathbf{y}_1\vert\mathbf{x})$ and prior $\pi(\mathbf{x})$ & 0.0012358 \\
Term 2 & Reconstruction loss ($-\mathbb{E}_{p(\mathbf{y}_0\vert\mathbf{x})}[\log q(\mathbf{x}\vert\mathbf{y}_0)]$) & 0.0103869 \\
Term 3 & Diffusion loss ($\sum \mathbb{E}_{p(\mathbf{y}_{t(i)}\vert\mathbf{x})}D_{\text{KL}}[p(\mathbf{y}_{s(i)}\vert\mathbf{y}_{t(i)},\mathbf{x})\|q(\mathbf{y}_{s(i)}\vert\mathbf{y}_{t(i)};\boldsymbol{\theta})]$) & 2.7836967 \\
\midrule
& \textit{Information-Theoretic Bound (IT)} & \textit{Total = 2.805} \\
\midrule
Term 1 & Integrated denoising score matching term & 0.7622061 \\
Term 2 & Mismatched entropy between $p(\mathbf{y}_1\vert\mathbf{x})$ and prior $\pi(\mathbf{x})$) & 2.0434874 \\
\bottomrule
\end{tabular}
\end{table}

Tables~\ref{tab:elbo_it_decomposition_1} and~\ref{tab:elbo_it_decomposition} report the average loss components of models trained on CIFAR-10. We observe that all IT-based models employing a standard sigmoid variance schedule achieve likelihood estimates that match or surpass those of the ELBO baseline. Notably, when combined with importance sampling, the IT bound consistently outperforms the ELBO across all configurations. However, when using uniform sampling in $t$-space, the IT-bound appears to trade off a small amount of likelihood for faster convergence, likely due to reduced variance in gradient estimation. We further observe that increasing the warm-up noise variance leads to a notable rise in the total ELBO. This increase is largely driven by the reconstruction term,
\begin{equation}
-\mathbb{E}_{p(\mathbf{y}_0\vert\mathbf{x})}[\log q(\hat{\mathbf{x}}\vert\mathbf{y}_0)],
\end{equation}
which becomes more pronounced as $p(\mathbf{y}_0\vert\mathbf{x})$ flattens under higher noise levels, making accurate reconstruction more challenging.

Intuitively, the loss structure is grounded in the auto-encoder paradigm \cite{Kingma2013AutoEncodingVB}, where one commonly assumes a factorised posterior $q(\hat{\mathbf{x}}\vert \mathbf{y}_0)$, implying conditional independence across the elements of $\hat{\mathbf{x}}$. In the context of image data, this assumption translates into the belief that pixel values are conditionally independent given the latent code $\mathbf{y}_0$, and each pixel depends only on its corresponding latent component. It follows the form:
\begin{align}
    q(\hat{x}_i|y_{0,i}) \propto p(y_{0,i}|x_i) = \mathcal{N}(y_{0,i};\alpha_0x_i,\sigma_0^2),
\end{align}
where we normalize over all possible values of $x_i$. However, we argue that this assumption is overly restrictive and does not hold in practice. Within the auto-encoding framework, $q(\hat{\mathbf{x}} \vert \mathbf{y}_0)$ can be interpreted as a decoder tasked with reconstructing the data. The approximation may hold when the SNR at $t=0$ is sufficiently high. In this case, the conditional distribution $p(\mathbf{y}_0\vert\mathbf{x})$ becomes sharply peaked around $\mathbf{y}_0 = \alpha_0 \mathbf{x}$, effectively imposing strong constraints on the reconstruction loss. Specifically, it induces high sensitivity to noise: small deviations in the warm-up noise level can significantly affect the fidelity of reconstructions. Moreover, since likelihood estimation is inherently sensitive to the probability of individual pixel values and fine-grained image details, this modelling choice represents a practical compromise rather than a principled solution. In effect, the assumption simplifies computation at the expense of capturing complex dependencies among pixels, which can be critical for accurate reconstruction and reliable likelihood evaluation.

Increasing the warm-up noise smooths the sharp concentration of $p(\mathbf{y}_0 | \mathbf{x})$, easing reconstruction by reducing the signal-to-noise ratio. In the VP setting, this implies $\alpha_0$ deviates more from 1, yielding more stable behaviour. Crucially, Theorem~\ref{theorem 1} ensures robustness under arbitrary noise, allowing principled tuning of the starting noise to balance stability and accuracy. See Appendix.~\ref{app:dequantization} for more related discussions.

Furthermore, we observe that small perturbations may fail to regularise extreme pixel values introduced during dequantisation, preventing effective mapping into a smooth continuous domain. Slightly increasing the initial noise helps suppress such outliers, resulting in improved likelihood behaviour and more stable training. Moreover, in score-based models with noise prediction, a very small starting noise can impair residual estimation near $t = 0$, degrading likelihood due to poor signal-noise separation. We conjecture that adopting a velocity-based parameterisation \cite{nielsen2024diffenc,zheng2023improved} ($\mathbf{v}$-network) may alleviate this issue.

\begin{table}[t]
\centering
\caption{Decomposition of ELBO and Information-Theoretic Bound (IT) with initial endpoint $\eta_0=-8.7$. Values reported on CIFAR-10 in 310K iterations (↓ lower is better).}
\label{tab:elbo_it_decomposition}
\renewcommand{\arraystretch}{1.3}
\begin{tabular}{lll}
\toprule
\textbf{Term} & \textbf{Description} & \textbf{Value (bits/dim)} \\
\midrule
& \textit{Evidence Lower Bound (ELBO)} & \textit{Total = 3.99} \\
\midrule
Term 1 & KL divergence between $p(\mathbf{y}_1\vert\mathbf{x})$ and prior $\pi(\mathbf{x})$ & 0.0012358 \\
Term 2 & Reconstruction loss ($-\mathbb{E}_{p(\mathbf{y}_0\vert\mathbf{x})}[\log q(\mathbf{x}\vert\mathbf{y}_0)]$) & 2.7219771 \\
Term 3 & Diffusion loss ($\sum \mathbb{E}_{p(\mathbf{y}_{t(i)}\vert\mathbf{x})}D_{\text{KL}}[p(\mathbf{y}_{s(i)}\vert\mathbf{y}_{t(i)},\mathbf{x})\|q(\mathbf{y}_{s(i)}\vert\mathbf{y}_{t(i)};\boldsymbol{\theta})]$) & 1.2691765 \\
\midrule
& \textit{Information-Theoretic Bound (IT)} & \textit{Total = 2.51} \\
\midrule
Term 1 & Integrated denoising score matching term & 0.4663643 \\
Term 2 & Mismatched entropy between $p(\mathbf{y}_1\vert\mathbf{x})$ and prior $\pi(\mathbf{x})$) & 2.0435017 \\
\bottomrule
\end{tabular}
\end{table}

\section{Samples Quality and FID}\label{App:FID}

\begin{table}[t]
\centering
\caption{
Likelihood in bits per dimension (BPD) and sample quality (FID scores) on CIFAR-10 and ImageNet-32, for vanilla VDM, MuLAN and ours. "\textsuperscript{†}" indicates the result from \cite{sahoo2024diffusion} for 10K samples generated using an adaptive-step ODE solver.
}
\label{tab:vlb_fid}
\small
\setlength{\tabcolsep}{6pt}
\begin{tabular}{lccccccc}
\toprule
\textbf{Model} & \multicolumn{3}{c}{\textbf{CIFAR-10}} & \multicolumn{3}{c}{\textbf{ImageNet-32}} \\
\cmidrule(lr){2-4} \cmidrule(lr){5-7}
& Steps & VLB ($\downarrow$) & FID ($\downarrow$) & Steps & VLB ($\downarrow$) & FID ($\downarrow$) \\
\midrule
VDM \cite{kingma2021variational}     & 10M & 2.65 & \textbf{7.6} & 2M & 3.72 & 14.26\textsuperscript{†}  \\
+ MuLAN \cite{sahoo2024diffusion}     & 2M & 2.65 & 18.54 & 2M & 3.71 & \textbf{13.19}  \\
\textbf{Ours} ($\eta_0 = -13.3$)                     & 0.3M & 2.78 & 9.42 & 0.3M & 3.28 & 13.80 \\
\textbf{Ours} ($\eta_0 = -8.7$)                     & 0.3M & \textbf{2.50} & 10.18 & 0.3M & \textbf{3.01} & 14.76  \\
\bottomrule
\end{tabular}
\end{table}

\begin{table}[t]
\centering
\caption{
Comparison of the mean FID scores with standard error for our model on CIFAR-10 with $\text{sigmoid}$ noise schedule after 0.3M steps. We provide both FID scores on 10K and 50K samples and with respect to both train and test set.
}
\label{tab:fid_comparison}
\small
\setlength{\tabcolsep}{10pt}
\begin{tabular}{lcccc}
\toprule
\textbf{Model} & \textbf{FID 10K train} & \textbf{FID 10K test} & \textbf{FID 50K train} & \textbf{FID 50K test} \\
\midrule
\textbf{Ours} ($\eta_0 = -8.7$)        & $11.85 \pm 0.2$ & $12.91 \pm 0.2$ & $10.18 \pm 0.2$ & $10.90 \pm 0.2$ \\
\textbf{Ours} ($\eta_0 = -13.3$)    & $10.16 \pm 0.2$ & $11.53 \pm 0.3$ & $9.41 \pm 0.2$ & $9.50 \pm 0.2$ \\
DiffEnc \cite{nielsen2024diffenc}  & $14.6 \pm 0.8$ & $18.5 \pm 0.7$ & $11.1 \pm 0.8$ & $15.0 \pm 0.7$ \\
\bottomrule
\end{tabular}
\end{table}

\begin{table}[t]
\centering
\caption{
Comparison of the mean FID scores with standard error for our model on ImageNet-32 with $\text{sigmoid}$ noise schedule after 0.3M steps. We provide both FID scores on 10K and 50K samples and with respect to both train and test set.
}
\label{tab:fid_comparison_im}
\small
\setlength{\tabcolsep}{10pt}
\begin{tabular}{lcccc}
\toprule
\textbf{Model} & \textbf{FID 10K train} & \textbf{FID 10K test} & \textbf{FID 50K train} & \textbf{FID 50K test} \\
\midrule
\textbf{Ours} ($\eta_0 = -8.7$)        & $16.90 \pm 0.2$ & $17.91 \pm 0.2$ & $14.72 \pm 0.2$ & $14.76 \pm 0.2$ \\
\textbf{Ours} ($\eta_0 = -13.3$)    & $15.76 \pm 0.2$ & $16.15 \pm 0.3$ & $13.21 \pm 0.2$ & $13.80 \pm 0.2$ \\
\bottomrule
\end{tabular}
\end{table}

Just as selecting appropriate training and optimization strategies is necessary to achieve strong performance in a given application, so too is the choice of evaluation metric pivotal for drawing valid conclusions. We must stress that the primary focus of this paper is on maximizing the likelihood learning metric, specifically, negative log‑likelihood measured in bits‑per‑dimension (BPD; lower is better), rather than on optimizing Fréchet Inception Distance (FID) or Inception Score (IS), for the following reasons.

Model samples undoubtedly serve as a valuable diagnostic tool, often enabling us to form an intuition about why a model may underperform and how it might be improved. From this standpoint, a generative model ought to produce samples that are indistinguishable from those in the training set, whilst encompassing its full variability. To quantify these properties, a variety of metrics, such as the IS and the FID, have been proposed. However, both qualitative and quantitative assessments based on model samples can be misleading with respect to a model’s density-estimation capabilities, as well as its effectiveness in probabilistic modelling tasks beyond image synthesis \cite{Theis2015ANO}. Consequently, average log‑likelihood remains the de facto standard for quantifying generative image‑modeling performance. For many sophisticated models, the average log‑likelihood is challenging to compute or even approximate. Indeed, it is possible for a model with sub‑optimal log‑likelihood to generate visually impressive samples, or conversely, for a model with excellent log‑likelihood to produce poor samples, an observation that underlines the lack of a direct relationship between FID and negative log‑likelihood (NLL).

From an information‑theoretic standpoint, it is well known that maximizing the log‑likelihood of a probabilistic model is equivalent to minimizing the KL divergence from the data distribution to the model distribution. By contrast, FID operates by fitting multivariate Gaussian distributions to the embeddings of real and generated images and then measuring their discrepancy via the Fréchet distance (equivalently, the 2‑Wasserstein or Earth Mover’s distance). Clearly, the mathematical formulations of these two metrics diverge fundamentally: one corresponds to a mismatched estimation problem under a KL‑based criterion, while the other embodies an optimal‑transport task.

Moreover, FID conflates both fidelity to the real data distribution and the diversity of generated samples into a single score, and its absolute value is highly sensitive to myriad factors, ranging from the number of samples and the particular checkpoint of the feature extractor network to low‑level image‑processing choices. Consequently, the visual appeal of generated images, as quantified by FID, correlates only imperfectly with a model’s log‑likelihood performance. In our work, we concentrate on advancing the state of the art in likelihood estimation; although we report FID scores for completeness, we leave the optimization of sample quality to future research.

Although our model was not explicitly optimized for perceptual sample quality, we report FID scores, a standard metric for visual realism, for both our model and VDM on CIFAR-10 and ImageNet-32 (Tables~\ref{tab:vlb_fid}, \ref{tab:fid_comparison}, and~\ref{tab:fid_comparison_im}). From these results, we observe that both models achieve comparable FID scores across datasets. Importantly, the reported values vary substantially depending on the number of generated samples and whether FID is computed against the training or test set. As expected, using more samples improves FID, and evaluation against the training set consistently yields better scores, likely due to closer distributional alignment.

Furthermore, among models with similar likelihood performance, such as i‑DODE \cite{zheng2023improved}, MuLAN \cite{sahoo2024diffusion} and DiffEnc \cite{nielsen2024diffenc}, our method not only achieves the best negative log‑likelihood but also retains the lowest FID despite requiring substantially fewer training iterations. Conversely, methods like W‑PCDM \cite{li2024likelihood} that explicitly optimize for FID exhibit a marked reduction in likelihood performance.

Consistent with prior observations~\cite{song2021maximum,ho2020denoising,nichol2021improved}, we find that models achieving better log-likelihood often exhibit slightly worse FID scores. Nevertheless, we emphasize that this degradation in FID is minor, and qualitatively, the generated samples from both models are visually indistinguishable (see Figs.~\ref{fig:VDMEMA}, \ref{fig:ISITEMA}, \ref{fig:VDM_IM} and \ref{fig:ISIT_IM}).

\section{Experimental Settings}\label{Appendix:experimental}

\paragraph{Datasets}

We perform all experiments on CIFAR-10 and ImageNet datasets. CIFAR-10 contains 50,000 training and 10,000 test images. The ImageNet variant includes 1,281,149 training and 49,999 test images. Among the two known versions of ImageNet32, we adopt the newer, anti-aliased version~\cite{lipmanflow}, which facilitates likelihood training and remains publicly available. The older version used in~\cite{song2021maximum,kingma2021variational} is no longer accessible. Furthermore, it is notable that ImageNet contains some personal sensitive information and may cause privacy concern \cite{song2021maximum}.

\paragraph{Model Architectures} Our model architecture closely follows the design of Variational Diffusion Models (VDMs)~\cite{kingma2021variational}. Specifically, we adopt the original U-Net backbone from VDM for pixel-space diffusion without modification. Our diffusion model is parameterized in terms of the $\eta$-timed normalized noise predictor. This architecture is optimized for likelihood-based training and includes key design choices such as the removal of internal downsampling and upsampling, and the use of Fourier feature embeddings to improve fine-scale detail prediction. Consistent with VDM’s dataset-dependent configurations, we use a U-Net of depth 32 with 128 channels for CIFAR-10, and 256 channels for ImageNet-32. Our model for ImageNet-64 and -128 uses double the depth at 64 ResNet layers in both the forward and backward direction in the U-Net. It also uses a constant number of channels of 256. All models apply a dropout rate of 0.1 in intermediate layers.

\paragraph{Hardware} For the ImageNet-64 and -128 experiments, we used a single GPU node with 8 A800s or 8 H20-NVLink. For the CIFAR-10 and ImageNet-32 experiments, the models were trained and evaluated on 4 GPUs spanning several GPUs types like V100, L20s, A40s, and 3090s with float32 precision.

\paragraph{Training} We follow the same default training settings as \cite{kingma2021variational}. For all our experiments, we use the Adam optimizer with learning rate $2 \times 10^{-4}$, exponential decay rates of $\beta_1 = 0.9$, $\beta_2 = 0.99$ and decoupled weight decay coefficient of 0.01. We also maintain an exponential moving average (EMA) of model parameters with an EMA rate of 0.9999 for evaluation.

For CIFAR-10, the training processes are conducted on a cluster of 4 GPU cards of NVIDIA V100s. We pretrain the model for 0.3 million iterations using a batch size of 128, which takes around 38 hours. Then we finetune the model for 1K iterations using a batch size of 256 and accumulate the gradient for every 4 batches. Note that in related works \cite{lipmanflow}, experiments on ImageNet-32 (new version) are conducted at a larger batch size (512 or 1024), which may improve the results. For ImageNet-32, the training processes are conducted on 4 GPU cards of NVIDIA A800 (80GB). We pretrain the model for 0.3 million iterations using a batch size of 512, which takes around 3 days. Then we finetune the model for 1K iterations using a batch size of 1024 and accumulate the gradient for every 4 batches.

\paragraph{FID} 
We report Fréchet Inception Distance (FID) scores computed on 50,000 generated samples, unless otherwise noted. This follows the standard setup used in \cite{kingma2021variational}, with ancestral sampling over 1,000 sampling timesteps. FID is evaluated against both the training and test sets for CIFAR-10 and ImageNet-32. While increasing the warm-up noise level may slightly increase FID scores, we find that the visual quality of generated samples remains comparable (see Figs.~\ref{fig:VDMEMA}, \ref{fig:ISITEMA}, \ref{fig:VDM_IM} and \ref{fig:ISIT_IM}).

\section{Consistency Across Predictors and Corresponding Objectives}

As discussed in \cite{kingma2021variational}, the diffusion model can be interpreted from three distinct perspectives: as a denoising process, a noise prediction model, and a score-based model. Similarly, our model admits four equivalent parameterizations, with the velocity-based $\mathbf{v}_t$-prediction \cite{salimans2022progressive,zheng2023improved} included in addition to the three canonical forms, which is summarized in Table~\ref{tab:predictor_equivalence}.

\paragraph{Remark} Let \( \mathbf{v}_t = \alpha_t \mathbf{n} - \sigma_t \mathbf{x} \) and define the instantaneous velocity as \( \tilde{\mathbf{v}} = \dot{\alpha}_t \mathbf{x} + \dot{\sigma}_t \mathbf{n} \). We consider four types of predictors, each parameterized by $\boldsymbol{\theta}$, along with their corresponding matching objectives. Each loss is weighted by a positive time-dependent function $w(t)$:

\begin{itemize}
    \item \textbf{Score predictor} \( \hat{\boldsymbol{s}}(\mathbf{y}_t; \boldsymbol{\theta}) \) with likelihood-weighted score matching loss \cite{song2021maximum,song2020score}:
    \[
    \mathcal{J}_{\text{SM}}(w(t);\boldsymbol{\theta}) := \mathbb{E}_t\left[w(t)\,\mathbb{E}_{\mathbf{x},\mathbf{n}}\left[\|\nabla\log p(\mathbf{y}_t)-\hat{\boldsymbol{s}}(\mathbf{y}_t;\boldsymbol{\theta})\|_2^2\right]\right].
    \]
    
    \item \textbf{Noise predictor} \( \hat{\mathbf{n}}(\mathbf{y}_t; \boldsymbol{\theta}) \) with standard noise-matching loss \cite{kingma2021variational,ho2020denoising,nichol2021improved}:
    \[
    \mathcal{J}_{\text{NL}}(w(t);\boldsymbol{\theta}) := \mathbb{E}_t\left[w(t)\,\mathbb{E}_{\mathbf{x},\mathbf{n}}\left[\|\mathbf{n}-\hat{\mathbf{n}}(\mathbf{y}_t; \boldsymbol{\theta})\|_2^2\right]\right].
    \]
    
    \item \textbf{Data predictor} \( \hat{\mathbf{x}}(\mathbf{y}_t; \boldsymbol{\theta}) \) with reconstruction-based data-matching loss \cite{song2021denoising}:
    \[
    \mathcal{J}_{\text{DL}}(w(t);\boldsymbol{\theta}) := \mathbb{E}_t\left[w(t)\,\mathbb{E}_{\mathbf{x},\mathbf{n}}\left[\|\mathbf{x}-\hat{\mathbf{x}}(\mathbf{y}_t; \boldsymbol{\theta})\|_2^2\right]\right].
    \]

    \item \textbf{Velocity predictor} \( \tilde{\mathbf{v}}(\mathbf{y}_t; \boldsymbol{\theta}) \) with flow-matching loss \cite{lu2022maximum,zheng2023improved}:
    \[
    \mathcal{J}_{\text{FM}}(w(t);\boldsymbol{\theta}) := \mathbb{E}_t\left[w(t)\,\mathbb{E}_{\mathbf{x},\mathbf{n}}\left[\|\tilde{\mathbf{v}}-\tilde{\mathbf{v}}(\mathbf{y}_t;\boldsymbol{\theta})\|_2^2\right]\right].
    \]
\end{itemize}

Under the Gaussian forward process, the optimal solutions to these objectives are analytically related, and yield equivalent predictors when appropriately reparameterized. Specifically, they are equivalent by the following relations:
\begin{align}
    \hat{\mathbf{n}}^{*}(\mathbf{y}_t;\boldsymbol{\theta}) &= - \sigma_t\boldsymbol{s}^{*}(\mathbf{y}_t;\boldsymbol{\theta}) = - \sigma_t\nabla_{\mathbf{y}_t}\log p(\mathbf{y}_t) \\
    \hat{\mathbf{x}}^{*}(\mathbf{y}_t;\boldsymbol{\theta}) &= \frac{\mathbf{y}_t - \sigma_t \hat{\mathbf{n}}^{*}(\mathbf{y}_t;\boldsymbol{\theta})}{\alpha_t} \\
    \hat{\mathbf{v}}_t^*(\mathbf{y}_t,\boldsymbol{\theta}) &= \frac{\alpha_t^2 + \sigma_t^2}{\alpha_t} \hat{\mathbf{n}}^{*}(\mathbf{y}_t;\boldsymbol{\theta}) - \frac{\sigma_t}{\alpha_t}\mathbf{y}_t\\
    \hat{\tilde{\mathbf{v}}}^{*}(\mathbf{y}_t;\boldsymbol{\theta}) &= \frac{\dot{\alpha}_t}{\alpha_t}\mathbf{y}_t + \left(\dot{\sigma}_t - \frac{\dot{\alpha}_t\sigma_t}{\alpha_t}\right) \hat{\mathbf{n}}^{*}(\mathbf{y}_t;\boldsymbol{\theta}).
\end{align}

The predictor $\hat{\mathbf{v}}_t^*(\mathbf{y}_t,\boldsymbol{\theta})$ defined as a linear combination of noise and data components, represents a static velocity target in the latent space. In contrast, the instantaneous flow $\hat{\tilde{\mathbf{v}}}^{*}(\mathbf{y}_t;\boldsymbol{\theta}) = d\mathbf{y}_t/dt$ arises from differentiating the forward process with respect to time. The two are related via the temporal dynamics of $\alpha_t$ and $\sigma_t$,  and coincide when the process is linear and velocity is time-invariant.

\begin{table}[t]
\centering
\caption{Analytical relationships between optimal predictors under the Gaussian forward process.}
\label{tab:predictor_equivalence}
\renewcommand{\arraystretch}{1.4}
\setlength{\tabcolsep}{8pt}
\begin{tabular}{llll}
\toprule
\textbf{Predictor Type} & \textbf{Symbol} & \textbf{Optimal Expression} & \textbf{Expressed via} \\
\midrule
Score & \( \boldsymbol{s}^{*}(\mathbf{y}_t) \) & \( \nabla_{\mathbf{y}_t} \log p(\mathbf{y}_t) \) & Score matching \\
Noise & \( \hat{\mathbf{n}}^{*}(\mathbf{y}_t) \) & \( -\sigma_t\, \boldsymbol{s}^{*}(\mathbf{y}_t) \) & Noise prediction \\
Data  & \( \hat{\mathbf{x}}^{*}(\mathbf{y}_t) \) & \( \dfrac{\mathbf{y}_t - \sigma_t \hat{\mathbf{n}}^{*}(\mathbf{y}_t)}{\alpha_t} \) & Denoising reconstruction \\
Velocity & \( \hat{\tilde{\mathbf{v}}}^{*}(\mathbf{y}_t) \) & \( f(t)\mathbf{y}_t - \frac{1}{2}g^2(t) \boldsymbol{s}^{*}(\mathbf{y}_t,t) \) & Flow parameterisation \\
Velocity (alt.) & \( \hat{\tilde{\mathbf{v}}}^{*}(\mathbf{y}_t) \) & 
\( \dfrac{\dot{\alpha}_t}{\alpha_t}\, \mathbf{y}_t - \left(\dot{\sigma}_t - \dfrac{\dot{\alpha}_t \sigma_t}{\alpha_t} \right) \sigma_t\, \boldsymbol{s}^{*}(\mathbf{y}_t) \) & Score-based ODE \\
\bottomrule
\end{tabular}
\end{table}

\section{Numerical Stability}

Finite-precision arithmetic is fragile for terms of the form \(1-\varepsilon\). In our discrete-time objective, several intermediates are extremely close to one (e.g., cumulative coefficients and survival factors). With a naïve float32 implementation, these values can round to exactly \(1\), corrupting the computation and yielding incorrect losses/gradients. Prior discrete-time diffusion implementations \cite{ho2020denoising} used float64 to sidestep such issues. In contrast, our formulation is numerically stable enough that float64 is unnecessary; standard float32 suffices.

A numerically problematic term, for example, is the sampling variance \(\sigma^2_{t\mid s}\). It is straightforward \cite{kingma2021variational} to verify that
\begin{equation}
\sigma^2_{t\mid s}
= -\,\mathrm{expm1}\!\big(\mathrm{softplus}(\gamma(s)) - \mathrm{softplus}(\gamma(t))\big),
\label{eq:sigma_ts}
\end{equation}
where \(\mathrm{expm1}(x)\equiv e^{x}-1\) and \(\mathrm{softplus}(x)\equiv \log(1+e^{x})\) are numerically stable primitives in common numerical computing packages. Evaluating \(\sigma^2_{t\mid s}\) via \eqref{eq:sigma_ts} avoids catastrophic cancellation near \(1\) and keeps computations stable in float32.

\section{Randomized Distribution Smoothing and Dequantization}\label{app:dequantization}

Modern generative models often lean on the \emph{manifold hypothesis}~\cite{deconvergence,meng2021improved}: real-world high-dimensional data concentrate near a low-dimensional manifold. When the hypothesis holds exactly, the data distribution is singular with respect to the ambient Lebesgue measure and its density is not well defined. When it holds approximately, only points in a thin neighborhood of the manifold carry non-negligible mass; elsewhere the density is near zero. Consequently, any ambient-space density that tries to fit such data must exhibit sharp transitions (large first-order derivatives, \textit{i.e.,} a large, possibly unbounded, Lipschitz constant), which is notoriously challenging for likelihood-based models.

Furthermore, while natural images are typically stored using 8-bit integers, they are often modeled using densities, \textit{i.e.}, an image is treated as an instance of a continuous random variable. Since the discrete data distribution has differential entropy of negative infinity, this can lead to arbitrary high likelihoods even on test data. To avoid this case, it is becoming best practice to add real-valued noise to the integer pixel values to dequantize the data.

To this end, we propose to address such issues in the density estimation problem via a \textit{warm-start} process. Inspired by the recent success of randomized smoothing techniques in adversarial defense and distribution smoothing \cite{pmlr-v97-cohen19c,meng2021improved}, we propose to apply randomized smoothing to diffusion generative modeling.

\paragraph{Randomized Distribution Smoothing}
Unlike \cite{pmlr-v97-cohen19c} where randomized smoothing is applied to a model, and \cite{meng2021improved} where symmetric random noise is applied to the data distribution, we inject the arbitrary randomized smoothing into both data $p(\mathbf{x})$ and model $q(\hat{\mathbf{x}};\boldsymbol{\theta})$. Specifically, we convolve an arbitrary isotropic noise distribution with the data distribution and model to obtain the new “smoother” distributions. By choosing an appropriate smoothing distribution, we aim to make warm start process easier than the original learning problem: smoothing facilitates learning in the first stage by making the input distribution fully supported without sharp transitions in the density function; generating a sample given a noisy one is easier than generating a sample from scratch.

\paragraph{Dequantization Mismatch}
Another representative issue in dequantization \cite{song2021maximum} method with diffusion models is training-evaluation mismatch. During training, each datapoint is treated as a narrow Gaussian (or logit-normal) \cite{kingma2021variational,nielsen2024diffenc} centred on the original value, while evaluation typically occurs on data perturbed with uniform noise \cite{Theis2015ANO}. This inconsistency introduces a distributional shift between training and test likelihood evaluation. Variational dequantization \cite{pmlr-v97-ho19a}, conditional autoregressive model \cite{meng2021improved} and soft truncation \cite{kim2022soft} alleviate this issue by learning the dequantization noise and finding the optimal $\epsilon$, but incur substantial computational cost and convergence instability.

\paragraph{Remarks}
Our Theorem~\ref{theorem 1} and Proposition~\ref{theorem 2} offer a theoretical perspective on these issues without adding extra network structures. From an information-theoretic standpoint, adding small noise to discrete data corresponds to smoothing data, increasing entropy at a rate controlled by the Fisher information of the data distribution. For highly peaked distributions \cite{kingma2021variational,nielsen2024diffenc}, this smoothing is especially effective, substantially increasing entropy and reducing irregularities, thereby providing a better-conditioned target for model fitting. Proposition~\ref{theorem 2} formally characterizes this entropy increase, while preserving the original KL divergence in the small-noise limit. Table~\ref{tab:nll_train_test} shows the training and evaluation mismatch of ELBO and our methods.

\begin{table}[t]
\centering
\caption{
Comparison of the NLL for ELBO and our bound on CIFAR-10 in training and testing with 0.3 million iterations. 
}
\label{tab:nll_train_test}
\small
\setlength{\tabcolsep}{10pt}
\begin{tabular}{lcccc}
\toprule
\textbf{Model} & \textbf{train} & \textbf{test} \\
\midrule
ELBO        & $2.75 \pm 0.002$ & $2.79 \pm 0.002$ \\
\textbf{Ours} $(\eta_0 = -13.3)$    & $2.79 \pm 0.002$ & $2.80 \pm 0.003$ \\
\textbf{Ours} $(\eta_0 = -8.7)$    & $2.49 \pm 0.002$ & $2.50 \pm 0.003$ \\
\bottomrule
\end{tabular}
\end{table}

Beyond smoothing effects, Theorem~\ref{theorem 1} provides a first-order expansion of the KL divergence with respect to small additive noise. Specifically, when training begins from a nonzero noise level \(\sigma_0^2\), the initial KL objective is reduced by a factor proportional to \(\frac{\sigma_0^2}{2} I(p \| q)\). Since this term is always non-negative and strictly positive when \(p \neq q\), introducing initial noise simplifies the optimization landscape by suppressing fine-scale discrepancies that are otherwise difficult to capture at the outset. Intuitively, the model first learns to match broader statistical structure before refining finer details, stabilizing gradient flow and avoiding early overfitting to discrete artefacts.

Taken together, Theorem~\ref{theorem 1} and Proposition~\ref{theorem 2} clarify why choosing a small but nonzero initial noise level is beneficial. From a signal processing perspective, $\sigma_0^2$ trades off approximation accuracy against numerical tractability. If $\sigma_0^2$ is too large, the training target becomes overly blurred, requiring extra effort to recover fine details. If it is too small, the model must approximate a distribution with sharp discontinuities or high-frequency details from the outset, which can hinder learning and lead to poor convergence.

\paragraph{Truncated Normal Dequantization}

We present a training-free dequantization strategy example, recently adopted by \cite{zheng2023improved}, that naturally fits the diffusion framework and exemplifies our theoretical results. Let $\mathbf{X}_0 \in \{0, ..., 255\}^D$ denote 8-bit discrete data scaled to $[-1,1]$. To define a discrete density, we use a continuous model $q(\cdot;\boldsymbol{\theta})$ evaluated on dequantized inputs:
\begin{equation}
    Q(\mathbf{x}_0;\boldsymbol{\theta}) = \int_{\mathbf{u} \in [-\frac{1}{256},\frac{1}{256}]^D} q(\mathbf{x}_0 + \mathbf{u};\boldsymbol{\theta})\, d\mathbf{u}.
\end{equation}
This matches the diffusion-based formulation if we write $\mathbf{y}_\epsilon = \alpha_\epsilon \mathbf{x}_0 + \sigma_\epsilon \tilde{\boldsymbol{\epsilon}}$ and choose $\tilde{\boldsymbol{\epsilon}} \sim \mathcal{T}\mathcal{N}(0, \mathbf{I}, -\tau, \tau)$ with $\tau = \frac{\alpha_\epsilon}{256\sigma_\epsilon}$. Then $\mathbf{u} = \frac{\sigma_\epsilon}{\alpha_\epsilon} \tilde{\boldsymbol{\epsilon}} \in [-\frac{1}{256}, \frac{1}{256}]^D$ by construction.

Applying a change of variables and accounting for Jacobian terms, the variational lower bound becomes:
\begin{equation}
    \log P_0 (\mathbf{x}_0) \geq \mathbb{E}_{\tilde{\boldsymbol{\epsilon}} \sim \mathcal{T}\mathcal{N}(0, \mathbf{I}, -\tau, \tau)} \left[ \log q\Big(\mathbf{x}+\frac{\sigma_\epsilon}{\alpha_\epsilon}\tilde{\boldsymbol{\epsilon}}\Big) - \log p(\tilde{\boldsymbol{\epsilon}}) \right] + D \log \sigma_\epsilon.
\end{equation}
Using the known entropy of the truncated Gaussian~\cite{zheng2023improved}, the bound further simplifies to:
\begin{align}
    \log P_0 (\mathbf{x}_0) \geq \mathbb{E}_{\hat{\boldsymbol{\epsilon}} \sim \mathcal{T}\mathcal{N}(0, \mathbf{I}, -3, 3)} \left[ \log q(\hat{\mathbf{x}}_\epsilon) \right] + \frac{D}{2} \log (2\pi e \sigma_\epsilon^2) - 0.01522 \times D.
\end{align}

\paragraph{Choosing the Warm-up Noise}\label{app:warm-up}

We further investigate the influence of different \textit{warm-up} noise distributions, Gaussian, Laplace, logistic and Uniform, each scaled to have equal variance. As shown in Table~\ref{tab:nll_fid_compact}, Gaussian noise yields the best performance, closely followed by Laplace and logistic, whereas Uniform significantly underperforms. This result aligns with our theoretical intuition that exponential-family noise distributions, characterized by heavier tails, stabilize training and enhance likelihood estimation. Introducing Laplace noise, previously unexplored in this context, allows us to explicitly examine the impact of heavier-tailed perturbations. 

This observation is consistent with the fact that the Gaussian distribution minimizes Fisher information among all distributions with fixed differential entropy~\cite{1056983}, and simultaneously maximizes differential entropy among all distributions with the same variance. We attribute the performance gap to tail behavior: Uniform noise has compact support and weakly perturbs extreme values, whereas Laplace and Gaussian assign higher probability mass to large deviations. This results in stronger regularization and more stable gradients. In particular, Laplace noise promotes robustness and sparsity, making it effective for high-dimensional or heavy-tailed data.

We consider a baseline that models the data with a mixture of logistic components. Although this parameterization is, in principle, expressive enough to represent a multimodal distribution, in practice the baseline fails to recover all modes. We attribute this gap to optimization/initialization difficulties that arise when the target density exhibits sharp transitions (i.e., a large Lipschitz constant). By contrast, our method is more robust: it captures the distinct modes even with as few as two mixture components. In this sense, we further generalize the framework of \cite{meng2021improved} by smoothing both data and model with a shared isotropic kernel and training via the score matching to stabilize optimization on high-Lipschitz targets.

\clearpage

\begin{figure}
    \centering
    \includegraphics[width=0.65\linewidth]{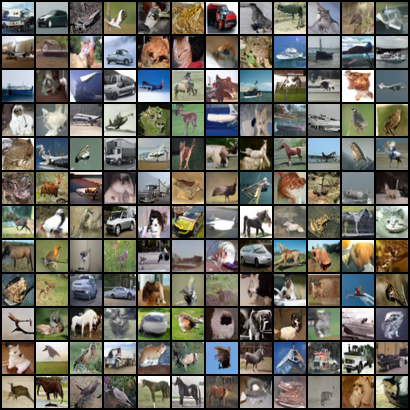}
    \caption{Random samples from our model trained on CIFAR-10 for 300000 parameter updates with EMA. The model was trained in VDM \cite{kingma2021variational} endpoints, and sampled using 1000 sampling timestep.}
    \label{fig:VDMEMA}
\end{figure}

\begin{figure}
    \centering
    \includegraphics[width=0.65\linewidth]{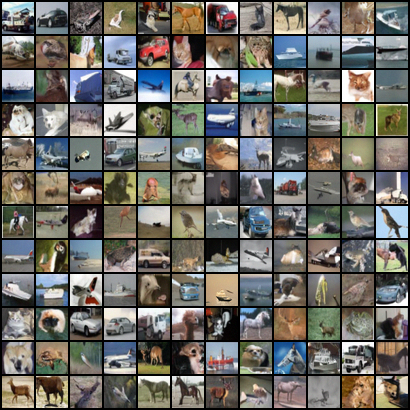}
    \caption{Random samples from our model trained on CIFAR-10 for 300000 parameter updates with EMA. The model was trained with our endpoints, and sampled using 1000 sampling timestep.}
    \label{fig:ISITEMA}
\end{figure}

\begin{figure}
    \centering
    \includegraphics[width=0.65\linewidth]{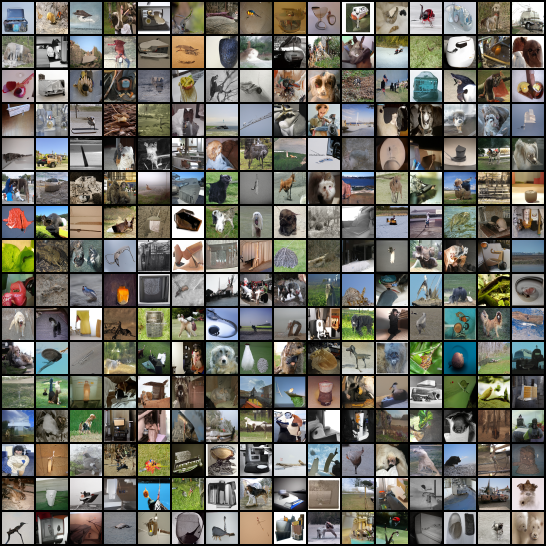}
    \caption{Random samples from our model trained on ImageNet32 for 300000 parameter updates. The model was trained in VDM \cite{kingma2021variational} endpoints, and sampled using 1000 sampling timestep.}
    \label{fig:VDM_IM}
\end{figure}

\begin{figure}
    \centering
    \includegraphics[width=0.65\linewidth]{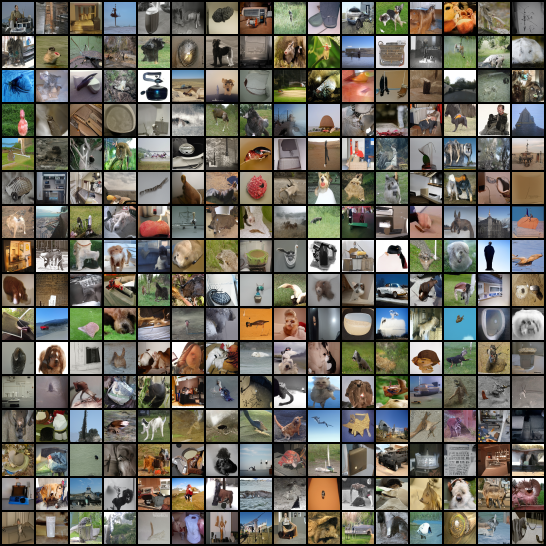}
    \caption{Random samples from our model trained on ImageNet32 for 300000 parameter updates. The model was trained with our endpoints, and sampled using 1000 sampling timestep.}
    \label{fig:ISIT_IM}
\end{figure}

\clearpage

\section*{NeurIPS Paper Checklist}

%%% BEGIN INSTRUCTIONS %%%
The checklist is designed to encourage best practices for responsible machine learning research, addressing issues of reproducibility, transparency, research ethics, and societal impact. Do not remove the checklist: {\bf The papers not including the checklist will be desk rejected.} The checklist should follow the references and follow the (optional) supplemental material.  The checklist does NOT count towards the page
limit. 

Please read the checklist guidelines carefully for information on how to answer these questions. For each question in the checklist:
\begin{itemize}
    \item You should answer \answerYes{}, \answerNo{}, or \answerNA{}.
    \item \answerNA{} means either that the question is Not Applicable for that particular paper or the relevant information is Not Available.
    \item Please provide a short (1–2 sentence) justification right after your answer (even for NA). 
   % \item {\bf The papers not including the checklist will be desk rejected.}
\end{itemize}

{\bf The checklist answers are an integral part of your paper submission.} They are visible to the reviewers, area chairs, senior area chairs, and ethics reviewers. You will be asked to also include it (after eventual revisions) with the final version of your paper, and its final version will be published with the paper.

The reviewers of your paper will be asked to use the checklist as one of the factors in their evaluation. While "\answerYes{}" is generally preferable to "\answerNo{}", it is perfectly acceptable to answer "\answerNo{}" provided a proper justification is given (e.g., "error bars are not reported because it would be too computationally expensive" or "we were unable to find the license for the dataset we used"). In general, answering "\answerNo{}" or "\answerNA{}" is not grounds for rejection. While the questions are phrased in a binary way, we acknowledge that the true answer is often more nuanced, so please just use your best judgment and write a justification to elaborate. All supporting evidence can appear either in the main paper or the supplemental material, provided in appendix. If you answer \answerYes{} to a question, in the justification please point to the section(s) where related material for the question can be found.

IMPORTANT, please:
\begin{itemize}
    \item {\bf Delete this instruction block, but keep the section heading ``NeurIPS Paper Checklist"},
    \item  {\bf Keep the checklist subsection headings, questions/answers and guidelines below.}
    \item {\bf Do not modify the questions and only use the provided macros for your answers}.
\end{itemize}

%%% END INSTRUCTIONS %%%

\begin{enumerate}

\item {\bf Claims}
    \item[] Question: Do the main claims made in the abstract and introduction accurately reflect the paper's contributions and scope?
    \item[] Answer: \answerYes{} % Replace by \answerYes{}, \answerNo{}, or \answerNA{}.
    \item[] Justification: See our introduction for a list of claims including connection with relative entropy with score matching with arbitrary noise.

\item {\bf Limitations}
    \item[] Question: Does the paper discuss the limitations of the work performed by the authors?
    \item[] Answer: \answerYes{} % Replace by \answerYes{}, \answerNo{}, or \answerNA{}.
    \item[] Justification: Yes, our method improves likelihood estimation but does not construct a generative diffusion process under alternative noise. See the paper for more details.

\item {\bf Theory assumptions and proofs}
    \item[] Question: For each theoretical result, does the paper provide the full set of assumptions and a complete (and correct) proof?
    \item[] Answer: \answerYes{} % Replace by \answerYes{}, \answerNo{}, or \answerNA{}.
    \item[] Justification: Please see our detailed proofs.

    \item {\bf Experimental result reproducibility}
    \item[] Question: Does the paper fully disclose all the information needed to reproduce the main experimental results of the paper to the extent that it affects the main claims and/or conclusions of the paper (regardless of whether the code and data are provided or not)?
    \item[] Answer: \answerYes{} % Replace by \answerYes{}, \answerNo{}, or \answerNA{}.
    \item[] Justification: Not only we do show all equations and train on standard datasets, we will open source the code.

\item {\bf Open access to data and code}
    \item[] Question: Does the paper provide open access to the data and code, with sufficient instructions to faithfully reproduce the main experimental results, as described in supplemental material?
    \item[] Answer: \answerNo{}{} % Replace by \answerYes{}, \answerNo{}, or \answerNA{}.
    \item[] Justification: We will open source after paper acceptance.

\item {\bf Experimental setting/details}
    \item[] Question: Does the paper specify all the training and test details (e.g., data splits, hyperparameters, how they were chosen, type of optimizer, etc.) necessary to understand the results?
    \item[] Answer: \answerYes{} % Replace by \answerYes{}, \answerNo{}, or \answerNA{}.
    \item[] Justification: Yes, we include all hyperparameters in the paper and will open source code.

\item {\bf Experiment statistical significance}
    \item[] Question: Does the paper report error bars suitably and correctly defined or other appropriate information about the statistical significance of the experiments?
    \item[] Answer: \answerTODO{} % Replace by \answerYes{}, \answerNo{}, or \answerNA{}.
    \item[] Justification: We will report the deviatations for NLL in Appendix.

\item {\bf Experiments compute resources}
    \item[] Question: For each experiment, does the paper provide sufficient information on the computer resources (type of compute workers, memory, time of execution) needed to reproduce the experiments?
    \item[] Answer: \answerYes{} % Replace by \answerYes{}, \answerNo{}, or \answerNA{}.
    \item[] Justification: We provide this in the paper.
    
\item {\bf Code of ethics}
    \item[] Question: Does the research conducted in the paper conform, in every respect, with the NeurIPS Code of Ethics \url{https://neurips.cc/public/EthicsGuidelines}?
    \item[] Answer: \answerYes{} % Replace by \answerYes{}, \answerNo{}, or \answerNA{}.
    \item[] Justification: Our paper is just a diffusion model useful for density estimation with standard datasets.

\item {\bf Broader impacts}
    \item[] Question: Does the paper discuss both potential positive societal impacts and negative societal impacts of the work performed?
    \item[] Answer: \answerNA{} % Replace by \answerYes{}, \answerNo{}, or \answerNA{}.
    \item[] Justification: \justificationTODO{}
    
\item {\bf Safeguards}
    \item[] Question: Does the paper describe safeguards that have been put in place for responsible release of data or models that have a high risk for misuse (e.g., pretrained language models, image generators, or scraped datasets)?
    \item[] Answer: \answerNA{} % Replace by \answerYes{}, \answerNo{}, or \answerNA{}.
    \item[] Justification: We do not believe our method will have a high risk of abuse as our models are not perceptually SOTA, they only provide for density estimation.

\item {\bf Licenses for existing assets}
    \item[] Question: Are the creators or original owners of assets (e.g., code, data, models), used in the paper, properly credited and are the license and terms of use explicitly mentioned and properly respected?
    \item[] Answer: \answerNA{} % Replace by \answerYes{}, \answerNo{}, or \answerNA{}.
    \item[] Justification: We are using standard benchmark datasets.

\item {\bf New assets}
    \item[] Question: Are new assets introduced in the paper well documented and is the documentation provided alongside the assets?
    \item[] Answer: \answerYes{} % Replace by \answerYes{}, \answerNo{}, or \answerNA{}.
    \item[] Justification: We provide the new official source of ImageNet32.

\item {\bf Crowdsourcing and research with human subjects}
    \item[] Question: For crowdsourcing experiments and research with human subjects, does the paper include the full text of instructions given to participants and screenshots, if applicable, as well as details about compensation (if any)? 
    \item[] Answer: \answerNA{} % Replace by \answerYes{}, \answerNo{}, or \answerNA{}.
    \item[] Justification: 

\item {\bf Institutional review board (IRB) approvals or equivalent for research with human subjects}
    \item[] Question: Does the paper describe potential risks incurred by study participants, whether such risks were disclosed to the subjects, and whether Institutional Review Board (IRB) approvals (or an equivalent approval/review based on the requirements of your country or institution) were obtained?
    \item[] Answer: \answerNA{} % Replace by \answerYes{}, \answerNo{}, or \answerNA{}.
    \item[] Justification: 

\item {\bf Declaration of LLM usage}
    \item[] Question: Does the paper describe the usage of LLMs if it is an important, original, or non-standard component of the core methods in this research? Note that if the LLM is used only for writing, editing, or formatting purposes and does not impact the core methodology, scientific rigorousness, or originality of the research, declaration is not required.
    %this research? 
    \item[] Answer: \answerNA{} % Replace by \answerYes{}, \answerNo{}, or \answerNA{}.
    \item[] Justification: 

\end{enumerate}

\end{document}